\documentclass[aps,pra,groupedaddress,onecolumn]{revtex4-2}

\usepackage{array}[=2016-10-06]

\usepackage{dcolumn}
\usepackage{bm}

\usepackage{graphicx} 
\usepackage{amssymb}
\usepackage{amsmath}
\usepackage{amsthm}
\usepackage{mathtools}
\usepackage{hyperref} 
\usepackage[left=1in,right=1in,]{geometry} 
\usepackage{setspace} 
\usepackage{cleveref}
\usepackage{physics}
\usepackage{forest}
\usepackage{mathrsfs}
\usepackage{multirow} 
\usepackage{enumitem}
\usepackage{tikz}
\usetikzlibrary{quantikz2}
\usepackage{subcaption}
\usepackage[most]{tcolorbox}
\usepackage{xfrac}
\usepackage{thmtools, thm-restate} 
\usepackage{fancyhdr}

\usepackage{float}
\makeatletter
\let\newfloat\newfloat@ltx
\makeatother

\usepackage{algorithm} 
\usepackage{algpseudocode} 

\usepackage{chngcntr}
\counterwithin{figure}{section}

\makeatletter
\newcommand{\vo}{\vec{}\@ifnextchar{^}{\,}{}}
\makeatother

\newtheorem{definition}{Definition}

\newtheorem{remark}{Remark}

\DeclareMathOperator{\Span}{span}

\tcbset{highlight math style={colback=white!10!white, colframe=black!10!black}}

\makeatletter
\newcommand\newtag[2]{#1\def\@currentlabel{#1}\label{#2}}
\makeatother

\fancypagestyle{firstpage}{\fancyfoot[C]{DISTRIBUTION STATEMENT A. Approved for public release: distribution is unlimited.}}

\renewcommand{\subsectionmark}[1]{} 

\hypersetup{
	colorlinks=true,
	linkcolor=blue,      
	filecolor=magenta,      
	urlcolor=cyan,           
}

\usepackage{chngcntr}
\counterwithout{figure}{section}
\counterwithout{table}{section}

\begin{document}

\begin{titlepage}
	\thispagestyle{firstpage}
	\begin{center}
		\Large\textbf{Quantum Data Loading for Carleman Linearized Systems: Application to the Lattice-Boltzmann Equation}\\
		\bigskip
		\large{Reuben Demirdjian\textsuperscript{1,*,$\dagger$}, Thomas Hogancamp\textsuperscript{1,2,*,$\ddagger$},\\ 
		Abeynaya Gnanasekaran\textsuperscript{3}, Amit Surana\textsuperscript{4} and Daniel Gunlycke\textsuperscript{5}}\\
		\bigskip 
		\large\textit{\textsuperscript{1}U.S. Naval Research Laboratory, Monterey, California}\\
		\large\textit{\textsuperscript{2}National Research Council, Washington, D.C.}\\
		\large\textit{\textsuperscript{3}SRI International, Menlo Park, California}\\
		\large\textit{\textsuperscript{4}RTX Technology Research Center (RTRC), East Hartford, Connecticut}\\
		\large\textit{\textsuperscript{5}U.S. Naval Research Laboratory, Washington, D.C.}\\
		\large\textit{\textsuperscript{*}These authors contributed equally to this work.}\\
		\large\textit{\textsuperscript{$\dagger$}\href{mailto:Reuben.Demirdjian.civ@us.navy.mil}{Reuben.Demirdjian.civ@us.navy.mil}}\\
		\large\textit{\textsuperscript{$\ddagger$}\href{mailto:Thomas.E.Hogancamp.ctr@us.navy.mil}{Thomas.E.Hogancamp.ctr@us.navy.mil}}\\
		\today
	\end{center}
	\noindent\textbf{Abstract.} Nonlinear ordinary and partial differential equations are ubiquitous in science and engineering, yet finding their solutions is often computationally intractable for classical hardware. To determine if quantum computers can offer a practical advantage, one critical challenge that must be solved is determining how to efficiently load exponentially sized matrices onto quantum hardware. In this article, we introduce an alternative linear combination of unitaries (LCU) strategy which relies on an intermediate linear combination of non-unitaries (LCNU) and a systematic embedding procedure. One major advantage of this LCU strategy is that it maintains the exact number of terms as in the LCNU. Therefore, this approach offers a data loading framework for matrices that lack an efficient decomposition using the standard LCU alone. Using this approach, we construct a generalized LCNU framework for any Carleman linearized autonomous dynamical system having a polynomial nonlinearity, covering a wide range of problems arising in physics. To demonstrate the effectiveness of our approach, we construct an LCNU for the 3-dimensional Carleman linearized lattice Boltzmann equation (LBE), an important problem in computational fluid dynamics. Here, we find that the number of terms in the decomposition scales like $N_s \sim \mathcal{O}(\alpha^2 Q^2)$, where $\alpha$ is the Carleman truncation order and $Q$ is the number of discrete velocities. Importantly, $N_s$ is completely independent of the number of spatial and temporal discretization points, indicating that the use of exponentially fine meshes may be possible. We then perform a resource estimation of our LCNU's T gate cost when combined with the (1) PREP and SELECT block encoding oracles for fault-tolerant computation, and (2) variational quantum linear solver for NISQ implementation. In the former, the T cost scales like $\mathcal{O}(\alpha^3 Q^2 (\log_2 n)^2)$, where $n$ is the total number of spatial grid points across all dimensions. The latter requires exactly $N_s^2(\log_2 (2n_tn^\alpha)+1)$ circuits per iteration for $n_t$ time steps, with a worst case T gate cost of $\mathcal{O}(\alpha (\log_2 Qn)^2)$ among them. We, therefore, provide not only a generalized data loading framework of wide-reaching applicability within a range of physics disciplines, but also demonstrate the most efficient approach to date for the Carleman linearized LBE. 
\end{titlepage}



\tableofcontents


\newpage 

\section{Introduction}
\twocolumngrid
	
Whether quantum computers can solve nonlinear ordinary or partial differential equations (ODEs or PDEs) with an advantage over classical computers remains an open question. An obvious obstacle being that quantum computers are inherently linear, making it challenging to directly solve nonlinear equations. To circumvent this challenge, the Carleman linearization method \cite{Liu2021,Kowalski1991CarlemanBook} can be used to transform a finite dimensional nonlinear ODE into an infinite dimensional linear ODE, which is then truncated and discretized to obtain a finite dimensional linear system of equations. A solution to the subsequent linear system can then be obtained with a quantum linear system algorithm (QLSA), a problem that quantum computers can solve with an exponential advantage under certain conditions \cite{Childs2017ImprovedHHL,Harrow2009HHL,Bravo2023,Morales2024QLSAsurvey,Dalzell2024QLSAshortcut,Aaronson2015HHLconditions}. 

One necessary condition for any quantum algorithm to maintain advantage, is that the classical data must be loaded onto the quantum computer efficiently. In the case of QLSAs, the classical data is often an exponentially sized non-unitary $2^n \times 2^n$ matrix and, therefore, requires a carefully crafted data loading strategy. If the data loading is too expensive then all quantum advantage is lost before the QLSA even begins. A common criterion for data loading is that its cost should be polylogarithmic with the size of the matrix, that is, its resource requirements scale like $\mathcal{O}(\text{poly} (\log_2 2^n))$. This is likely not possible for dense, unstructured matrices. Fortunately, Carleman linearization admits matrices that are sparse and structured. 

When considering how to load a non-unitary matrix, an obvious first choice is the Pauli decomposition -- a linear combination of unitaries (LCU) where each decomposition term is a Kronecker product of Pauli matrices. However, a common issue with this method is that the number of terms in the LCU of an arbitrary matrix scales like $4^n$ \cite{hantzko2024tensorized}. An alternative consideration is to create a linear combination of non-unitaries (LCNU), whereby the non-unitary operators are embedded within unitary operators resulting in an LCU.  This was initially proposed in Lemma 7 of \cite{Childs2017ImprovedHHL} (therein termed the \textit{non-unitary LCU}), where they provided a quantum algorithm that abstractly prepares the desired state given the LCNU. This was followed up in \cite{GS24,GS2025_LCofThings,Surana2025ConstrainedOpt} where they defined the so-called Sigma basis -- comprised of the identity matrix and the non-unitary $2 \times 2$ standard basis elements -- allowing them to construct explicit quantum circuits. While this LCNU approach is related to the methods of \cite{Li2023SigmaBasis,Bae2024,Liu2021SigmaBasis,Kondo2021SigmaBasis}, the key difference in the former is that the non-unitaries are embedded into unitary matrices resulting in an LCU with an equal number of terms as the LCNU. While this approach does not guarantee an efficient data loading strategy, it can lead to efficient decompositions for specific classes of sparse, structured matrices. Building upon this initial LCNU approach, \cite{Demirdjian2025} defined a new set composed of the Sigma basis and certain permutation matrices to construct an efficient LCNU for the 1-dimensional Carleman linearized Burgers' equation. Despite the inclusion of these permutation matrices, the embedding procedure was shown to be just as straightforward as in the pure Sigma basis approach.

While having an efficient data loading strategy is necessary for advantage, it is not a sufficient condition. In the case of simulating fluid dynamics, a whole suite of challenges must be overcome if any quantum advantage is to be achieved \cite{Demirdjian2022Variational,Huang2025FourierReadout,Nguyen2026QAlgo4PDEReview,Bharadwaj2025,Jennings2026CarlLBEApplied,Lewis2024Limitations,Lin2022Challenges,Sarma2024VariationalPDE,Song2025Incompressible,Wu2025RevisitingCarl,Surana2024PolynomialDynSys,Gourianov2025,Garrett2025Feasibility,Jennings2024FastForwarding,Tennie2023,Gaitan2020,Williams2025Iterative,Oz2021,Jin2024QAlgoPDE,Lapworth2022Hybrid,Ljubomir2022,Lubasch2020,Hogancamp2026Linear,gan2025provably,penuel2024fesibility}. Specific to the Carleman linearization approach, an important challenge is to determine the truncation order required to accurately represent the nonlinear dynamics \cite{Gonzalez2025Carleman,Lin2022Carleman,Jennings2025Carleman,Li2025CarlemanLBE,Liu2021}. Here, a tension exists between the size of the linear system, which grows exponentially with truncation order, and the
Carleman truncation error (i.e. the resulting error induced by truncation). If the truncation order is too large, then advantage may be lost during the data loading step. However, if it is too small, then the Carleman truncation error may dominate the solution. 

An obvious starting point for fluid dynamics simulations is the Navier-Stokes equation. However, an analytical analysis of the Carleman truncation error \cite{Liu2021} found that the error may actually grow with increasing truncation order in the $R>1$ regime, where $R$ is a ratio of nonlinearity to linear dissipation and is related to the Reynolds number. An implication of this convergence issue is that Carleman linearization might only be useful for laminar flow conditions (small Reynolds number) and not the more interesting turbulent flow conditions (large Reynolds number). An alternate approach is to simulate turbulent flows by applying the Carleman linearization procedure to the lattice Boltzmann equation (LBE) \cite{itani2022analysis,sanavio2024lattice,sanavio2024three,succi2026foundational,sanavio2025carleman} since its nonlinearity may be determined not by the Reynolds number, but instead by the Mach number \cite{Li2025CarlemanLBE}. While this might appear to solve the convergence issue, a recent analysis has suggested that the Carleman truncation error for the LBE may actually depend on both the Reynolds and Mach numbers \cite{Jennings2025LBE}. Therefore, it remains to be seen whether the Carleman linearization method can be used to accurately simulate turbulent flow conditions. Though analyses of the errors are insightful, it is important to note that there may exist specific flow regimes for which convergence and accuracy are well-behaved with respect to truncation order. We believe that such flows, if they exist, will likely be determined empirically.

\subsection{Contributions}

In this work, we address the challenge of efficiently loading exponentially sized matrices onto quantum hardware. This is a critical bottleneck in achieving a quantum advantage across a diverse range of application areas, including computational fluid dynamics. Our contributions are structured as follows:

In Section \ref{sec:LinComBlock}, we substantially extend the existing LCNU theory discussed in  \cite{GS24,GS2025_LCofThings,Demirdjian2025,Surana2025ConstrainedOpt,Surana2024PolynomialDynSys} in part by introducing the set $\mathscr{P}_1$, which is composed of the $2 \times 2$ standard basis, Pauli matrices and any permutation matrix with dimension of a power of two. This new class of matrices serves as the fundamental building blocks for constructing non-unitary decompositions, where we demonstrate that every non-unitary element within this class can be directly embedded into a unitary matrix yielding an LCU that preserves the exact term count of the original LCNU. This is accomplished by formalizing the concept of ``unitary completion", i.e. the idea of extending a non-unitary matrix into a unitary one. In Theorem \ref{thm:Theorem for Lbar}, we show that it is straightforward to find the unitary completion for matrices composed of arbitrary combinations of matrix and tensor products of elements from our set $\mathscr{P}_1$. This extension greatly broadens the types of matrices, and consequently the application domains, that can be efficiently loaded with the LCNU method.

In Section \ref{sec:LinComCarl}, we propose a zero padding methodology requiring only a single qubit, an improvement over previous methods \cite{Demirdjian2025,penuel2024fesibility,Jennings2025LBE}, to ensure any resulting Carleman linearized dynamical system has a dimension equal to a power of two, which is a necessary condition for quantum computers. We then derive an efficient LCNU decomposition of this padded system down to its problem specific dynamical matrices, establishing a generalized data loading framework for any Carleman linearized ODE with polynomial nonlinearity.

In Section \ref{sec:Carl_LBE}, we apply this generalized framework to the 3D Carleman linearized LBE, a cornerstone model in computational fluid dynamics and demonstrate the most efficient data loading approach to date. Here, we find that the number of terms in the decomposition scales like $N_s \sim \mathcal{O}(\alpha^2Q^2)$, where $\alpha$ is the Careleman truncation order and $Q$ is the number of discrete velocities in the LBE. We then provide explicit circuit constructions for each term down to the last gate. The result of this is a demonstration of how to implement data loading using our LCNU approach for a useful, real-world problem.

In Section \ref{sec:ResourceEstimation}, we conclude by performing a rigorous resource estimation of our T gate cost considering both a fault-tolerant and noisy intermediate scale quantum (NISQ) framework. In both frameworks, we demonstrate that our approach scales polylogarithmically with the number of spatial and temporal discretization points. Our data loading method is then compared to the Pauli decomposition approach whereby we find that ours not only exhibits a four order of magnitude improvement for even the smallest problem sizes, but also that it has a demonstrably more efficient scaling. This analysis proves that the Carleman linearized LBE can be efficiently loaded onto quantum hardware even for exponentially large spatial grids and time steps, providing a viable path toward obtaining a quantum advantage in fluid dynamics simulations.


\onecolumngrid

\section{A Linear Combination of Non-Unitaries} \label{sec:LinComBlock}
Throughout this paper, $\log$ is used to denote $\log_2$. Consider a square matrix $L \in \mathbb{C}^{N\times N}$ with $N=2^{q_N}$ for some positive integer $q_N$. It is always possible to find an LCU of the form
\begin{equation*}
	L = \sum_{l=1}^{N_A} a_l A_l ,
\end{equation*}
where $a_l \in \mathbb{C}$ and $A_l \in \mathbb{C}^{N\times N}$ are unitary matrices. However, the number of terms $N_A$ in this decomposition is not guaranteed to scale like $\mathcal{O}(\text{poly}(\log N))$, and therefore, may be too large to be practically useful. To overcome this challenge, we seek an LCNU of the form 
\begin{equation} \label{eqn:LCNU}
	L = \sum_{l=1}^{N_s}c_lL_l ,
\end{equation}
where $c_l\in\mathbb{C}$, each $L_l \in \mathbb{C}^{N\times N}$ is either a unitary or specifically designed non-unitary matrix, and the imposed constraint $N_s\sim \mathcal{O}(\text{poly}(\log N))$. In this section, we demonstrate how to embed specific kinds of non-unitary $L_l$ terms into unitary matrices resulting in an LCU with $N_s$ terms. We do this by building off of the approach introduced in \cite{GS24,GS2025_LCofThings,Demirdjian2025}, which requires the following definition.
 
\begin{definition}
	\label{def:completion}
	Henceforth, let $V,W \subset \mathbb{C}^N$ be vector spaces over $\mathbb{C}$. Suppose that $W \subset V$ and that $Q: W \to V$ is a unitary operator on $W$, i.e., for any $w_1, w_2 \in W$, $w_1^\dagger Q^\dagger Qw_2 = w_1^\dagger w_2$. Then, a unitary operator $\overline{Q}: V \to V$ is said to be a unitary completion of $Q$ if $\overline{Q} w = Q w, \text{ for all } w \in W$. Next, let $Q$ be trivial on $W^\perp$, where $W^\perp$ is the orthogonal complement of $W$. Then we also define $Q^\perp \colon W^\perp \to V$, the unitary complement of $Q$, by the relation $Q^{\perp} \coloneq \overline{Q} - Q$. 
\end{definition}

Such a unitary operator $\overline{Q}$ always exists. 

\begin{restatable}[]{lemma}{LemmaUl} \label{Lemma for Ul}
	Using the notation of Definition \ref{def:completion}. Let $Q$, $Q^\perp$ and $\overline{Q}$ represent the matrix forms of their respective operators with respect to the standard basis. Then $\overline{Q}$ is a unitary completion of $Q$ if and only if $U$ is unitary where
	\begin{equation} \label{eqn:U block of Q}
		U=\begin{pmatrix} Q & Q^\perp \\ Q^\perp & Q \end{pmatrix} .
	\end{equation}
\end{restatable}
\begin{proof}
	See Appendix \ref{Proof for U_l}.
\end{proof}

Lemma 1 provides a blueprint to embed each LCNU term from \eqref{eqn:LCNU} into unitary matrices, provided that their unitary complements can be computed. In this case, one can obtain an LCU with $N_s$ terms in the form 
\begin{equation} \label{eqn:LCU}
	U_L = \sum_{l=1}^{N_s} c_l U_l ,
\end{equation}
where each non-unitary term $L_l$ is embedded into $U_l\in\mathbb{C}^{2N \times 2N}$ according to 
\begin{equation} \label{eqn:U_l}
	U_l \coloneq\begin{pmatrix} L_l & L_l^\perp \\ 
		L_l^\perp & L_l \end{pmatrix}
	= U_{l,1} U_{l,2} ,
\end{equation}
and
\begin{equation} \label{eqn:Ul1 Ul2}
	U_{l,1} \coloneq
	\begin{pmatrix}
		I-L_lL_l^\dagger & L_lL_l^\dagger \\
		L_lL_l^\dagger & I-L_lL_l^\dagger
	\end{pmatrix} , 
	\quad
	U_{l,2} \coloneq 
	\begin{pmatrix}
		0 & \overline{L}_l \\
		\overline{L}_l & 0
	\end{pmatrix} .
\end{equation}

In general, unitary completions are not unique and may not have efficient circuit implementations. We therefore seek a construction for each $L_l$ such that their corresponding $U_{l,1}$ and $U_{l,2}$ unitaries have efficient quantum circuit implementations. To that end, we first define $\mathbb{S}_{2^n}=\bigcup_{i=2}^n S_{2^i}$, where $S_{2^i}$ is the set of all $2^i \times 2^i$ permutation matrices. Next, consider the set $\mathscr{P}_1=\mathbb{P}_\rho \cup \mathbb{P}_\sigma \cup \mathbb{S}_{2^n}$, where $\mathbb{P}_\sigma=\{\sigma_x,\sigma_y,\sigma_z,I\}$ is the Pauli basis and $\mathbb{P}_\rho=\{\rho_0,\rho_1,\rho_2,\rho_3\}$ is the standard $2 \times 2$ basis defined as
\begin{equation} \label{eqn:rhos}
	\rho_0 = \begin{pmatrix} 1&0\\0&0 \end{pmatrix}, \quad
	\rho_1 = \begin{pmatrix} 0&1\\0&0 \end{pmatrix}, \quad
	\rho_2 = \begin{pmatrix} 0&0\\1&0 \end{pmatrix}, \quad
	\rho_3 = \begin{pmatrix} 0&0\\0&1 \end{pmatrix}.
\end{equation}

With this, the $L_l$ matrices that we seek for the LCNU decomposition are combinations of matrix and Kronecker products of elements from the set $\mathscr{P}_1$. Specifically, they admit a matrix representation of the form
\begin{equation} \label{eqn:Llgeneral}
	L_l = \prod_{i=1}^{m} \Bigl( \bigotimes_{j=1}^{n_i} P_{i,j} \Bigr) ,
\end{equation}
where $P_{i,j}\in\mathscr{P}_1$ and $\text{dim}(\bigotimes_{j=1}^{n_i} P_{i,j})$ is fixed for all $i\in\{1,\dots,m\}$.

Next, to help construct a completion for \eqref{eqn:Llgeneral}, we adopt the following specific unitary completions for each element in $\mathscr{P}_1$:
\begin{equation} \label{eqn:P bar}
	\overline{P} = 
	\begin{cases} 
		I, & P\in\{\rho_0,\rho_3\} \\
		\sigma_x, & P\in\{\rho_1,\rho_2\} \\
		P, & P\in \mathbb{P}_\sigma \cup \mathbb{S}_{2^n} .
	\end{cases}
\end{equation} 

\begin{remark} \label{remark:Special Pij}
	In the cases where $P_{i,j}\in\mathbb{P}_\rho$ or $P_{i,j}\in\mathbb{P}_\sigma$ then $n_i=n$. In both of these cases, the matrix products in \eqref{eqn:Llgeneral} are redundant because they are each closed under multiplication. Thus, \eqref{eqn:Llgeneral} is intended for cases where the $P_{i,j}$ terms are mixed among $\mathbb{P}_{\rho}$, $\mathbb{P}_\sigma$ and $\mathbb{S}_{2^n}$.
\end{remark}

Any $2^n \times 2^n$ matrix can be constructed using a linear combination of matrix and tensor products of elements from $\mathscr{P}_1$. To see this, note that $\mathbb{P}_\rho$ and $\mathbb{P}_\sigma$ both independently form a basis in $\mathbb{C}^{2 \times 2}$, and therefore the Kronecker products in \eqref{eqn:Llgeneral} enable the formation of a basis over $\mathbb{C}^{2^n \times 2^n}$. 

Next, we show that, given $L_l$ from \eqref{eqn:Llgeneral}, a unitary completion $\overline{L}_l$ is straightforward to find. A construction of $U_{l,2}$ immediately follows since, from \eqref{eqn:Ul1 Ul2}, $U_{l,2}=\sigma_x \otimes \overline{L}_l$.

\begin{restatable}[]{theorem}{TheoremLbar} \label{thm:Theorem for Lbar}
	Consider a non-trivial $L_l$ in the form of \eqref{eqn:Llgeneral}. Then, $\overline{L}_l = \prod_{i=1}^{m} \left( \bigotimes_{j=1}^{n_{i}} \overline{P}_{i,j} \right)$ is a valid unitary completion of $L_l$, where $\overline{P}_{i,j}$ is given by \eqref{eqn:P bar}.
\end{restatable}
\begin{proof}
	See Appendix \ref{Proof for Lbar}.
\end{proof}

\begin{restatable}[]{corollary}{CorUPV} \label{thm:UPV}
	Let $U,V$ be unitary and $P_{i,j} \in \mathscr{P}_1$. Then, 
	\begin{equation} \label{p1}
		L = U \left(\prod_{i=1}^{m} \bigotimes_{j=1}^{n_{i}}P_{i,j}\right) V
	\end{equation}
	can be completed as 
	\begin{equation} \label{p2}
		\overline{L} = U \left(\prod_{i=1}^{m} \bigotimes_{j=1}^{n_{i}}\overline{P}_{i,j}\right) V \, . 
	\end{equation}
\end{restatable}
\begin{proof}
	See Appendix \ref{Proof for UPV}.
\end{proof}

Next, we provide a sufficient condition for $L_l$ under which the associated $U_{l,1}$ is guaranteed to have an efficient circuit representation. To do this, we must first define the set $\mathcal{R} = \{ \lambda \bigotimes_{j=1}^n r_j \;|\; r_j \in \{\rho_0,\rho_3,I\} ,\, \lambda \in \mathbb{C},\, n=1,2,\dots\}$.

\begin{restatable}[]{theorem}{TheoremLLT} \label{thm:Theorem for U1}
	Consider a non-trivial $L_l$ in the form of \eqref{eqn:Llgeneral}. If $L_l L_l^T \in \mathcal{R}$, then the associated $U_{l,1}$ from \eqref{eqn:Ul1 Ul2} can be implemented with a single multi-controlled \textsc{NOT} gate.
\end{restatable}
\begin{proof}
 	This follows directly from the proof of Theorem 2 in \cite{GS2025_LCofThings}. 
\end{proof}

While, Theorem \ref{thm:Theorem for U1} is useful to check if $U_{l,1}$ has an efficient circuit implementation given some $L_l$, it does not aid in the construction of $L_l$ itself. To do this, we provide a corollary showing the existence of an efficient $U_{l,1}$ circuit implementation under specific conditions of $L_l$. First, let $f \in \textup{Sym}(n)$ and define $P_f$ as a permutation matrix that permutes $n$-fold tensor products of $2 \times 2$ matrices, i.e. $P_f (x_1 \otimes \cdots \otimes x_n) P_f^T = x_{f(1)} \otimes \cdots \otimes x_{f(n)}$, where $x_k \in \mathbb{C}^{2 \times 2}$. Then, define $\mathcal{S}_{n!} \coloneq \{P_f \;|\; f \in \textup{Sym}(n)\}$. Following this, we define the set 
\begin{equation} \label{eqn:Set Sn}
	\mathbb{S}_n \coloneq \left\{S \left( \bigotimes_{k=1}^n \sigma_x^{b_k} \right) \;\middle|\; S \in \mathcal{S}_{n!},\, b_k\in\{0,1\},\, n=1,2,\dots \right\} .
\end{equation}
In words, $\mathbb{S}_n$ is the set of permutation matrices that both reorder the sequence of and rotate the components of an $n$-fold tensor product of $2 \times 2$ matrices.

\begin{restatable}[]{corollary}{corollaryLLT} \label{thm:U1_P2}
	Consider a non-trivial $L_l$ in the form of \eqref{eqn:Llgeneral} with the restriction that $P_{i,j} \in \mathscr{P}_2$, where $\mathscr{P}_2 \coloneq \mathbb{P}_\rho \cup \mathbb{P}_\sigma \cup \mathbb{S}_n$. Then the associated $U_{l,1}$ from \eqref{eqn:Ul1 Ul2} is either the identity or can be implemented with a single multi-controlled \textsc{NOT} gate.
\end{restatable}
\begin{proof}
	See Appendix \ref{Proof for U1_P2}.
\end{proof}

It is important to note that the restriction $P_{i,j} \in \mathscr{P}_2$ from Corollary \ref{thm:U1_P2} is a sufficient, but not necessary condition for which $U_{l,1}$ has an efficient implementation. First and foremost, we do not consider the conditions under which $U_{l,1}$ is implemented with $r$ multi-controlled \textsc{NOT} gates for some small integer $r$, which would also be considered efficient. Secondly, there exist permutation matrices outside of $\mathscr{P}_2$ that also admit $U_{l,1}$ matrices that can be implemented with a single multi-controlled \textsc{NOT} gate. For example, consider that $CX \notin \mathscr{P}_2$ and define $L_l = CX (I \otimes \rho_0)$. Then we have $L_lL_l^T = I \otimes \rho_0 \in \mathcal{R}$, which has an efficient $U_{l,1}$ circuit following Theorem \ref{thm:Theorem for U1}. However, if we now consider $L_l = CX (\rho_0 \otimes I)$, we find that $L_lL_l^T =\text{diag}(1,0,0,1) \notin \mathcal{R}$. From these examples, we can see that there exist permutation matrices outside the set from Corollary \ref{thm:U1_P2} that conditionally admit an $U_{l,1}$ circuits composed of a single multi-controlled \textsc{NOT} gate.

We conclude this section by providing algorithms to construct both $U_{l,1}$ and $U_{l,2}$ given an $L_l$ of the form \eqref{eqn:Llgeneral} that satisfies Theorem \ref{thm:Theorem for U1}. Fortunately, all constructions in this article are of this form. With this restriction, $U_{l,1}$ is straightforward to find and is outlined in Algorithm \ref{alg:Ul1}. Similarly, the method to construct a circuit for $U_{l,2}$ is outlined in Algorithm \ref{alg:Ul2}. These two circuits are combined to encode the unitary $U_l$, which, from \eqref{eqn:U_l}, embeds the desired $L_l$ term. 

\begin{algorithm}[h]
	\begin{algorithmic}[1]
		\caption{A pseudo-code to construct the circuit for $U_{l,1}$} \label{alg:Ul1}
		\Require A matrix $L_l \in \mathbb{C}^{N \times N}$ of the form \eqref{eqn:Llgeneral} that satisfies Theorem \ref{thm:Theorem for U1}
		\Require A $\log N+1$ qubit register $q_0,\dots,q_{\log N}$
		\State Compute $L_l L_l^T = \bigotimes_{j=0}^{\log N-1} r_j$ \Comment{$r_j \in \{\rho_0,\rho_3,I\}$}
		\For{$j \gets 0 \text{ to } \log N - 1$} 
		\Comment{Generate control qubits}
		\If{$r_j = \rho_0$}
		\State Add an open control to qubit $q_j$
		\ElsIf{$r_j = \rho_3$}
		\State Add a closed control to qubit $q_j$
		\ElsIf{$r_j = I$}
		\State No control on qubit $q_j$
		\EndIf
		\EndFor	
		\State Target on the ancilla qubit $q_{\log N}$
	\end{algorithmic}
\end{algorithm}

\begin{algorithm}
	\begin{algorithmic}[1]
		\caption{A pseudo-code to construct the circuit for $U_{l,2}$} \label{alg:Ul2}
		\Require A matrix $L_l \in \mathbb{C}^{N \times N}$ of the form \eqref{eqn:Llgeneral} 
		\Require A $\log N+1$ qubit register $q_0,\dots,q_{\log N}$
		\State Compute $\overline{L}_l$ using Theorem \ref{thm:Theorem for Lbar}
		\State Apply the unitary $\overline{L}_l$ onto qubits $q_0,\dots,q_{\log N-1}$
		\State Apply a NOT gate onto qubit $q_{\log N}$
	\end{algorithmic}
\end{algorithm}


\section{Data Loading for the Carleman Linearization Method} \label{sec:LinComCarl}
In this section, we find an LCNU as in \eqref{eqn:LCNU} with terms of the form \eqref{eqn:Llgeneral} for \textit{any} Carleman linearized autonomous dynamical system with polynomial nonlinearity. Note, since a dynamical system with general nonlinearity can be approximated into one with a polynomial nonlinearity in a process referred to as polynomialization \cite{kramer2025polynomialization}, the methods outlined below are applicable more broadly if this extra level of complexity is used. We consider the $N$-dimensional ODE of the form
\begin{equation} \label{eqn:GeneralIVP}
	\frac{\partial \vec{f}}{\partial t} = \sum_{k=0}^{N_F} F_k\vec{f}^{\otimes k}, \qquad \vec{f}(0)=\vec{f}_0 ,
\end{equation}
where $\vec{f}(t) \coloneq \bigl( f_1(t),\dots,f_N(t) \bigr)^T\in\mathbb{C}^N$, $N=2^{q_N}$ for some positive integer $q_N$, $\vec{f}^{\otimes 0} \coloneq 1$, and each $F_k\in\mathbb{C}^{N\times N^k}$ are time independent.

\subsection{A Review of the Carleman Method} \label{sec:Carl Lin Review}
To transform \eqref{eqn:GeneralIVP} using the Carleman linearization method following \cite{Kowalski1991CarlemanBook,Liu2021}, let $\alpha \ge N_F + 1$ be the truncation order and define $\vec{y}(t) \coloneq (\vec{f}(t), \vec{f}^{\otimes2}(t), \dots, \vec{f}^{\otimes\alpha}(t))^T \in \mathbb{C}^{\Delta}$ with $\Delta=\sum_{j=1}^\alpha N^j$. This yields 
\begin{equation} \label{eqn:dydt}
	\frac{d\vec{y}}{dt} = A\vec{y} + \vec{b} , \quad \vec{y}(0)=\vec{y}^{\,0} ,
\end{equation}
where 
\begin{equation} \label{eqn:A}
	A \coloneq                                   
	\begin{pmatrix}
		A_1^1    & \cdots & A_{N_F}^1    &       		&  					   &  \\
		A_1^2    & A_2^2  & \cdots 	     & A_{N_F+1}^2 	&             		   &  \\    
		&  & 		 &     		    & \ddots      		   &  \\
		&  	  	 &   \ddots     		 & 	\ddots            &    	     		   & A_{\alpha}^{\alpha-N_F+1}  \\
		&	  	  & 	   	  	 & 			 	&             		   & \vdots \\
		&   	  &        		 & 			 	& A_{\alpha-1}^\alpha  & A_\alpha^\alpha \\
	\end{pmatrix} , 
	\quad 
	\vec{b} \coloneq \begin{pmatrix} F_0 \\ 0 \\ \vdots \\ 0 \end{pmatrix} ,
\end{equation}
with $A\in\mathbb{C}^{\Delta\times\Delta}$, $\vec{b}\in\mathbb{C}^{\Delta}$ and $A_{j+k-1}^j \in \mathbb{C}^{N^j \times N^{j+k-1}}$ for $k = 0,1,\dots,N_F$ defined by
\begin{equation} \label{eqn:Ajpkm1j}
	A_{j+k-1}^j \coloneq \sum_{l=0}^{j-1} I_{N^l} \otimes F_k \otimes I_{N^{j-l-1}} .
\end{equation}
Next, we temporally discretize \eqref{eqn:dydt} using $n_t=2^{q_t}$ time steps of size $\Delta t$ for some positive integer $q_t$ by applying the backward Euler method. This results in the linear system 
\begin{equation} \label{eqn:LYB}
	L\vec{Y}=\vec{B} ,
\end{equation}
where $L\in\mathbb{C}^{n_t\Delta\times n_t\Delta}$, $\vec{Y},\vec{B}\in\mathbb{C}^{n_t\Delta}$ and
\begin{equation} \label{eqn:L Matrix}
	L \coloneq
	\begin{pmatrix}
		I  &  			 & 	 	  &  \\
		-I & I-\Delta tA &  	  &  \\
		& \ddots 	 & \ddots &  \\
		&      		 & -I 	  & I-\Delta tA
	\end{pmatrix} , \quad
	\vec{Y} \coloneq
	\begin{pmatrix}
		\vec{y}^{\,0} \\
		\vec{y}^{\,1} \\
		\vdots \\
		\vec{y}^{\,n_t-1}
	\end{pmatrix} , \quad
	\vec{B} \coloneq 
	\begin{pmatrix}
		\vec{y}^{\,0} \\
		\vec{b}\Delta t \\
		\vdots \\
		\vec{b}\Delta t
	\end{pmatrix} ,
\end{equation}
for $\vec{y}^{\,r}=\vec{y}(r\Delta t)$. Thus, $L\vec{Y} = \vec{B}$ is a linear system of equations that approximates the original nonlinear ODE in \eqref{eqn:GeneralIVP} with accuracy depending on the truncation order $\alpha$. One disadvantage of \eqref{eqn:LYB} is that it is challenging to efficiently decompose $L$ into an LCU since each of the submatrices have different size. In the next subsection, we provide an approach to solve this problem.

\subsection{An LCNU for Carleman Linearized Systems} \label{sec:Carl Lin Zero Pad}
In this subsection, we implement a judicious zero padding approach in a similar (but improved) fashion as in \cite{Demirdjian2025}. This zero padding is advantageous because it lays the foundation for an LCNU decomposition with a polylogarithmic number of terms. While a general decomposition is not possible since each dynamical system has unique structure, we do outline a general LCNU to the fullest extent possible for an arbitrary Carleman linearized system. From this, an LCU is formed by using the methods of Section \ref{sec:LinComBlock} to embed each term of the LCNU.

First, define $\vec{y}^{\,(\text{e})}(t) \coloneq (\vec{z}(t),\vec{y}(t))^T \in \mathbb{C}^{2N^\alpha}$, $\vec{b}^{(\text{e})} \coloneq (\vec{z}(0),\vec{b})^T \in \mathbb{C}^{2N^\alpha}$, $\vec{z}(t) \in \mathbb{C}^{N_0}$, $N_0 = 2N^\alpha-\Delta$ and
\begin{equation} \label{eqn:Ae_def}
	A^{(\text{e})} \coloneq 
	\begin{pmatrix}
		0_{N_0 \times N_0} & 0_{N_0 \times \Delta} \\ 0_{\Delta \times N_0} & A
	\end{pmatrix}
	\in \mathbb{C}^{2N^\alpha \times 2N^\alpha} .
\end{equation}
We can then embed the ODE from \eqref{eqn:dydt} into the larger ODE
\begin{equation} \label{eqn:dyedt}
	\frac{d\vec{y}^{\,(\text{e})}}{dt} = A^{(\text{e})}\vec{y}^{\,(\text{e})} + \vec{b}^{\,(\text{e})} , \quad \vec{y}^{\,(\text{e})}(0)=\vec{y}^{\,(\text{e}),0} .
\end{equation}
Note, the system is formulated such that $\vec{z}(t)=\vec{z}(0)$, where we can let $\vec{z}(0)=\vec{0}$. As before, we apply the backward Euler method to \eqref{eqn:dyedt} to obtain the linear system  
\begin{equation} \label{eqn:LeYeBe}
	L^{(\text{e})}\vec{Y}^{(\text{e})}=\vec{B}^{(\text{e})} ,
\end{equation}
where $L^{(\text{e})}\in\mathbb{C}^{2n_tN^\alpha \times 2n_tN^\alpha}$, and $\vec{Y}^{(\text{e})},\vec{B}^{(\text{e})}\in\mathbb{C}^{2n_tN^\alpha}$ with
\begin{equation} \label{eqn:Le Matrix}
	L^{(\text{e})} \coloneq
	\begin{pmatrix}
		I  &   		           &        &   \\
		-I & I-\Delta tA^{(\text{e})} &        &  \\
		   & \ddots    		   & \ddots & \\
		   &  				   & -I 	& I-\Delta tA^{(\text{e})}
	\end{pmatrix} ,\quad
	\vec{Y}^{(\text{e})} \coloneq
	\begin{pmatrix} \vec{y}^{\,(\text{e}),0} \\ \vec{y}^{\,(\text{e}),1} \\ \vdots \\ \vec{y}^{\,(\text{e}),n_t-1} \end{pmatrix}
	,\quad
	\vec{B}^{(\text{e})} \coloneq
	\begin{pmatrix} \vec{y}^{\,(\text{e}),0} \\ \vec{b}^{(\text{e})}\Delta t \\ \vdots \\ \vec{b}^{(\text{e})}\Delta t \end{pmatrix} ,
\end{equation}
and $\vec{y}^{\,(\text{e}),r} \coloneq \vec{y}^{\,(\text{e})}(r\Delta t)$. The resulting system in \eqref{eqn:LeYeBe} embeds the original system from \eqref{eqn:LYB} and can therefore provide the desired solution.

To see how this new matrix $L^{(\text{e})}$ can be decomposed into an efficient LCNU, let
\begin{equation} \label{eqn:Le}
	L^{(\text{e})} = L_1^{(\text{e})} - \Delta tL_2^{(\text{e})}
\end{equation}
where $L_1^{(\text{e})}  \in \mathbb{C}^{2n_tN^\alpha \times 2n_tN^\alpha}$ is defined by
\begin{equation} \label{eqn:L1e}
	L_1^{(\text{e})} \coloneq 
	\begin{pmatrix}
		I_{2N^\alpha}  &               & 				&  \\
		-I_{2N^\alpha} & I_{2N^\alpha} & 				&  \\
		 			   & \ddots 	   & \ddots 		& \\
					   & 			   & -I_{2N^\alpha} & I_{2N^\alpha}        
	\end{pmatrix} 
	=
	(I_{n_t} + \rho_1^{\otimes \log n_t} - S_{+1}^{n_t} ) \otimes I_{2N^\alpha} .
\end{equation}
Here,  $I_r$ is the $r \times r$ identity matrix and $S_{+1}^{n_t}\in\mathbb{R}^{n_t \times n_t}$ is the incrementer permutation matrix defined by
\begin{equation} \label{eqn:Sp1}
	S_{+1}^r \coloneq 
	\left( \begin{array}{cccc}
		0      & \cdots & 0 & 1 \\
		1      & \ddots &   & 0 \\
		& \ddots	&  	& \vdots \\
		&  		& 1 & 0 \\
	\end{array} \right)_{r \times r} .
\end{equation}
Additionally, $L_2^{(\text{e})}  \in \mathbb{C}^{2n_tN^\alpha \times 2n_tN^\alpha}$ is defined by 
\begin{equation} \label{eqn:L2e}
	L_2^{(\text{e})} \coloneq 
	\begin{pmatrix}
		0  & 0 & \cdots & 0 \\
		0 & A^{(\text{e})} & \cdots & 0 \\
		\vdots & \ddots & \ddots & \vdots \\
		0  & \cdots & 0 & A^{(\text{e})}
	\end{pmatrix}
	=
	(I_{n_t} - \rho_0^{\otimes \log n_t})\otimes A^{(\text{e})} ,
\end{equation}
where, from the derivation from Appendix \ref{sec:Ae deriv}, we have
\begin{equation} \label{eqn:Ae_expr}
	\begin{split}
		A^{(\text{e})} 
		&= \underbrace{
			\sum_{j=2}^{\alpha} (\rho_0 \otimes \rho_3^{\otimes \log N-1})^{\otimes \alpha-j} 
			\otimes \rho_2 \otimes \left[ A^{(\text{e}),j}_{j-1} \left( \mathcal{P}_2^T \otimes I_{N^{j-1}} \right) \right]
			}_\text{Constant Forcing} \\		
		&+ \underbrace{
			\sum_{j=1}^{\alpha} (\rho_0 \otimes \rho_3^{\otimes \log N-1})^{\otimes \alpha-j} 
			\otimes \rho_3 \otimes A^j_j
			}_\text{Linear} \\  
		&+ \underbrace{
			\sum_{k=2}^{N_F} \sum_{j=1}^{\alpha-k+1} (\rho_0 \otimes \rho_3^{\otimes \log N-1})^{\otimes \alpha-k-j+1} 
			\otimes \rho_{1} \otimes \left[ \left( \mathcal{P}_k\otimes I_{N^j} \right) A^{(\text{e}),j}_{j+k-1} \right] 
			}_\text{Nonlinear} ,
	\end{split}
\end{equation} 
where $\mathcal{P}_k\in\mathbb{R}^{N^{k-1} \times N^{k-1}}$ is an easy to construct permutation matrix described in Appendix \ref{sec:Pk circuit} and  
\begin{subequations}
	\label{equations}
	\begin{align}
		\label{eqn:Aejm1j_def}
		A^{(\text{e}),j}_{j-1} &\coloneq \begin{pmatrix} A^j_{j-1} & 0_{N^j \times (N^{j}-N^{j-1})} \end{pmatrix} 
		\in \mathbb{C}^{N^j \times N^j}
		, \\[4pt]
		\label{eqn:Aejpkm1j_def}
		A^{(\text{e}),j}_{j+k-1} &\coloneq
		\begin{pmatrix} A^j_{j+k-1} \\ 0_{(N^{j+k-1}-N^j) \times N^{j+k-1}} \end{pmatrix} 
		\in \mathbb{C}^{N^{j+k-1} \times N^{j+k-1}} ,\qquad k\ge2 .
	\end{align} 
\end{subequations}
In \eqref{eqn:Ae_expr}, the ``Constant Forcing", ``Linear" and ``Nonlinear" terms are associated with the $F_0$, $F_1$, and $F_k\vert_{k\ge2}$ terms respectively. Using the derivation in Appendix \ref{Aejpkm1j_Derivation}, an expression to decompose both \labelcref{eqn:Aejm1j_def,eqn:Aejpkm1j_def} is
\begin{equation} \label{eqn:Aejpkm1j}
	A^{(\text{e}),j}_{j+k-1} =
	\sum_{l=0}^{j-1}
	\left[
	\left( \rho_0^{\otimes \log N^{k-1}} \otimes K^{(N^l,N)} \right) 
	\left(
	F^{(\text{e})}_k 
	\otimes I_{N^l} 
	\right)
	K^{(N^k,N^l)} \right]
	\otimes I_{N^{j-l-1}} , 
\end{equation}
where $K^{(a,b)}\in\mathbb{C}^{ab\times ab}$ is the commutation matrix \cite{Xu2018} defined in Appendix \ref{sec:Circ for Commutation} and 
\begin{equation} \label{eqn:Fke def}
	F^{(\text{e})}_k \coloneq
	\begin{cases}
		\begin{pmatrix} F_0 & 0_{N \times N-1} \end{pmatrix} , & k=0 \\[4pt]
		\begin{pmatrix} F_k \\ 0_{(N^k-N) \times N^k} \end{pmatrix} , & k\ge2 .
	\end{cases}
\end{equation}

Together, equations \labelcref{eqn:L2e,eqn:Ae_expr,eqn:Aejpkm1j,eqn:Fke def} provide a decomposition for $L_2^{(\text{e})}$ of the form \eqref{eqn:Llgeneral} up to $F_0^{(\text{e})}$, $F_1$ and $F^{(\text{e})}_k$, which are problem dependent matrices that must be specified before further decomposition is possible. The framework provided here admits an efficient number of terms $N_s$ in the LCNU provided that efficient decompositions for the problem dependent $F_0^{(\text{e})}$, $F_1$ and $F^{(\text{e})}_k$ matrices exists. In the next section, we look at an explicit dynamical system that admits a polylogarithmic decomposition.


\section{Data Loading for the Carleman Linearized LBE} \label{sec:Carl_LBE}
We begin this section by using the discrete-velocity Boltzmann equation (DVBE) to derive the lattice Boltzmann equation (LBE) in the form of \eqref{eqn:GeneralIVP}. This form allows us to apply the zero padded Carleman linearization approach from Section \ref{sec:Carl Lin Zero Pad}, yielding a linear system that approximates the LBE. Next, we decompose the linear system using the framework described in Section \ref{sec:Carl Lin Zero Pad} including the problem specific $F_k$ matrices. Lastly, we create explicit circuit constructions for each term in the LCNU decomposition, providing a complete method to encode the zero padded Carleman linearized LBE.


\subsection{Derivation of the LBE} \label{sec:LBE Derivation}

The DVBE is obtained from the classical Boltzmann equation via a suitable discretization of velocity space. In particular, if a velocity set $\{(e_m, w_m)\}_{m=1}^{Q}$, for velocity $\vec{e}_m\in\mathbb{C}^3$ and weight $w_m$, satisfies a number of algebraic relations to ensure conservation of macroscopic variables (see Section 3.4.7 of \cite{Kruger2017}), then the DVBE can be used to model the continuous dynamics of $f(t,\vec{X}) = (f_1(t,\vec{X}), ..., f_Q(t,\vec{X}))$, where $f_m(t,\vec{X})$ is the velocity distribution function describing the probability density of finding a particle at position $\vec{X}=(x,y,z)$ and time $t$. We consider the DVBE on the domain $(0,L_x) \times (0,L_y) \times (0,L_z)$ with a BGK collision operator \cite{BGK1954} and periodic boundary conditions in each spatial direction. The relevant system is:
\begin{gather} \label{eqn:DVBE}
	\frac{\partial f_m}{\partial t} + \vec{e}_m\cdot\nabla f_m = - \frac{1}{\tau}(f_m-f_m^\text{(eq)}) 
	, \quad \vec{f}_m(t=0,\vec{X}) = \vec{f}_{m,0}(\vec{X}) ,
\end{gather}
where $f_m^\text{(eq)}$ is the equilibrium distribution function and $\tau$ is the relaxation parameter. For this work, we assume that $m\in\{1,\dots,Q\}$ with $Q=2^{q_Q}$ for some positive integer $q_Q$, i.e. that the number of discrete velocities are a power of two. 

Before deriving the LBE, define $\Delta \eta \coloneq L_\eta/n_\eta$, where $\eta \in \{x,y,z\}$ and $n_\eta$ is the number of discretization points for dimension $\eta$. Assume that $\Delta x = \Delta y = \Delta z$, then we non-dimensionalize the DVBE using the following parameters:
\begin{equation} \label{eqn:nondim}
	\frac{\partial}{\partial t} = \frac{1}{\Delta t}\frac{\partial}{\partial t^\star} ,\quad 
	\frac{\partial}{\partial \eta} = \frac{1}{\Delta x}\frac{\partial}{\partial \eta^\star} ,\quad
	\vec{e}_m = \frac{\Delta x}{\Delta t} \vec{e}_m^{\,\star} ,\quad
	\tau = \Delta t \tau^\star = \frac{\Delta t^2}{\Delta x^2} \frac{\nu}{c_s^{\star 2}} ,\quad
	c_s = \frac{\Delta x}{\Delta t} c_s^\star ,
\end{equation}
where $c_s$ is the speed of sound, $\nu$ is the kinematic viscosity, $e_m \in \mathbb{Z}$, and the $\star$ superscript represents a non-dimensional variable. Next, by inserting the rescalings from \eqref{eqn:nondim} into \eqref{eqn:DVBE}, we obtain 
\begin{equation} \label{eqn:FullLBE NonDim}
	\frac{\partial f_m}{\partial t^\star} + \vec{e}_m^{\,\star} \cdot \nabla^\star f_m 
	= - \frac{1}{\tau^\star}(f_m-f_m^\text{(eq)}) ,
\end{equation}
where $\nabla^\star = \partial/\partial x^\star + \partial/\partial y^\star + \partial/\partial z^\star$. 

Next, we expand the Maxwell equilibrium function using a second order Taylor expansion of the form
\begin{equation} \label{eqn:MaxEq}
	f_m^\text{(eq)} 
	= \rho^\star w_m^\star(a + b\vec{e}_m^{\,\star} \cdot \vec{u}^{\,\star} + c(\vec{e}_m^{\,\star} \cdot \vec{u}^{\,\star})^2 + d\norm{\vec{u}^{\,\star}}^2) ,
\end{equation}
where $a=1$, $ b=c_s^{\star-2}$, $c=(2c_s^{\star 4})^{-1}$, $d=-(2c_s^{\star 2})^{-1}$, and $w_m^\star$ are the lattice weights for the non-dimensional speeds $\vec{e}_m^{\,\star}$. The non-dimensionalized macroscopic variables are defined by
\begin{equation} \label{eqn:rho_u}
	\rho^\star(t,\vec{X}) \coloneq \sum_{m=1}^Q f_m ,\quad \vec{u}^{\,\star}(t,\vec{X}) \coloneq \frac{1}{\rho^\star} \sum_{m=1}^Q f_m \vec{e}_m^{\,\star} .
\end{equation}

Following \cite{Li2025}, assume that $\rho^\star\approx1$, then $1/\rho^\star \approx2-\rho^\star$ and $\abs{1-\rho^\star}\ll1$. Evaluating \labelcref{eqn:MaxEq,eqn:rho_u} into $\eqref{eqn:FullLBE NonDim}$ and, henceforth dropping the $\star$ superscripts for convenience, yields
\begin{equation} \label{eqn:LBE Expanded}
	\begin{split}
		\frac{\partial f_m}{\partial t} &= -\vec{e}_m\cdot\nabla f_m - \frac{1}{\tau}f_m
		+ \frac{1}{\tau}w_m\biggr( a\sum_{q=1}^Q f_q + b\vec{e}_m\cdot\sum_{q=1}^Q \vec{e}_q f_q \biggl) \\
		&+ \frac{2}{\tau}w_m\biggr( c\Bigr(\vec{e}_m\cdot\sum_{q=1}^Q \vec{e}_q f_q\Bigl)^2 + d\norm{\sum_{q=1}^Q \vec{e}_q f_q} ^2 \biggl)
		- \frac{1}{\tau}w_m\biggr( c\Bigr(\vec{e}_m\cdot\sum_{q=1}^Q \vec{e}_q f_q\Bigl)^2 + d\norm{\sum_{q=1}^Q \vec{e}_q f_q} ^2 \biggl)\biggr(\sum_{s=1}^Q f_s\biggl) .
	\end{split}
\end{equation}
To simplify \eqref{eqn:LBE Expanded}, we use the following relations
\begin{equation} \label{eqn:LBE properties}
	\begin{split}
		\norm{\sum_{q=1}^Q f_q\vec{e}_q} ^2
		&= \sum_{q,r=1}^Q (\vec{e}_q \cdot \vec{e}_r) f_q f_r , \\
		\Bigr(\vec{e}_m\cdot\sum_{q=1}^Q \vec{e}_q f_q\Bigl)^2 
		&= \sum_{q,r=1}^Q (\vec{e}_m\cdot\vec{e}_q)(\vec{e}_m\cdot\vec{e}_r)f_qf_r ,\\
		\Bigr(\vec{e}_m\cdot\sum_{q=1}^Q \vec{e}_q f_q\Bigl)^2 \Bigr(\sum_{s=1}^Q f_s\Bigl) 
		&= \sum_{q,r,s=1}^Q (\vec{e}_m\cdot\vec{e}_q)(\vec{e}_m\cdot\vec{e}_r)f_qf_rf_s ,
	\end{split}
\end{equation}
and define
\begin{equation} \label{eqn:gamma_beta}
	\begin{split}
		\beta_{m,q} &\coloneq \frac{1}{\tau}\Bigr( w_m\big(a + b(\vec{e}_m\cdot\vec{e_q})\big) - \delta_{m,q} \Bigl) ,\\
		\gamma_{q,m,r} &\coloneq \frac{1}{\tau}w_m\big( c(\vec{e}_m\cdot\vec{e}_q)(\vec{e}_m\cdot\vec{e}_r) + d(\vec{e}_q\cdot\vec{e}_r) \big) .
	\end{split}
\end{equation}
Inserting \labelcref{eqn:LBE properties,eqn:gamma_beta} into \eqref{eqn:LBE Expanded} yields
\begin{equation} \label{eqn:LBE_Compact}
	\frac{\partial f_m}{\partial t} = -\underbrace{\vec{e}_m\cdot\nabla f_m}_\text{Streaming}
	+ \underbrace{\sum_{q=1}^Q\beta_{m,q}f_q}_\text{Linear Collision}
	+ \underbrace{2\sum_{q,r=1}^Q \gamma_{q,m,r}f_qf_r}_\text{Quadratic Collision}
	- \underbrace{\sum_{q,r,s=1}^Q \gamma_{r,m,s}f_qf_rf_s}_\text{Cubic Collision} .
\end{equation}

Next, we spatially discretize the non-dimensional domain using lattice units, i.e. $\Delta x^\star = \Delta y^\star = \Delta z^\star = 1$, with the following grids: $\vec{x}=(1,\dots,n_x)^T$, $\vec{y}=(1,\dots,n_y)^T$, and $\vec{z}=(1,\dots,n_z)^T$. Then, the non-dimensional domain is $[0,n_x) \times [0,n_y) \times [0,n_z)$, where we assume that $n_x$, $n_y$ and $n_z$ are each a power of two. Next, define $f_m^{kji} \coloneq f_m(t,z_k,y_j,x_i)$, then the discretized velocity distribution function may be written as
\begin{equation} \label{eqn:fvec}
	\begin{split}
		\vec{f} = (&\underbrace{\underbrace{f_1^{111},\dots,f_Q^{111},
				\dots,f_1^{11n_x}, \dots,f_Q^{11n_x}}_{Qn_x},
			\dots,f_1^{1n_yn_x},\dots,f_Q^{1n_yn_x}}_{Qn_xn_y}, \\
		&\vdots \\
		&\underbrace{\underbrace{f_1^{n_z11},\dots,f_Q^{n_z11},
				\dots,f_1^{n_z1n_x},\dots,f_Q^{n_z1n_x}}_{Qn_x},
			\dots,f_1^{n_zn_yn_x},\dots,f_Q^{n_zn_yn_x}}_{Qn_xn_y})^T ,
	\end{split} 
\end{equation}
where $\vec{f}\in\mathbb{C}^{Qn}$ and $n=n_xn_yn_z$. Using this discretization, our goal is to obtain a vectorized version of \eqref{eqn:LBE_Compact} in the form of \eqref{eqn:GeneralIVP} with $N_F=3$ (cubic nonlinearity). To that end, we must determine the structure of $F_1 \in \mathbb{C}^{Qn\times Qn}$, $F_2 \in \mathbb{C}^{Qn\times (Qn)^2}$ and $F_3 \in \mathbb{C}^{Qn\times (Qn)^3}$. First, observe that there are two linear terms in \eqref{eqn:LBE_Compact}: the streaming and linear collision term. We, can therefore write $F_1 = S + \tilde{F}_1$, where $S$ is the streaming operator obtained via the central finite difference method and $\tilde{F}_1$ is the linear collision operator. In the remainder of this section, we find the explicit forms of $S$, $\tilde{F}_1$, $F_2$ and $F_3$, which we then use to construct a vectorized version from \eqref{eqn:LBE_Compact}. 

From \eqref{eqn:LBE_Compact} and the assumption that $\Delta x^* = 1$, the streaming operator $S \in \mathbb{C}^{Qn \times Qn}$ must satisfy
\begin{equation*}
	\big( S \vec{f} \,\big) \big|_m^{kji} = -\frac{1}{2}
	\left( e_m^x(f_m^{k,j,i+1} - f_m^{k,j,i-1}) + e_m^y(f_m^{k,j+1,i} - f_m^{k,j-1,i}) + e_m^z(f_m^{k+1,j,i} - f_m^{k-1,j,i}) \right) ,
\end{equation*}
for $i \in \{1,\dots,n_x\}$, $j \in \{1,\dots,n_y\}$, $k \in \{1,\dots,n_z\}$ and $m \in \{1,\dots,Q\}$. To determine a matrix representation of $S$, first define the velocity matrices
\begin{equation} \label{eqn:ExEyEz}
	E_x \coloneq \begin{pmatrix} e_1^x & & \\ & \ddots & \\ & & e_Q^x \end{pmatrix} ,\quad
	E_y \coloneq \begin{pmatrix} e_1^y & & \\ & \ddots & \\ & & e_Q^y \end{pmatrix} ,\quad
	E_z \coloneq \begin{pmatrix} e_1^z & & \\ & \ddots & \\ & & e_Q^z \end{pmatrix} ,
\end{equation}
where $\vec{e}_m=(e_m^x,e_m^y,e_m^z)$ is the $m\text{th}$ velocity vector. Next, we define 
\begin{equation} \label{eqn:SxSySz}
\begin{split}
	S_x &\coloneq \frac{1}{2} I_{n_yn_z} \otimes (S_{+1}^{n_x} - S_{-1}^{n_x}) \otimes E_x ,\\
	S_y &\coloneq \frac{1}{2} I_{n_z} \otimes (S_{+1}^{n_y} - S_{-1}^{n_y}) \otimes I_{n_x} \otimes E_y ,\\
	S_z &\coloneq \frac{1}{2} (S_{+1}^{n_z} - S_{-1}^{n_z}) \otimes I_{n_xn_y} \otimes E_z .
\end{split}
\end{equation}
Then, it follows that $S = S_x + S_y + S_z$.

Next, from \eqref{eqn:LBE_Compact}, the linear collision operator $\tilde{F}_1 \in \mathbb{C}^{Qn \times Qn}$ must satisfy 
\begin{equation}
	\big( \tilde{F}_1 \vec{f} \,\big) \big|_m^{kji} = \sum_{q=1}^Q \beta_{m,q} f_{q}^{kji} .
\end{equation}
We therefore have 
\begin{equation} \label{eqn:F1tilde}
	\tilde{F}_1 \coloneq\begin{pmatrix} R & & \\ & \ddots & \\ & & R \end{pmatrix}
	,\quad 	
	R \coloneq 
	\begin{pmatrix} 
		\beta_{1,1} & \cdots & \beta_{1,Q} \\ 
		\vdots & \ddots & \vdots \\ 
		\beta_{Q,1} & \cdots & \beta_{Q,Q}
	\end{pmatrix}
	\in\mathbb{C}^{Q\times Q} .
\end{equation}

Again using \eqref{eqn:LBE_Compact}, the quadratic collision operator must satisfy 
\begin{equation}
	\big( F_2 \vec{f}^{\,\otimes 2} \,\big) \bigr|_{m}^{kji}
	= 2\sum_{q,r=1}^Q \gamma_{q,m,r}f_q^{kji}f_r^{kji} .
\end{equation}
If we define $B_{2,q}\in\mathbb{C}^{n\times Qn^2}$ as
\begin{equation} \label{eqn:Bq}
	B_{2,q} \coloneq
	\begin{pmatrix}
		\underbrace{0\;\dots\;0}_{n(q-1)} \; \underbrace{1\;0\;\dots\;0}_{n} \; \underbrace{0\;\dots\;0}_{n(Q-q)} \\
		& \underbrace{0\;\dots\;0}_{n(q-1)} \; \underbrace{0\;1\;0\;\dots\;0}_{n} \; \underbrace{0\;\dots\;0}_{n(Q-q)} \\
		&& \ddots \\
		&&& \underbrace{0\;\dots\;0}_{n(q-1)} \; \underbrace{0\;\dots\;0\;1}_{n} \; \underbrace{0\;\dots\;0}_{n(Q-q)}
	\end{pmatrix} ,
\end{equation}
and define $\Gamma_q\in\mathbb{C}^{Q\times Q}$ as
\begin{equation} \label{eqn:Gammaq}
	\Gamma_q \coloneq
	\begin{pmatrix}
		\gamma_{q11} & \cdots & \gamma_{q1Q} \\
		\vdots & \ddots & \vdots \\
		\gamma_{qQ1} & \cdots & \gamma_{qQQ}
	\end{pmatrix} ,
\end{equation}
then it follows that
\begin{equation} \label{eqn:F2}
	F_2 = 2\sum_{q=1}^Q B_{2,q} \otimes \Gamma_q .
\end{equation}

Lastly, using \eqref{eqn:LBE_Compact}, the cubic collision operator must satisfy 
\begin{equation}
	\big( F_3 \vec{f}^{\,\otimes 3} \big) \big|_{m}^{kji}
	= -\sum_{q,r,s=1}^Q \gamma_{r,m,s}f_q^{kji}f_r^{kji}f_s^{kji} .
\end{equation}
Similar to before, define $B_{3,q}\in\mathbb{C}^{n\times Q^2n^3}$ as
\begin{equation} \label{eqn:B3q}
	B_{3,q} \coloneq
	\begin{pmatrix}
		\underbrace{
			\underbrace{0_{p_q^0}
				\; 1 \; 0_{\hat{p}_q^0}}_{Qn^2}
			\; \dots \;
			0_{p_q^0} \; 1 \; 0_{\hat{p}_q^0}}_{(Qn)^2} \\[-3em]
		& 0_{p_q^1} \; 1 \; 0_{\hat{p}_q^1}
		\; \dots \;
		0_{p_q^1} \; 1 \; 0_{\hat{p}_q^1} \\
		&&\ddots \\
		&&& 0_{p_q^{n-1}} \; 1 \; 0_{\hat{p}_q^{n-1}} 
		\; \dots \;
		0_{p_q^{n-1}} \; 1 \; 0_{\hat{p}_q^{n-1}} \\        
	\end{pmatrix} ,
\end{equation}
with $p_q^{\,j}=Qnj+n(q-1)+j$, $\hat{p}_q^{\,j}=Qn(n-j)-n(q-1)-j-1$ for $j\in\{0,\dots,n-1\}$. Here, the $0_{r}$ blocks are sized $1 \times r$ for an integer $r$. Using this and $\Gamma_q$ defined in \eqref{eqn:Gammaq}, it follows that 
\begin{equation} \label{eqn:F3}
	F_3 = -\sum_{q=1}^Q B_{3,q} \otimes \Gamma_q .
\end{equation}

Finally, bringing together \labelcref{eqn:SxSySz,eqn:F1tilde,eqn:F2,eqn:F3} we obtain the desired vectorized form of the LBE 
\begin{equation} \label{eqn:LBE_ExtraCompact}
	\frac{\partial \vec{f}}{\partial t} = F_1\vec{f} + F_2\vec{f}^{\,\otimes2} + F_3\vec{f}^{\,\otimes3} .
\end{equation}
The advantage of \eqref{eqn:LBE_ExtraCompact} is that it is in the form \eqref{eqn:GeneralIVP}, which means that we can use the methods described in Section \ref{sec:Carl Lin Zero Pad} to: (1) implement the zero padded Carleman linearization method to obtain the linear system $L^{(\text{e})}\vec{Y}^{(\text{e})}=\vec{B}^{(\text{e})}$, and (2) decompose $L^{(\text{e})}$ up to the level of $F_k$ matrices. In Appendix \ref{sec:Condition Number}, we perform a condition number analysis of this zero padded system $L^{(\text{e})}$ from \eqref{eqn:Le Matrix} relative to its non-padded counterpart $L$ from \eqref{eqn:L Matrix}, finding that $\kappa(L^{(e)}) \le \mathcal{O}(\sqrt{n_t} \kappa(L))$. This suggests that a substantial increase in condition number as a result of zero padding cannot be ruled out. In the next section, we find decompositions for $F_1$, $F_2^{(\text{e})}$ and $F_3^{(\text{e})}$ in form of \eqref{eqn:Llgeneral}, with the latter two being zero padded forms of $F_2$ and $F_3$ using \eqref{eqn:Fke def}.


\subsection{Decomposition of the $F_1$ Matrix} \label{sec:F1 Decomp}
From the previous section $F_1=S+\tilde{F}_1$, so here we treat $S$ and $\tilde{F}_1$ separately. Starting with the streaming operator $S$, we first apply the Pauli decomposition on the velocity matrices from \eqref{eqn:ExEyEz}. To do this, first let $\sigma_0=\sigma_x$, $\sigma_1=\sigma_y$, $\sigma_2=\sigma_z$, and $\sigma_3=I$, and consider the quaternary bitstring $\nu = \nu_1, \dots, \nu_N$, for $\nu_i \in \{0,1,2,3\}$ with $i \in \{1,\dots,N\}$ and some positive integer $N$. Then, a string of tensor products of Pauli operators may be written as $\sigma_\nu \coloneq \sigma_{\nu_1} \otimes \cdots \otimes \sigma_{\nu_N}$. Following this, if $f(\eta,m)$ is a quaternary bitstring of length $\log Q$ for $\eta \in \{x,y,z\}$ and an integer $m$, then the Pauli decomposition of $E_\eta$ is
\begin{equation} \label{eqn:E_eta}
	E_\eta = \sum_{m=1}^{N_{E_\eta}} b_{\eta,m} \sigma_{f(\eta,m)}  ,
\end{equation}
where $b_{\eta,m} \in \mathbb{C}$. Then, from \eqref{eqn:SxSySz} we may write
\begin{equation} \label{eqn:Sterm}
	S = \sum_{\eta \in \{x,y,z\}} \sum_{p\in\{+1,-1\}} \sum_{m=1}^{N_{E_\eta}} S_{\eta,p,m} ,
\end{equation}
with
\begin{equation} \label{eqn:S_etapm}
	S_{\eta,p,m} =
	\begin{cases} 
		p I_{n_yn_z} \otimes S_p^{n_x} \otimes \sigma_{f(x,m)} & \text{for } \eta = x \\
		p I_{n_z} \otimes S_p^{n_y} \otimes I_{n_x} \otimes \sigma_{f(y,m)} & \text{for } \eta = y \\
		p S_p^{n_z} \otimes I_{n_xn_y} \otimes \sigma_{f(z,m)} & \text{for } \eta = z ,
	\end{cases}
\end{equation}
where the incrementer $S_{+1}^r$ is defined in \eqref{eqn:Sp1} and the decrementer by $S_{-1}^r \coloneq (S_{+1}^r)^T$. 

Next, from \eqref{eqn:F1tilde}, we can see that $\tilde{F}_1 \coloneq I_n \otimes R$. The real valued matrix $R$ is decomposed in two steps. First, the singular value decomposition (SVD) is used to obtain $R = \frac{1}{\tau} W_R \Sigma_R V_R^T$ where $W_R, V_R \in \mathbb{R}^{Q \times Q}$ are unitary and $\Sigma_R \in \mathbb{R}^{Q \times Q}$ is a diagonal matrix. Second, if $g(m)$ is a quaternary bitstring of length $\log Q$ for an integer $m$, then the Pauli decomposition applied to $\Sigma_R$ is
\begin{equation*}
	\Sigma_R = \sum_{m=1}^{N_R} a_m \sigma_{g(m)} ,
\end{equation*}
where $a_m \in \mathbb{C}$. Bringing both steps together yields
\begin{equation} \label{eqn:R1 Decomp}
	R = \frac{1}{\tau} \sum_{m=1}^{N_{R}} a_m W_R \sigma_{g(m)} V_R^T .
\end{equation}

By combining \labelcref{eqn:S_etapm,eqn:Sterm,eqn:R1 Decomp}, we find that the total number of terms in the decomposition of $F_1$ is $N_{R} + 2(N_{E_x}+N_{E_y}+N_{E_z})$. To determine values for $N_{R}$, $N_{E_x}$, $N_{E_y}$ and $N_{E_z}$, we must select specific lattice structures. Here, we consider the commonly used D1Q3, D2Q9 and D3Q15 lattices defined in Appendix \ref{sec:Velocity Sets}. Since none of these cases satisfy the assumption from Section \ref{sec:LBE Derivation} of $Q$ being a power of two, we embed them into larger systems using the approach from Appendix \ref{sec:Velocity Sets Embed}. This transforms the D1Q3, D2Q9 and D3Q15 cases into D1Q3*, D2Q9* and D3Q15*, respectively, where the asterisk indicates that the $Q$ dimension has been embedded into the nearest power of two larger than $Q$. The relevant numerical values for these three latices are provided in Table \ref{tbl:Pauli SVD Decompositions}.

\begin{table}
	\begin{center}
		\renewcommand{\arraystretch}{1.25}
		\begin{tabular}{p{0.18\linewidth} | p{0.43\linewidth} | p{0.062\linewidth} | p{0.062\linewidth} | p{0.062\linewidth}}
			\hline
			\multicolumn{5}{c}{Relevant Numerical Figures by Lattice Type} \\
			\hline 
			Quantity & Description & D1Q3* & D2Q9* & D3Q15* \\
			\hline \hline
			$N_{E_x}$  & No. of Terms in $E_x$ Pauli Decomp. & 2  & 10   & 12   \\
			$N_{E_y}$  & No. of Terms in $E_y$ Pauli Decomp. & --  & 10   & 11   \\
			$N_{E_z}$  & No. of Terms in $E_z$ Pauli Decomp. & --  & --    & 10   \\
			$N_{R}$    & No. of Terms in $\Sigma_R$ Pauli Decomp. & 4  & 16   & 16  \\
			$N_\Gamma = \sum_{q=1}^Q N_{\Gamma_q}$ & No. of Terms in $\Sigma_{\Gamma_q}$ Pauli Decomp. (summed)
			& 8  & 128  & 224 \\
			$G[CW_R] + G[CV_R]$ & T cost of Controlled SVD unitaries in \eqref{eqn:R1 Decomp} 
			& 27 & 272 & 273 \\
			$G[CW_{\Gamma}] + G[CV_{\Gamma}]$ & T cost of Controlled SVD unitaries in \eqref{eqn:Gamma Decomp}
			& 20 & 255 & 265 \\
		\end{tabular}
	\end{center}
	\caption{Numerical results associated with the Pauli Decompositions and SVDs used in \labelcref{eqn:R1 Decomp,eqn:E_eta,eqn:Gamma Decomp}. For the Pauli decompositions, we use the \texttt{SparsePauliOp} from IBM's \texttt{Qiskit} software using an error tolerance of $10^{-8}$, which provides the results reported rows 1--5 of the table. For rows 6--7, we first decompose the SVD unitaries using the \texttt{Qiskit} transpiler with basis gates $\{CX,U3\}$, then we apply the \texttt{Qiskit SolovayKitaev} algorithm to decompose the single qubit gates into the gate set $\{T, T^\dagger, H\}$. Note, we assume that $G[CW_{\Gamma_q}] = G[CW_\Gamma]$ and $G[CV_{\Gamma_q}] = G[CV_\Gamma]$, for all $q \in \{1,\dots,Q\}$.}
	\label{tbl:Pauli SVD Decompositions}
\end{table}


\subsection{Decomposition of the $F_2$ and $F_3$ Matrices } \label{sec:F2 F3 Decomp}
The $F_2$ and $F_3$ matrices share similar structures so they are decomposed together. From Section \ref{sec:LBE Derivation}, we have the general form
\begin{equation*}
	F_k = d_k \sum_{q=1}^Q B_{k,q} \otimes \Gamma_q ,
\end{equation*}
where $k \in \{2,3\}$, $d_2=2$, $d_3=-1$, $\Gamma_q$ is defined in \eqref{eqn:Gammaq} and $B_{k,q}\in\mathbb{C}^{n\times Q^{k-1}n^k}$ is defined in \labelcref{eqn:Bq,eqn:B3q} for $k=2$ and $k=3$, respectively. Next, $F_k$ is embedded in a square matrix by zero padding using \eqref{eqn:Fke def}, yielding the relations
\begin{equation} \label{eqn:F20}
	F_k^{(\text{e})} \coloneq
	\begin{pmatrix} F_k \\ 0_{((Qn)^k-Qn)\times (Qn)^k} \end{pmatrix} 
	= \begin{pmatrix} d_k\sum_{q=1}^Q B_{k,q} \otimes \Gamma_q \\ 0_{((Qn)^k-Qn)\times (Qn)^k} \end{pmatrix} 
	= d_k \sum_{q=1}^Q\begin{pmatrix} B_{k,q} \\ 0_{(Q^{k-1}n^k-n)\times Q^{k-1}n^k} \end{pmatrix} \otimes \Gamma_q .
\end{equation}
If we define 
\begin{equation} \label{eqn:Dk}
	D_k \coloneq \rho_0^{\otimes \log(Qn)^{k-1}} \otimes I_{n} ,
\end{equation}
then it follows that
\begin{equation} \label{eqn:DPq}
	\begin{pmatrix} B_{k,q} \\ 0_{(Q^{k-1}n^k-n)\times Q^{k-1}n^k} \end{pmatrix}
	= 2^{-(k-2)(\log Q)/2} D_k \overline{B}_{k,q} ,
\end{equation}
where $\overline{B}_{k,q}$ is a unitary completion of $B_{k,q}$ and the coefficient is derived in Appendix \ref{sec:Circuit for B3q}. Next, we apply the same combined SVD and Pauli decomposition technique as used in \eqref{eqn:R1 Decomp} for each matrix $\Gamma_q$. If $h(q,m)$ is a quaternary bitstring of length $\log Q$ for $q \in \{1,\dots,Q\}$ and an integer $m$, then the desired decomposition of each $\Gamma_q$ is
\begin{equation} \label{eqn:Gamma Decomp}
	\Gamma_q 
	=\frac{1}{\tau} W_{\Gamma_q} \Sigma_{\Gamma_q} V_{\Gamma_q}^T
	=\frac{1}{\tau} \sum_{m=1}^{N_{\Gamma_q}} c_{q,m} W_{\Gamma_q} \sigma_{h(q,m)} V_{\Gamma_q}^T ,
\end{equation}
where $c_{q,m} \in \mathbb{C}$, the SVD is used to obtain the first equality with unitaries $W_{\Gamma_q}, V_{\Gamma_q} \in \mathbb{R}^{Q \times Q}$ and diagonal matrix $\Sigma_{\Gamma_q} \in \mathbb{R}^{Q \times Q}$ for each $q \in \{1,\dots,Q\}$, and the second equality is the result of the Pauli decomposition on $\Sigma_{\Gamma_q}$. Finally, inserting \labelcref{eqn:DPq,eqn:Gamma Decomp} into \eqref{eqn:F20} yields
\begin{equation} \label{eqn:Fke expression}
	F_k^{(\text{e})} =
	\sum_{q=1}^Q \sum_{m=1}^{N_{\Gamma_q}} \hat{d}_{k,q,m}
	\left( D_k \overline{B}_{k,q} \right)
	\otimes \left( W_{\Gamma_q} \sigma_{h(q,m)} V_{\Gamma_q}^T \right) ,
\end{equation}
where $\hat{d}_{k,q,m} = c_{q,m} d_k 2^{-(k-2)(\log Q)/2}$. This form is desirable because it is of the form \eqref{eqn:Llgeneral} and $\overline{B}_{k,q}$ is a unitary matrix with an efficient circuit for $k=2$ (Appendix \ref{sec:Circuit for B2q}) and $k=3$ (Appendix \ref{sec:Circuit for B3q}). For example, for the $k=2$ case we insert \eqref{eqn:B2q circ eq} into \eqref{eqn:Fke expression} to obtain
\begin{equation}
	F_2^{(\text{e})} = 
	\sum_{q=1}^Q \sum_{m=1}^{N_{\Gamma_q}}
	\hat{d}_{2,q,m}  \Biggl(
	\left( \rho_0^{\otimes \log n} \otimes \big( \rho_0^{\otimes \log Q} \mathcal{X}_{\log Q}(q-1) \big) 
	\otimes I_n \right)
	M_{Qn+1}^n 
	\Biggr)
	\otimes \left( W_{\Gamma_q} \sigma_{h(q,m)} V_{\Gamma_q}^T \right) ,
\end{equation}
where $\mathcal{X}_{\log Q}(q-1)$ is composed solely of identiy and NOT gates (Appendix \ref{sec:Pk circuit}), and $M_{m+1}^r$ is composed solely of identity, NOT and CNOT gates (Appendix \ref{sec:Construction of M}). A similar expansion for $F_3^{(\text{e})}$ can be made by inserting \eqref{eqn:B3q circ eq} into \eqref{eqn:Fke expression} with a similar resource cost.


\subsection{Explicit Circuit Constructions} \label{sec:LBE Circuits}
We now combine the results of the previous sections to explicitly construct circuits for each term in the decomposition of the zero padded Carleman linearized LBE. From \labelcref{eqn:Le,eqn:L2e,eqn:Ae_expr}, we can see that $L^{(\text{e})}$ is split into four types of terms: $L_1^{(\text{e})}$, constant forcing, linear and nonlinear. Since we are considering a homogeneous equation, the constant forcing term is zero. Therefore, we may write $L^{(\text{e})} = L_1^{(\text{e})} + L_\text{lin}^{(\text{e})} + L_\text{nlin}^{(\text{e})}$, where the first term is defined in \eqref{eqn:L1e}, the second represents the linear components and the third the nonlinear components. Since we grouped linear terms in Sections \labelcref{sec:LBE Derivation,sec:F1 Decomp} by either streaming or collision, we now write $L^\text{(e)}_{\text{lin}} = L^\text{(e)}_{\text{lin,1}} + L^\text{(e)}_{\text{lin,2}}$, where the first is associated with $S$ and the second with $\tilde{F}_1$. In the following subsections, the general expression for each type of term is combined with Algorithms \labelcref{alg:Ul1,alg:Ul2} to construct explicit circuits.

\subsubsection{Circuits for $L_1^{(\textnormal{e})}$} \label{sec:L1e}
Using \eqref{eqn:L1e}, we can write $L_1^{(\text{e})} = \sum_{i=1}^3 c_i L_{1,i}$ with 
\begin{equation} \label{eqn:L1e Terms}
\begin{split}
	L_{1,1} &= I_{2n_t(Qn)^\alpha} , \\
	L_{1,2} &= \rho_1^{\otimes \log n_t} \otimes I_{2(Qn)^\alpha} , \\
	L_{1,3} &= S_{+1}^{n_t} \otimes I_{2(Qn)^\alpha} ,
\end{split}
\end{equation}
where $c_1=c_2=1$ and $c_3=-1$. To obtain the circuits for each of these expressions, we use the embedding procedure from Section \ref{sec:LinComBlock}. Using Theorem \ref{thm:Theorem for Lbar}, we find
\begin{equation} \label{eqn: L1e completions}
	\begin{split}
		\overline{L}_{1,1} &= I_{2n_t(Qn)^\alpha} , \\
		\overline{L}_{1,2} &= \sigma_x^{\otimes \log n_t} \otimes I_{2(Qn)^\alpha} , \\
		\overline{L}_{1,3} &= S_{+1}^{n_t} \otimes I_{2(Qn)^\alpha} .
	\end{split}
\end{equation}
Additionally, we can see that
\begin{equation} \label{eqn:L1eL1eT}
	\begin{split}
		L_{1,1}L_{1,1}^T &= I_{2n_t(Qn)^\alpha} , \\
		L_{1,2}L_{1,2}^T &= \rho_0^{\otimes \log n_t} \otimes I_{2(Qn)^\alpha} , \\
		L_{1,3}L_{1,3}^T &= I_{2n_t(Qn)^\alpha} .
	\end{split}
\end{equation}
Using the expressions in \labelcref{eqn: L1e completions,eqn:L1eL1eT} with Algorithms \labelcref{alg:Ul1,alg:Ul2}, we construct the circuits for the $L_{1,2}$ and $L_{1,3}$ terms in Figure \ref{fig:L1e Circuit}, where $L_{1,1}$ is excluded since it is trivial.

\begin{figure}
	\centering
	\hspace{-4em}
	\begin{subfigure}[t]{0.48\textwidth}
		\centering
		\includegraphics[]{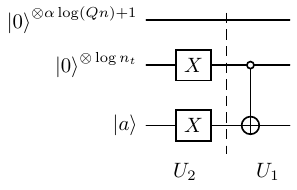}
		\vspace{-1em}
		\begin{equation*}
			\tcbhighmath{
				\begin{split}
					L_{1,2} &= \rho_1^{\otimes \log n_t} \otimes I_{2(Qn)^\alpha} \\
					\overline{L}_{1,2} &= \sigma_x^{\otimes \log n_t} \otimes I_{2(Qn)^\alpha} \\
					L_{1,2} L_{1,2}^T &= \rho_0^{\otimes \log n_t} \otimes I_{2(Qn)^\alpha}
			\end{split}	}
		\end{equation*}
	\end{subfigure}
	\begin{subfigure}[t]{0.48\textwidth}
		\centering
		\includegraphics[]{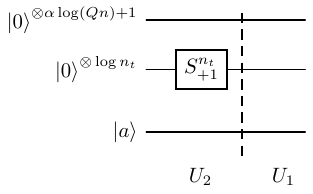}
		\vspace{-1em}
		\begin{equation*}
			\tcbhighmath{
				\begin{split}
					L_{1,3} &= S_{+1}^{n_t} \otimes I_{2(Qn)^\alpha} \\
					\overline{L}_{1,3} &= S_{+1}^{n_t} \otimes I_{2(Qn)^\alpha} \\
					L_{1,3} L_{1,3}^T &= I_{2n_t(Qn)^\alpha} 
			\end{split}	}
		\end{equation*}
	\end{subfigure}
	\caption{Embeddings for the two nontrivial $L_1^{(\text{e})}$ terms of \eqref{eqn:L1e Terms}: the emebedded $L_{1,2}$ term (left) and the embedded $L_{1,3}$ term (right). The $\ket{a}$ wire is a single ancillary qubit required to embed each term into a unitary operation. A control operation or a single qubit gate on a multi-qubit register should be interpreted as the respective operation being applied to every wire within that register. The vertical dashed line separates the $U_1$ component from the $U_2$ component corresponding to the embedding procedure from \eqref{eqn:Ul1 Ul2}. The equations below the circuits are: the desired term to load ($L$), its completion ($\overline{L}$) which constructs $U_2$ using Algorithm \ref{alg:Ul2}, and $LL^T$ which constructs $U_1$ using Algorithm \ref{alg:Ul1}. Circuits drawn using \cite{Kay2018}.}
	\label{fig:L1e Circuit}
\end{figure}

\subsubsection{Circuits for $L_\textnormal{lin,1}^{(\textnormal{e})}$: Streaming Term} \label{sec:Llin1}
We now turn our attention to the streaming terms associated with $S$. Define the set of tuples
\begin{equation*}
\begin{split}
	\Lambda_1 \coloneq \{ (i,j,l,\eta,p,m) \mid &i\in\{1,2\}, j\in\{1,\dots\alpha\}, l\in\{0,\dots,j-1\},\\ 
	&\eta\in\{x,y,z\}, p\in\{+1,-1\}, m\in\{1,\dots, N_{E_\eta}\} \}.
\end{split}
\end{equation*}
Then, by combining \eqref{eqn:L2e} with the linear component of \eqref{eqn:Ae_expr} and inserting \labelcref{eqn:Ajpkm1j,eqn:Sterm}, we obtain
\begin{subequations} 
	\begin{align}
		L_\text{lin,1}^{(\text{e})} &= \sum_{\lambda_\in\Lambda_1} c_{\lambda}^\text{lin,1} L_{\lambda}^\text{lin,1} , \label{eqn:Llin1a} \\
		L_{\lambda}^\text{lin,1} & = \tilde{I}_i \otimes \left( \rho_0 \otimes \rho_3^{\otimes \log(Qn) -1} \right)^{\otimes \alpha-j}
		\otimes \rho_3 \otimes I_{(Qn)^l} \otimes S_{\eta,p,m} \otimes I_{(Qn)^{j-l-1}} , \label{eqn:Llin1b} \\
		c_{\lambda}^\text{lin,1} &= (-1)^{i-1} b_{\eta,m} \label{eqn:Llin1c} ,
	\end{align}
\end{subequations}
where $b_{\eta,m} \in \mathbb{C}$ comes from \eqref{eqn:E_eta}, $S_{\eta,p,m}$ is defined in \eqref{eqn:S_etapm}, and we define $\tilde{I}_1 \coloneq I_{n_t}$ and $\tilde{I}_2 \coloneq \rho_0^{\otimes \log n_t}$. By applying Theorem \ref{thm:Theorem for Lbar} to \eqref{eqn:Llin1b}, we have
\begin{equation} \label{eqn:Llin1 bar}
	\overline{L_{\lambda}^\text{lin,1}} = I_{2n_t(Qn)^{\alpha-j+l}} \otimes S_{\eta,p,m} \otimes I_{(Qn)^{j-l-1}} .
\end{equation}
Next, using $S_{\eta,p,m} S_{\eta,p,m}^T = I_{Qn}$ and $\tilde{I}_i \tilde{I}_i^T = \tilde{I}_i$, we find that 
\begin{equation} \label{eqn:Llin1 Llin1T}
	L_{\lambda}^\text{lin,1}(L_{\lambda}^\text{lin,1})^T = \tilde{I}_i \otimes \left( \rho_0 \otimes \rho_3^{\otimes \log(Qn) -1} \right)^{\otimes \alpha-j} \otimes \rho_3 \otimes I_{(Qn)^j} .
\end{equation}
Using \labelcref{eqn:Llin1 bar,eqn:Llin1 Llin1T} with Algorithms \labelcref{alg:Ul1,alg:Ul2}, we can construct the circuit for $L_{\lambda}^\text{lin,1}$ for any value of $\lambda \in \Lambda_1$, with the most expensive such circuit illustrated in Figure \ref{fig:Lin1 Circuit} corresponding to $\lambda = (2,1,0,x,+1,m)$ for any $m \in \{1,\dots,N_{E_\eta}\}$.

\begin{figure}[!htbp]
	\centering
	\includegraphics[]{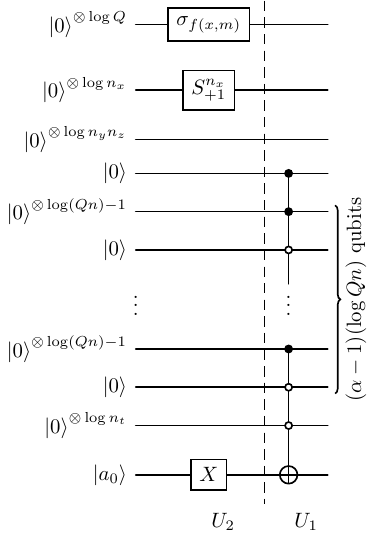}
	\vspace{-1em}
	\begin{equation*}
		\tcbhighmath{
			\begin{split}
			L_{\lambda}^\text{lin,1} &= \rho_0^{\otimes \log n_t} \otimes (\rho_0 \otimes \rho_3^{\otimes \log (Qn)-1})^{\otimes \alpha-1}
			\otimes \rho_3 \otimes I_{n_yn_z} \otimes S_{+1}^{n_x} \otimes \sigma_{f(\eta,m)} \\[5pt]
			\overline{L_{\lambda}^\text{lin,1}} &= I_{2 n_t (Qn)^{\alpha-1}} \otimes I_{n_yn_z} \otimes S_{+1}^{n_x} \otimes 	\sigma_{f(\eta,m)} \\[5pt]
			L_{\lambda}^\text{lin,1} (L_{\lambda}^\text{lin,1})^T &= \rho_0^{\otimes \log n_t} \otimes (\rho_0 	\otimes \rho_3^{\otimes \log (Qn)-1})^{\otimes \alpha-1}
			\otimes \rho_3 \otimes I_{Qn} 
		\end{split}	}
	\end{equation*}	
	\caption{Embedding for $L_{\lambda}^\text{lin,1}$ as defined in \eqref{eqn:Llin1b} using $\lambda = (2,1,0,x,+1,m)$ for any $m \in \{1,\dots,N_{E_\eta}\}$, which is the most expensive circuit among any $\lambda\in\Lambda_1$. The $\sigma_{f(\eta,m)}$ block is the $m\text{th}$ term in the Pauli decomposition of $x$-component of \eqref{eqn:E_eta}, and is therefore a tensor product of $\log Q$ Pauli gates. The $S_{+1}^n$ block is the incrementer circuit from \eqref{eqn:Sp1}.}
	\label{fig:Lin1 Circuit}
\end{figure}

\subsubsection{Circuits for $L_\textnormal{lin,2}^{(\textnormal{e})}$: Linear Collision Term} \label{sec:Llin2}
We now turn our attention to the linear collision terms associated with $\tilde{F}_1$. Define the set of tuples
\begin{equation*}
	\Lambda_2 = \{(i,j,l,m) \mid i\in\{1,2\}, j\in\{1,\dots\alpha\}, l\in\{0,\dots,j-1\}, m\in\{1,\dots,N_{R}\} \} .
\end{equation*}
Then, by combining \eqref{eqn:L2e} with the linear component of \eqref{eqn:Ae_expr} and inserting \labelcref{eqn:Ajpkm1j,eqn:F1tilde}, we obtain
\begin{subequations} 
	\begin{align}
		L_\text{lin,2}^{(\text{e})} &= \sum_{\lambda \in \Lambda_2} c_{\lambda}^\text{lin,2} L_{\lambda}^\text{lin,2}
		, \label{eqn:Llin2a} \\
		L_{\lambda}^\text{lin,2}
		&= \tilde{I}_i \otimes \left( \rho_0 \otimes \rho_3^{\otimes \log(Qn) -1} \right)^{\otimes \alpha-j}
		\otimes \rho_3 \otimes I_{n(Qn)^l} 
		\otimes ( W_R \sigma_{g(m)} V_R^T ) \otimes I_{(Qn)^{j-l-1}} , \label{eqn:Llin2b} \\
		c_{\lambda}^\text{lin,2} &= \frac{1}{\tau}(-1)^{i-1} a_m , \label{eqn:Llin2c}
	\end{align}
\end{subequations}
where $a_m \in \mathbb{C}$, $W_R$, $V_R$ and $\sigma_{g(m)}$ come from \eqref{eqn:R1 Decomp}. Using the mixed product property to factor out $W_R$ and $V_R^T$ on the left and right sides of \eqref{eqn:Llin2b}, respectively, and then applying Corollary \ref{thm:UPV} we have 
\begin{equation} \label{eqn:Llin2 bar}
	\overline{L_{\lambda}^\text{lin,2}} = I_{2n_t(Qn)^{\alpha-j+l}n} \otimes ( W_R \sigma_{g(m)} V_R^T ) \otimes I_{(Qn)^{j-l-1}} .
\end{equation}
Next, from \eqref{eqn:Llin2b} we can see that
\begin{equation} \label{eqn:Llin2 Llin2T}
	L_{\lambda}^\text{lin,2}(L_{\lambda}^\text{lin,2})^T = \tilde{I}_i \otimes \left( \rho_0 \otimes \rho_3^{\otimes \log(Qn) -1} \right)^{\otimes \alpha-j} \otimes \rho_3 \otimes I_{(Qn)^j} ,
\end{equation}
where we have used $\tilde{I}_i\tilde{I}_i^T = \tilde{I}_i$ and $( W_R \sigma_{g(m)} V_R^T )( W_R \sigma_{g(m)} V_R^T )^T = I_Q$ for all $m \in \{1,\dots,N_R\}$. Using \labelcref{eqn:Llin2 bar,eqn:Llin2 Llin2T} with Algorithms \labelcref{alg:Ul1,alg:Ul2}, we can construct the circuit for $L_\lambda^\text{lin,2}$ for any value of $\lambda \in \Lambda_2$, with the most expensive such circuit illustrated in Figure \ref{fig:Lin2 Circuit} corresponding to $\lambda = (2,1,0,m)$ for any $m \in \{1,\dots,N_R\}$.

\begin{figure}[!htbp]
	\centering
	\includegraphics[]{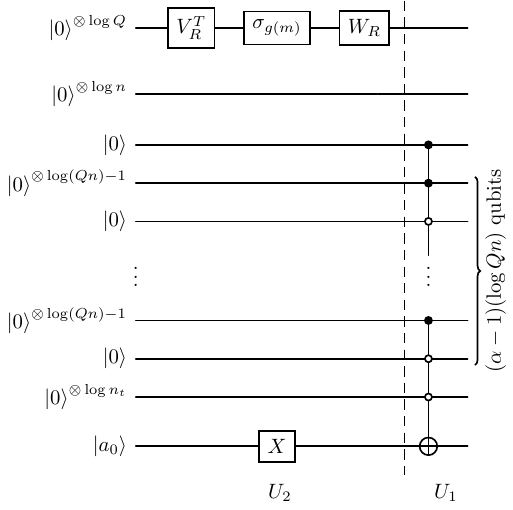}
	\vspace{-1em}
	\begin{equation*}
		\tcbhighmath{
			\begin{split}
				L_{\lambda}^\text{lin,2} &= \rho_0^{\otimes \log n_t} \otimes (\rho_0 \otimes \rho_3^{\otimes \log (Qn)-1})^{\otimes \alpha-1}
				\otimes \rho_3 \otimes I_n \otimes ( W_R \sigma_{g(m)} V_R^T ) \\[5pt]
				\overline{L_{\lambda}^\text{lin,2}} &= I_{2 n_t Q^{\alpha-1} n^\alpha} \otimes ( W_R \sigma_{g(m)} V_R^T )
				\\[5pt]
				L_{\lambda}^\text{lin,2} (L_{\lambda}^\text{lin,2})^T &= \rho_0^{\otimes \log n_t} \otimes (\rho_0 \otimes \rho_3^{\otimes \log (Qn)-1})^{\otimes \alpha-1}
				\otimes \rho_3 \otimes I_{Qn}
		\end{split} }
	\end{equation*}		
	\caption{Embedding for $L_{\lambda}^\text{lin,2}$ as defined in \eqref{eqn:Llin2b} using $\lambda = (2,1,0,m)$ for any $m \in \{1,\dots,N_R\}$, which is the most expensive circuit among any $\lambda\in\Lambda_2$. The $\sigma_{g(m)}$ gate is the $m\text{th}$ term in the Pauli decomposition used in \eqref{eqn:R1 Decomp}, and is therefore a tensor product of $\log Q$ Pauli gates. The $V_R$ and $W_R$ circuits come from the SVD in \eqref{eqn:R1 Decomp} and depend on the specific lattice structure used.}
	\label{fig:Lin2 Circuit}
\end{figure}

\subsubsection{Circuits for $L_\textnormal{nlin}^{(\textnormal{e})}$: Nonlinear Collision Terms}
The final component to analyze is $L_\text{nlin}^{(\text{e})}$, which contains both the quadratic and cubic nonlinear collision terms. First, define the set of tuples
\begin{equation*}
\begin{split}
	\Lambda_3 = \{ (k,i,j,l,q,m) \mid &k\in\{2,3\}, i\in\{1,2\}, j\in\{1,\dots,\alpha-k+1\},\\ &l\in\{0,\dots,j-1\}, q\in\{1,\dots,Q\}, m\in\{1,\dots,N_{\Gamma_q}\}  \}. 
\end{split}
\end{equation*}
Then, by combining \eqref{eqn:L2e} with the nonlinear component of \eqref{eqn:Ae_expr}, and inserting \labelcref{eqn:Aejpkm1j,eqn:Fke expression} we obtain
\begin{subequations} 
	\begin{align}
		L_{\text{nlin}}^{(\text{e})} &= \sum_{\lambda\in\Lambda_3} c_{\lambda}^\text{nlin} L_{\lambda}^\text{nlin} , \label{eqn:Lnonlina} \\
		\begin{split}
		L_{\lambda}^\text{nlin} &= \tilde{I}_i \otimes 
		(\rho_0 \otimes \rho_3^{\otimes \log (Qn)-1})^{\otimes \alpha-k-j+1} 
		\otimes \rho_{1}
		\otimes \big[ \left( \mathcal{P}_k \otimes I_{(Qn)^j} \right) \cdot E \big] , \\
		E &= \bigg[ \left( \rho_0^{\otimes \log(Qn)^{k-1}} \otimes K^{((Qn)^l,Qn)} \right) \\
		&\quad\qquad \cdot 
		\Big( 
		\left( D_k \overline{B}_{k,q} \right) 
		\otimes \left( W_{\Gamma_q} \sigma_{h(q,m)} V_{\Gamma_q}^T \right)
		\otimes I_{(Qn)^l} 
		\Big)
		\cdot K^{((Qn)^k,(Qn)^l)} \bigg]
		\otimes I_{(Qn)^{j-l-1}} , 
		\end{split} \label{eqn:Lnonlinb} \\
		c_{\lambda}^\text{nlin} &= \frac{1}{\tau}(-1)^{i-1} \hat{d}_{k,q,m} , \label{eqn:Lnonlinc}
	\end{align}		 
\end{subequations}
where $\tilde{I}_i$ is defined in Section \ref{sec:Llin1}, $\mathcal{P}_k$ is defined in Appendix \ref{sec:Pk circuit}, $K^{(a,b)}$ is the commutation matrix from Appendix \ref{sec:Circ for Commutation}, $D_k$ is defined in \eqref{eqn:Dk}, $\overline{B}_{k,q}$ is defined in Appendix \ref{sec:Circuit for B2q} (for $k=2$) and \ref{sec:Circuit for B3q} (for $k=3$), $W_{\Gamma_q}$, $V_{\Gamma_q}$ and $\sigma_{h(q,m)}$ come from \eqref{eqn:Gamma Decomp} and the coefficient $\hat{d}_{k,q,m}$ is defined in Section \ref{sec:F2 F3 Decomp}. After factoring out appropriate terms using the mixed product property on \eqref{eqn:Lnonlinb}, and then applying Corollary \ref{thm:UPV}, we have 
\begin{equation} \label{eqn:Lnonlin bar}
	\begin{split}
		\overline{L_{\lambda}^\text{nlin}} &= 
		I_{2n_t(Qn)^{\alpha-k-j+1}} \otimes \sigma_x
		\otimes \big[ \left( \mathcal{P}_k \otimes I_{(Qn)^j} \right) \overline{E} \big] , \\
		\overline{E} &= \bigg[ 
		\left( I_{(Qn)^{k-1}} \otimes K^{((Qn)^l,Qn)} \right) \\
		&\qquad\qquad \left( \overline{B}_{k,q} \otimes \left( W_{\Gamma_q} \sigma_{h(q,m)} V_{\Gamma_q}^T \right)
		\otimes I_{(Qn)^l} \right)
		K^{((Qn)^k,(Qn)^l)} 
		\bigg] 
		\otimes I_{(Qn)^{j-l-1}} .
	\end{split}
\end{equation}
Next, from \eqref{eqn:Lnonlinb} we can see that
\begin{equation} \label{eqn:Lnonlin LnonlinT} 
	L_{\lambda}^\text{nlin} (L_{\lambda}^\text{nlin})^T = \tilde{I}_i \otimes (\rho_0 \otimes \rho_3^{\otimes \log(Qn) -1})^{\otimes \alpha-k-j+1} 
	\otimes \rho_0 \otimes \rho_3^{\otimes \log(Qn)^{k-1}} \otimes I_{(Qn)^j} ,
\end{equation}
where we have used the fact that the following are all unitary matrices: $\mathcal{P}_k$, $K^{(a,b)}$, $\overline{B}_{k,q}$, $W_{\Gamma_q}$, $V_{\Gamma_q}$ and $\sigma_{h(q,m)}$. Additionally, we used the properties $\tilde{I}_i\tilde{I}_i^T=\tilde{I}_i$, $\mathcal{P}_2 \rho_0^{\otimes \log Qn} \mathcal{P}_2^T = \rho_3^{\otimes \log Qn}$, and $\mathcal{P}_3\rho_0^{\otimes 2\log Qn}\mathcal{P}_3^T=\rho_3^{\otimes \log (Qn)-1} \otimes \rho_0 \otimes \rho_3^{\otimes \log Qn}$. Using \labelcref{eqn:Lnonlin bar,eqn:Lnonlin LnonlinT} with Algorithms \labelcref{alg:Ul1,alg:Ul2}, we can construct the circuit for $L_\lambda^\text{nlin}$ for any value of $\lambda \in \Lambda_3$, with the most expensive such circuit illustrated in Figure \ref{fig:nonlin Circuit} corresponding to $\lambda = (3,2,\alpha-2,\alpha-3,q,m)$ for any $m \in \{1,\dots,N_{\Gamma_q}\}$ and $q \in \{1,\dots,Q\}$. 

\begin{figure}
	\centering
	\includegraphics[]{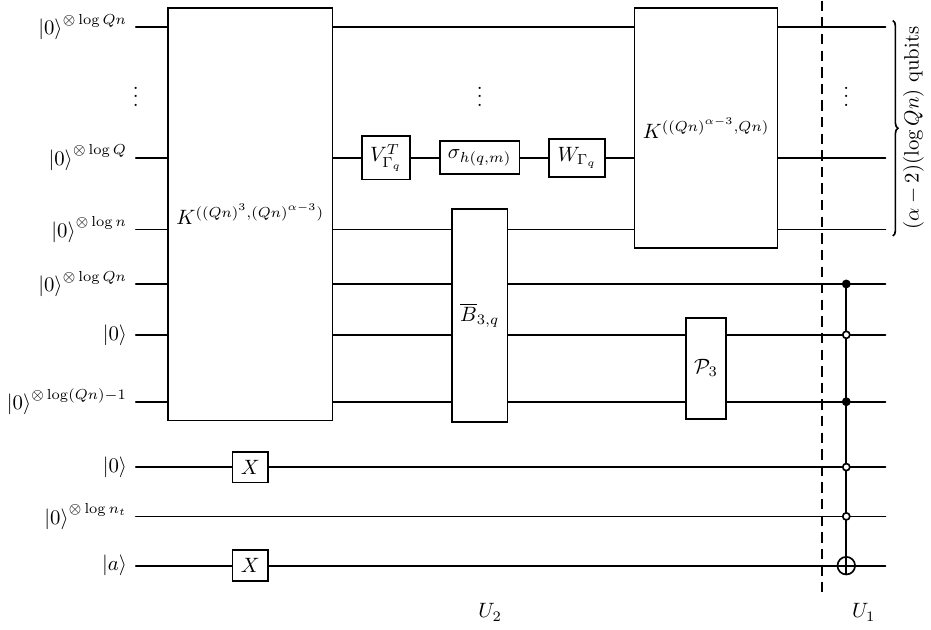}
	\begin{equation*}
	\tcbhighmath{
		\begin{split}
			L_{\lambda}^\text{nlin}
			&= \rho_0^{\otimes \log n_t} \otimes \rho_1 \otimes 
			\Big[ 
			\left( \mathcal{P}_3 \otimes I_{(Qn)^{\alpha-2}} \right) 
			\left( \rho_0^{\otimes 2\log Qn} \otimes K^{((Qn)^{\alpha-3},Qn)} \right) \\
			&\qquad 
			\left( (D_3 \overline{B}_{3,q}) \otimes \left( W_{\Gamma_q} \sigma_{h(q,m)} V_{\Gamma_q}^T \right) \otimes I_{(Qn)^{\alpha-3}} \right)
			K^{((Qn)^3,(Qn)^{(\alpha-3)})}
			\Big] \\[5pt]
			\overline{L_{\lambda}^\text{nlin}} 
			&= I_{2n_t} \otimes 
			\Big[ 
			\left( \mathcal{P}_3 \otimes I_{(Qn)^{\alpha-2}} \right)
			\left( I_{(Qn)^2} \otimes K^{((Qn)^{\alpha-3},Qn)} \right) \\
			&\qquad 
			\left( \overline{B}_{3,q} \otimes \left( W_{\Gamma_q} \sigma_{h(q,m)} V_{\Gamma_q}^T \right) \otimes I_{(Qn)^{\alpha-3}} \right)
			K^{((Qn)^3,(Qn)^{(\alpha-3)})}
			\Big]
			\\[5pt]
			L_{\lambda}^\text{nlin} (L_{\lambda}^\text{nlin})^T 
			&= \rho_0^{\otimes \log n_t} \otimes \rho_0 \otimes \rho_3^{\otimes \log (Qn)-1} \rho_0 
			\otimes \rho_3^{\log Qn} \otimes I_{(Qn)^{\alpha-2}}
		\end{split} }
	\end{equation*}
	\caption{Embedding for $L_\lambda^\text{nlin}$ as defined in \eqref{eqn:Lnonlinb} using  $\lambda=(3,2,\alpha-2,\alpha-3,q,m)$ for any $m \in \{1,\dots,N_{\Gamma_q}\}$ and $q \in \{1,\dots,Q\}$, which is the most expensive circuit among any $\lambda\in\Lambda_3$. The circuits for $\mathcal{P}_3$, $\overline{B}_{3,q}$ and the commutation matrix $K^{(a,b)}$ for integers $a$ and $b$ are provided in Appendices \labelcref{sec:Circuit for B3q,sec:Circ for Commutation,sec:Pk circuit}, respectively. The $\sigma_{h(q,m)}$ gate is the $m\text{th}$ term in the Pauli decomposition used in \eqref{eqn:Gamma Decomp} and is therefore a tensor product of $\log Q$ Pauli gates. The $V_{\Gamma_q}$ and $W_{\Gamma_q}$ circuits come from the SVD in \eqref{eqn:Gamma Decomp} and are determined for a specific lattice structure.}
	\label{fig:nonlin Circuit}
\end{figure}


\section{Resource Estimation} \label{sec:ResourceEstimation}
In this section, we estimate the T gate counts required to implement the decomposition from Section \ref{sec:Carl_LBE} into (1) fault tolerant algorithms using the PREP and SELECT access model from \cite{Childs2012LCU,Hariprakash2025}, and (2) the variational quantum linear solver (VQLS) algorithm. In both methods, we consider the cost of only a single call to the data loading subroutine. The gate counts are estimated following \cite{Jennings2025LBE} whereby they introduce the function $G[U]$, which is defined as the T gate count to implement an arbitrary unitary circuit $U$ with zero error. Here, we consider an LCU of the form \eqref{eqn:LCU} with the assumption that $L_l$ satisfies Theorem \ref{thm:Theorem for U1}, which yields the following relations for the T cost
\begin{equation} \label{eqn:Cost Ul}
	G[U_l] =  G[U_{l,1}] + G[U_{l,2}] = G[C^{\mathcal{T}(L_l)}X] + G[\overline{L_l}] ,
\end{equation}
where $C^kX$ is a multi-controlled NOT operator with $k$ controls and $\mathcal{T}(A)$ is a function which counts the number of elements from the set $\{\rho_0,\rho_3\}$ in the product $AA^T$ for some matrix $A$ of the form \eqref{eqn:Llgeneral}. For example, if $A=\rho_0\otimes\rho_1\otimes\rho_2\otimes\rho_3\otimes I$, then $AA^T=\rho_0^{\otimes2} \otimes \rho_3^{\otimes2} \otimes I$ and so $\mathcal{T}(A)=4$. The first term of the RHS of \eqref{eqn:Cost Ul} comes from Algorithm \ref{alg:Ul1}, where $U_{l,1}$ is implemented with a single multi-control NOT gate having $\mathcal{T}(L_l)$ control operations. The second term comes from the definition $U_{l,2}\coloneq\sigma_x \otimes \overline{L_l}$ (Algorithm \ref{alg:Ul2}). Therefore, to calculate the T gate count of any $U_l$, we require only $\mathcal{T}(L_l)$ and $\overline{L_l}$. In the following subsections, we compute the T gate count to load $L^{(\text{e})}$ from Section \ref{sec:LBE Circuits}.

\begin{table}
	\begin{center}
		\renewcommand{\arraystretch}{1.15}
		\begin{tabular}{p{0.02\linewidth} | p{0.08\linewidth} | p{0.25\linewidth} | p{0.18\linewidth} | p{0.12\linewidth} | p{0.1\linewidth}}
			\hline
			\multicolumn{6}{c}{Relevant Gate Costs} \\
			\hline 
			ID & Type & Description & T Count & Ancilla & Sources \\
			\hline \hline
			\newtag{1}{row:CX} & $CX$ & CNOT & 0 & 0 & -- \\ 
			\newtag{2}{row:Tof} & $C^2X$ & Toffoli & 7 & 0 & \cite{Selinger2013} \\
			\newtag{3}{row:CkX} & $C^kX$ & Multi-controlled \textsc{NOT} & $8k-12$ & two clean & \cite{Khattar2025} \\
			\newtag{4}{
			} & $S_{+1}^r$ & $r \times r$ Incrementer & $12\log r$ & $\log^* r$ clean & \cite{Khattar2025} \\
			\newtag{5}{row:CInc} & $CS_{+1}^r$ & Controlled Incrementer & $12\log 2r$ & $\log^* 2r$ clean & \cite{Gidney2015Increment} \\
			\newtag{6}{row:CSWAP} & $C\text{SWAP}$ & Fredkin & $7$ & 0 & \cite{Cruz2024} \\
			\newtag{7}{row:CkU} & $C^kU$ & Multi-controlled U & $14(k-1)+G[CU]$ & $k-1$ clean & \cite{Jennings2025LBE} \\
			\newtag{8}{row:CH} & $CH$ & Controlled Hadamard & 2 & 0 & \cite{Jennings2025LBE} \\
			\newtag{9}{row:Pk} & $C\mathcal{P}_k$ & Controlled $\mathcal{P}_k$ & 0 & 0 & Appx \ref{sec:Pk circuit} \\
			\newtag{10}{row:CKab} & $CK^{(a,b)}$ & Controlled Commutation & $7\log a \log b$ & 0 & Appx \ref{sec:Circ for Commutation} \\
			\newtag{11}{row:barB2q} & $C\overline{B}_{2,q}$ & Controlled $\overline{B}_{2,q}$ & $7\log n$ & 0 & Appx \ref{sec:Circuit for B2q} \\
			\newtag{12}{row:barB3q} & $C\overline{B}_{3,q}$ & Controlled $\overline{B}_{3,q}$ & $21 \log n + 2\log Q$ & 0 & Appx \ref{sec:Circuit for B3q} \\
		\end{tabular}
	\end{center}
	\caption{A list of relevant gate costs used in the resource estimation. For each of the multi-controlled gates we assume $k>2$. Here, $\log^* $ is the iterated logarithm base $2$. Note, the T counts for rows \ref{row:Pk}--\ref{row:barB3q} use controlled versions of their respective sources.}
	\label{tbl: Res Est Gates}
\end{table}

\subsection{Resource Estimate for the PREP and SELECT Block Encoding Oracles} \label{sec:PREP and SELECT}
The goal of this section is to analyze the T cost of loading the non-unitary matrix $L^{(\text{e})}$ using the PREP and SELECT procedures for the lattices D1Q3*, D2Q9* and D3Q15* (see Appendix \ref{sec:Velocity Sets Embed} for lattice details). First, a minor modification of the PREP and SELECT oracles is required since our decomposition relies on the LCNU approach from Section \ref{sec:LinComBlock} as opposed to the standard LCU. Using the finding of \cite{GS2025_LCofThings}, suitable oracles for our setting can be achieved by including a single ancilla qubit for SELECT and leaving PREP unchanged. Specifically, using the LCU of the form \eqref{eqn:LCU}, the oracles to load $U_L$ are
\begin{equation} \label{eqn:PREP SELECT}
\begin{split}
	\text{PREP} \ket{0}_a &= \frac{1}{\sqrt{c}}\sum_{l=1}^{N_s} \sqrt{\abs{c_l}}\ket{l}_a ,\\
	\text{SELECT} \ket{l}_a \ket{0} \ket{\psi} &= \ket{l}_a U_l \ket{0} \ket{\psi} \\
	&= \ket{l}_a\ket{0}L_l\ket{\psi} + \ket{l}_a\ket{1}L_l^\perp\ket{\psi} ,
\end{split}
\end{equation}
where the subscript $a$ refers to an ancilla register with $\log N_s$ qubits, the coefficients $c_l$ are from \eqref{eqn:LCNU}, $c=\sum_{l=1}^{N_s} \abs{c_l}$ and $\ket{\psi}$ is an arbitrary state. Then, by applying the PREP, SELECT and $\text{PREP}^\dagger$ oracles we load the unitary matrix $\begin{pmatrix} U_L/c & * \\ * & * \end{pmatrix}$ where the asterisk entries represent arbitrary blocks of appropriate size. 

Assuming the worst case where no structure exists between the $c_l$ coefficients, the T cost of the PREP oracle is $\mathcal{O}(N_s\log \frac{1}{\epsilon})$ for desired accuracy $\epsilon$, which is the cost to implement an arbitrary $N_s \times N_s$ unitary matrix. From Section \ref{sec:LBE Circuits}, we find that $N_s = \alpha(\alpha+1)(2N_{E_x}+2N_{E_y}+2N_{E_z}+N_{R}) + 2(\alpha-1)^2N_\Gamma + 3$. Using $\alpha=4$, $\epsilon=10^{-12}$ and the relevant quantities from Table \ref{tbl:Pauli SVD Decompositions}, we find that the PREP oracle requires $\mathcal{O}(10^{4})$, $\mathcal{O}(10^{6})$ and $\mathcal{O}(10^{6})$ T gates for the D1Q3*, D2Q9* and D3Q15* cases, respectively. As will be demonstrated in the following subsections, these costs are negligible relative to the T cost of the SELECT oracles.

The SELECT oracle is implemented by controlling each term in the LCU on the $\log N_s$ ancilla qubits from the PREP oracle. First, define the function $\mathcal{U} \colon \mathbb{C}^{N \times N} \to \mathbb{C}^{2N \times 2N}$ such that $\mathcal{U}(L)=\begin{pmatrix} L & L^\perp \\ L^\perp & L \end{pmatrix}$ following the constructions of Lemma \ref{Lemma for Ul}. Then, following Section \ref{sec:LBE Circuits}, we consider the LCU
\begin{equation*}
	U_{L^{(\text{e})}} = \sum_{i=1}^3 \mathcal{U}(L_{1,i})
	+ \sum_{\lambda \in \Lambda_1} \mathcal{U}(L_{\lambda}^\text{lin,1})
	+ \sum_{\lambda \in \Lambda_2} \mathcal{U}(L_{\lambda}^\text{lin,2})
	+ \sum_{\lambda \in \Lambda_3} \mathcal{U}(L_{\lambda}^\text{nlin}) .
\end{equation*}
With this expression, the total cost to implement the SELECT oracle for $L^{(\text{e})}$ is
\begin{equation} \label{eqn:Cost Le Eq}
	\begin{split}
		G[C^{\log N_s} U_{L^{(\text{e})}}] 
		&=  14N_s(\log N_s -1) 
		+ \sum_{i=1}^3 G[C\mathcal{U}(L_{1,i})]
		+ \sum_{\lambda \in \Lambda_1} G[C\mathcal{U}(L_{\lambda}^\text{lin,1})] \\
		&+ \sum_{\lambda \in \Lambda_2} G[C\mathcal{U}(L_{\lambda}^\text{lin,2})] 
		+ \sum_{\lambda \in \Lambda_3} G[C\mathcal{U}(L_{\lambda}^\text{nlin})] ,
	\end{split}
\end{equation}
where we have used row \ref{row:CkU} of Table \ref{tbl: Res Est Gates}. In the remainder of this section we explicitly calculate the contributions of each term on the RHS of \eqref{eqn:Cost Le Eq}.

\subsubsection{T Count for $L_1^{\textnormal{(e)}}$ Terms} \label{sec:Tcount L1e}
By applying \eqref{eqn:Cost Ul} to the controlled $L_{1,i}$ terms of \eqref{eqn:Cost Le Eq}, we have $G[C\mathcal{U}(L_{1,i})] = G[C\overline{L}_{1,i}] + G[C^{\mathcal{T}(L_{1,i})+1} X]$ for $i = 1,2,3$. From the completions in \eqref{eqn: L1e completions} and using rows \ref{row:CX} and \ref{row:CInc} from Table \ref{tbl: Res Est Gates}, we find
\begin{equation} \label{eqn:Cost U2_L1}
	\sum_{i=1}^3 G[C{\overline{L}_{1,i}}] 
	= \log n_t G[CX] + G[CS_{+1}^{n_t}]
	= 12\log 2n_t .
\end{equation}
Next, with \eqref{eqn:L1eL1eT} we find that $\mathcal{T}(L_{1,1})=0$, $\mathcal{T}(L_{1,2})=\log n_t$, and $\mathcal{T}(L_{1,3})=0$. This gives
\begin{equation} \label{eqn:Cost U1_L1}
	\sum_{i=1}^3 G[C^{\mathcal{T}(L_{1,i})+1} X] 
	= 2G[CX] + G[C^{\log n_t+1}X]
	= 8\log n_t - 4 ,
\end{equation}
where we have used rows \ref{row:CX} and \ref{row:CkX} from Table \ref{tbl: Res Est Gates}. Combining \labelcref{eqn:Cost U2_L1,eqn:Cost U1_L1} yields
\begin{equation} \label{eqn:Cost L1e}
	\sum_{i=1}^3 G[C\mathcal{U}(L_{1,i})] = 20\log n_t + 8 .
\end{equation}

\subsubsection{T Count for $L_\textnormal{lin,1}^\textnormal{(e)}$ Terms} \label{sec:Tcount Llin2}
Next, from \labelcref{eqn:Cost Ul,eqn:Llin1b}, we see that $G[C\mathcal{U}(L_{\lambda}^\text{lin,1})] = G[C\overline{L_{\lambda}^\text{lin,1}}] + G[C^{\mathcal{T}(L_{\lambda}^\text{lin,1})+1} X]$ for a given $\lambda \in \Lambda_1$. From \eqref{eqn:Llin1 bar}, we see that $G[C\overline{L_{\lambda}^\text{lin,1}}] = G[CS_{\eta,p,m}]$. Using the definition of $S_{\eta,p,m}$ from \eqref{eqn:S_etapm} and assuming $N_{E_x}=N_{E_y}=N_{E_z}$, we find
\begin{equation} \label{eqn:Cost U2_Llin2}
	\begin{split}
		\sum_{\lambda\in\Lambda_1} G[C\overline{L_{\lambda}^\text{lin,1}}]
		&= \sum_{\lambda\in\Lambda_1} G[CS_{\eta,p,m}] \\
		&= \alpha(\alpha+1) N_{E_x} \sum_{p\in\{+1,-1\}} \left( G[CS_p^{n_x}] + G[CS_p^{n_y}] + G[CS_p^{n_z}] \right) \\
		&= 72N_{E_x}\alpha(\alpha+1)\log 2n_x ,
	\end{split}
\end{equation}
where the third line uses row \ref{row:CInc} from Table \ref{tbl: Res Est Gates} and $n_x=n_y=n_z$. Next, from \eqref{eqn:Llin1 Llin1T}, we can see that $\mathcal{T}(L_{\lambda}^\text{lin,1}) = (i-1)\log n_t + (\alpha-j)\log Qn + 1$, and so we have
\begin{equation} \label{eqn:Cost U1_Llin2}
\begin{split}
	\sum_{\lambda\in\Lambda_1}
	G[C^{\mathcal{T}(L_{\lambda}^\text{lin,1})+1} X]
	&= \sum_{\lambda\in\Lambda_1} (8\mathcal{T}(L_{\lambda}^\text{lin,1})-4) \\
	&= 8N_{E_x}\alpha(\alpha+1) \left[ 2(\alpha-1)\log Qn + 3\log n_t + 3 \right] .
\end{split}
\end{equation}
Combining \labelcref{eqn:Cost U2_Llin2,eqn:Cost U1_Llin2} yields
\begin{equation} \label{eqn:Cost Llin2}
	\sum_{\lambda \in \Lambda_1} G[C\mathcal{U}(L_{\lambda}^\text{lin,1})]
	= 8N_{E_x}\alpha(\alpha+1) \left[ 2(\alpha-1)\log Qn + 3(\log n_t + 3\log2n_x +1) \right] .
\end{equation}

\subsubsection{T Count for $L_\textnormal{lin,2}^\textnormal{(e)}$ Terms} \label{sec:Tcount Llin1}
Next, using \eqref{eqn:Cost Ul} with the $L_{\lambda}^\text{lin,2}$ terms \eqref{eqn:Llin2b} for $\lambda\in\Lambda_2$, we have $G[C\mathcal{U}(L_{\lambda}^\text{lin,2})] = G[C\overline{L_{\lambda}^\text{lin,2}}] + G[C^{\mathcal{T}(L_{\lambda}^\text{lin,2})+1} X]$. From the completion in \eqref{eqn:Llin2 bar}, and given that the T cost for any controlled Pauli gate is zero, we find that
\begin{equation} \label{eqn:Cost U2_Llin1}
	\sum_{\lambda \in \Lambda_2} G[C\overline{L_{\lambda}^\text{lin,2}}]
	= N_R \alpha(\alpha+1)( G[CW_R] + G[CV_R] ) ,
\end{equation}
where $G[CW_R]$ and $G[CV_R]$ are found numerically and provided in Table \ref{tbl:Pauli SVD Decompositions}.

Next, from \eqref{eqn:Llin2 Llin2T}, we can see that $\mathcal{T}(L_{\lambda}^\text{lin,2})=(i-1)\log n_t + (\alpha-j)\log Qn + 1$. So, using row \ref{row:CkX} from Table \ref{tbl: Res Est Gates}, we have
\begin{equation} \label{eqn:Cost U1_Llin1}
	\begin{split}
		\sum_{\lambda \in \Lambda_2}  
		G[C^{\mathcal{T}(L_{\lambda}^\text{lin,2}) + 1} X] 
		&= \sum_{\lambda \in \Lambda_2} 8(\mathcal{T}(L_{\lambda}^\text{lin,2}) + 1) - 12 \\
		&= \frac{4}{3}N_{R}\alpha(\alpha+1) 
		\left[ 2(\alpha-1) \log Qn + 3\log n_t + 3 \right] .
	\end{split}
\end{equation}
Combining \labelcref{eqn:Cost U1_Llin1,eqn:Cost U2_Llin1} yields
\begin{equation} \label{eqn:Cost Llin1}
	\sum_{\lambda \in \Lambda_2} G[C\mathcal{U}(L_{\lambda}^\text{lin,2})] 
	= \frac{4}{3}N_{R}\alpha(\alpha+1) 
	\left[ 2(\alpha-1) \log Qn + 3\log n_t + \frac{3}{4}\left( G[CW_R] + G[CV_R] \right) + 3\right] .
\end{equation}

\subsubsection{T Count for $L_\textnormal{nlin}^\textnormal{(e)}$ Terms} \label{secTcount nonlin}
Next, using \eqref{eqn:Cost Ul} with the $L_{\lambda}^\text{nlin}$ terms \eqref{eqn:Lnonlinb} for $\lambda\in\Lambda_3$, we have $G[C\mathcal{U}(L_{\lambda}^\text{nlin})] = G[C\overline{L_{\lambda}^\text{nlin}}] + G[C^{\mathcal{T}(L_{\lambda}^\text{nlin})+1} X]$. We make the assumption that $G[CW_{\Gamma_q}] = G[CW_\Gamma]$ and $G[CV_{\Gamma_q}] = G[CV_\Gamma]$ for all $q \in \{1,\dots,Q\}$, where $G[CW_\Gamma]=\max\{G[CW_{\Gamma_q}] \;|\; q \in \{1,\dots,Q\}\}$ and $G[CV_\Gamma]=\max\{G[CV_{\Gamma_q}] \;|\; q \in \{1,\dots,Q\}\}$. Then, from the completion in \eqref{eqn:Lnonlin bar}, we have
\begin{equation} \label{eqn:Cost Lnonlin Completion}
\begin{split}
	G[C\overline{L_{\lambda}^\text{nlin}}]
	&= G[C\mathcal{P}_k] + G[CK^{((Qn)^l,Qn)}] + G[C\overline{B}_{k,q}]
	+ G[CW_{\Gamma}] + G[CV_{\Gamma}] + G[CK^{((Qn)^k,(Qn)^l)}] \\
	&= 7l(\log Qn)^2(k+1) + 7(k+1)\log n + 2(k-2)\log Q + G[CW_{\Gamma}] + G[CV_{\Gamma}],
\end{split}
\end{equation}
where the second equality is the result of using rows \ref{row:Pk}, \ref{row:CKab}, \ref{row:barB2q} and \ref{row:barB3q} from Table \ref{tbl: Res Est Gates}. Note, $G[CW_{\Gamma}]$ and $G[CV_{\Gamma}]$ are found numerically and provided in Table \ref{tbl:Pauli SVD Decompositions}. This results in
\begin{equation} \label{eqn:Cost U2_nonlin}
\begin{split}
	\sum_{\lambda\in\Lambda_3} G[C\overline{L_{\lambda}^\text{nlin}}] 
	&= \frac{1}{3}N_\Gamma(\alpha-1)\Big[ 7(\alpha-2)(7\alpha-12)(\log Qn)^2  \\
	&\quad + 3(\alpha-1)\left( 28\log n + 2\log Q + 2G[CW_{\Gamma}] + 2G[CV_{\Gamma}] \right) \Big] ,
\end{split}
\end{equation}
where $N_\Gamma \coloneq \sum_{q=1}^Q N_{\Gamma_q}$. Next, from \eqref{eqn:Lnonlin LnonlinT}, we can see that $\mathcal{T}(L_{\lambda}^\text{nlin}) = (i-1)\log n_t + (\alpha-j)\log Qn + 1$, and so using row \ref{row:CkX} from Table \ref{tbl: Res Est Gates}, we have
\begin{equation} \label{eqn:Cost U1_nonlin}
	\begin{split}
		\sum_{\lambda\in\Lambda_3}
		G[C^{\mathcal{T}(L_{\lambda}^\text{nlin})+1} X]
		&= \sum_{\lambda\in\Lambda_3} 8(\mathcal{T}(L_{\lambda}^\text{nlin})+1) - 12 \\
		&= \frac{8}{3}N_\Gamma(\alpha-1) 
		\big[ 2(\alpha^2+\alpha-3)\log Qn + 3(\alpha-1)(\log n_t +1) \big] .
	\end{split}
\end{equation}
Combining \labelcref{eqn:Cost U1_nonlin,eqn:Cost U2_nonlin} yields
\begin{equation} \label{eqn:Cost Lnonlin}
\begin{split}
	\sum_{\lambda \in \Lambda_3} G[C\mathcal{U}(L_{\lambda}^\text{nlin})]
	&=
	\frac{1}{3}N_\Gamma (\alpha-1)
	\Big[ 7(\alpha-2)(7\alpha-12)(\log Qn)^2 + 16(\alpha^2+\alpha-3)\log Qn \\
	&\qquad\qquad + 3(\alpha-1)(56\log n + 8\log n_t + 4\log Q + 2G[CW_{\Gamma}] + 2G[CV_{\Gamma}] + 8) \Big] .
\end{split}
\end{equation}

We can now finally calculate the full T cost of the PREP and SELECT oracles. By inserting \labelcref{eqn:Cost L1e,eqn:Cost Llin1,eqn:Cost Llin2,eqn:Cost Lnonlin} into \eqref{eqn:Cost Le Eq}, we obtain the full expression
\begin{equation} \label{eqn:Cost Prep Select}
	\begin{split}
		G[C^{\log N_s} U_{L^{(\text{e})}}] 
		&= \frac{1}{3}N_\Gamma (\alpha-1)
		\Big[ 7(\alpha-2)(7\alpha-12)(\log Qn)^2 + 16(\alpha^2+\alpha-3)\log Qn \\
		&\qquad\qquad + 3(\alpha-1)(28\log n + 8\log n_t + 2\log Q + 2G[CW_{\Gamma}] + 2G[CV_{\Gamma}] + 8) \Big] \\
		&+ \alpha(\alpha+1) 
		\Big[ (16N_{E_x} + \sfrac{8}{3}N_{R})(\alpha-1)\log Qn + (24N_{E_x} + 4N_{R})\log n_t \\
		& \qquad\qquad + 72N_{E_x}\log n_x + 32N_{E_x} + 4N_{R} + G[CW_R] + G[CV_R] 
		\Big] \\
		& + 20\log n_t + 8 + 14N_s(\log N_s-1) .
	\end{split}
\end{equation}
Note that upper bounds for the 1D and 2D geometries may obtained by changing $n$ to $n_x$ and $n_xn_y$, respectively. Finally, this result is added to the T cost of the PREP oracle from Section \ref{sec:PREP and SELECT} and plotted in Figure \ref{fig:Cost Prep Select} for the D1Q3*, D2Q9* and D3Q15* cases using $N_{E_x}$, $N_R$, $N_\Gamma$, $G[CW_R]$, $G[CV_R]$, $G[CW_{\Gamma}]$ and $G[CV_{\Gamma}]$ from Table \ref{tbl:Pauli SVD Decompositions}. The leading order term for the D3Q15* case ($\frac{7}{3}N_\Gamma(\alpha-1)(\alpha-2)(7\alpha-12)(\log Qn_xn_yn_z)^2$) from \eqref{eqn:Cost Prep Select} is also included in Figure \ref{fig:Cost Prep Select} (dashed line) to demonstrate that the T gate cost is dominated by a single term. This leading term is attributed to the commutation matrices used in the nonlinear components. Therefore, either an improved circuit for the commutation matrix or a decrease in $N_\Gamma$ could result in large savings to the total T cost. Finally, we conclude that the T cost for PREP and SELECT scales like $\mathcal{O}(\alpha^3 Q^2 (\log n)^2)$ where we use $N_\Gamma \sim \mathcal{O}(Q^2)$ following from Table \ref{tbl:Pauli SVD Decompositions}.
		
\begin{figure}[h]
	\centering
	\includegraphics[width=0.75\linewidth]{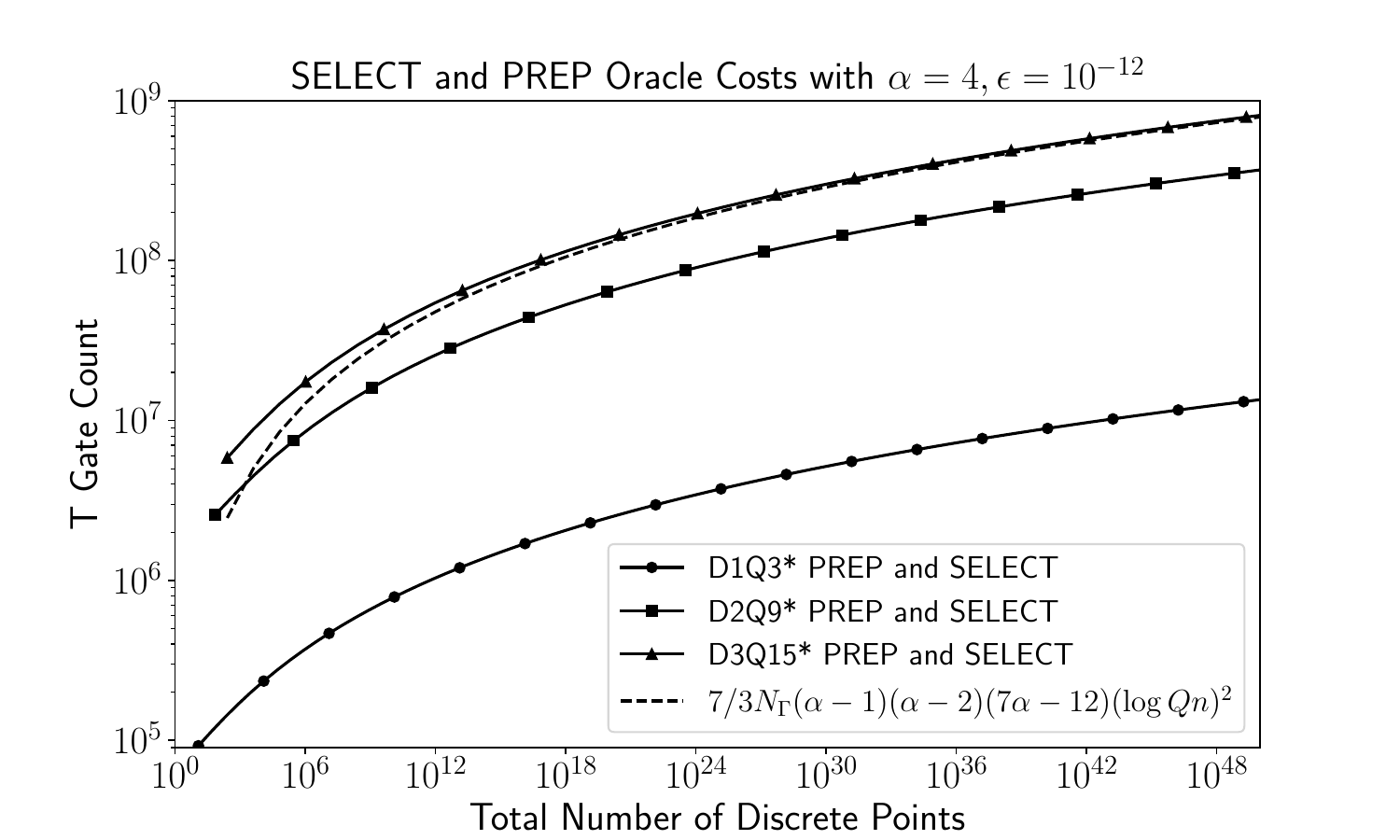}
	\caption{T gate count to encode the Carleman linearized LBE matrix using the PREP and SELECT oracles for the D1Q3*, D2Q9* and D3Q15* lattices. T count is plotted as a function of the total number of discrete points defined as $nn_tQ$. Each case uses \eqref{eqn:Cost Prep Select} with $n_t=n_x$ and the following: for D1Q3* $n=n_x$ and $Q=4$, for D2Q9* $n=n_xn_y$ and $Q=16$, and for D3Q15* $n=n_xn_yn_z$ and $Q=16$. We also use $\alpha=4$, $\epsilon=10^{-12}$ and $N_{E_x}$, $N_{E_y}$, $N_{E_z}$, $N_{R}$, and $N_{\Gamma}$ come from Table \ref{tbl:Pauli SVD Decompositions}. The leading term of \eqref{eqn:Cost Prep Select} is plotted in the dashed curve for the D3Q15* lattice.}
	\label{fig:Cost Prep Select}
\end{figure}


\subsection{Resource Estimate for VQLS} 
The VQLS method is a variational approach whereby an ansatz $V(\vec{\theta}_i)$ with variational parameters $\vec{\theta}_i$ is used to approximate the solution of the linear system $L\ket{x}=\ket{b}$ by searching the parameter space to find the optimal parameters $\vec{\theta}_\text{opt}$ such that $V(\vec{\theta}_\text{opt}) \ket{0} \approx \ket{x}$. To do this, the quantum computer is tasked with calculating classically intractable expectation values, which are then passed to a classical computer to (1) compute a cost function, and (2) update the variational parameters from $\vec{\theta}_i$ to $\vec{\theta}_{i+1}$ using an optimization routine. Given the updated variational parameters $\vec{\theta}_{i+1}$, the quantum computer recomputes the expectation values and again passes the information back to the classical computer to repeat (1) and (2). This back-and-forth process is repeated until the cost function converges to a predefined criteria at which point $\vec{\theta}_\text{opt}$ is obtained. An important note for the present study is that VQLS relies on the LCU data loading procedure.

The goal of this section is to calculate the T cost required to compute the expectation values per iteration by adopting the specific VQLS approach used in \cite{GS24,Demirdjian2025,Gnanasekaran2024ConstrainedOpt}, which is modified to work with the LCNU approach. To do this, first consider $L=\sum_{l=1}^{N_s} c_l L_l$ with $c_l\in\mathbb{C}$, $L,L_l \in \mathbb{C}^{2^{n_q} \times 2^{n_q}}$ where each $L_l$ is a specific non-unitary matrix of the form \eqref{eqn:Llgeneral}. Then from \cite{Bravo2023}, the local cost function is 
\begin{equation*}
	\begin{split}
		C(\vec{\theta}_i) &= \frac{1}{2} \left(1 - \frac{1}{n_q} 
		\frac{\sum_{r=1}^{n_q}\sum_{l,l^\prime=1}^{N_s} c_l c_{l^\prime}^* \delta_{l,l^\prime}^r}
		{\sum_{l,l^\prime=1}^{N_s} c_l c_{l^\prime}^* \beta_{l,l^\prime}} \right) \\
		\delta_{l,l^\prime}^r &= \bra{V(\vec{\theta}_i)} L_{l^\prime}^\dagger U_b Z_r U_b^\dagger L_l \ket{V(\vec{\theta}_i)} \\
		\beta_{l,l^\prime} &= \bra{V(\vec{\theta}_i)} L_{l^\prime}^\dagger L_l \ket{V(\vec{\theta}_i)} ,
	\end{split}
\end{equation*}
where $U_b\ket{0}=\ket{b}$ and $Z_r$ is the Pauli-Z matrix applied on the $r\text{th}$ qubit. In the case that $\delta_{l,l^\prime}^r$ and $\beta_{l,l^\prime}$ are real valued, then $N_s^2(n_q+1)$ circuits are required to evaluate the cost function. From Section \ref{sec:PREP and SELECT}, we find that $N_s = \alpha(\alpha+1)(2N_{E_x}+2N_{E_y}+2N_{E_z}+N_{R}) + 2(\alpha-1)^2N_\Gamma + 3$. The $\delta_{l,l^\prime}^r$ and $\beta_{l,l^\prime}$ expectation values are calculated using a modified Hadamarad test, as introduced in \cite{GS24}, where an additional ancilla qubit is required to embed the non-unitary $L_l$ matrices into the unitary matrices $\mathcal{U}(L_l)$, using the $\mathcal{U}$ function from Section \ref{sec:PREP and SELECT}. The Hadamard test requires that each gate of the $\mathcal{U}(L_l)$ circuits be controlled on an ancilla qubit, transforming their T costs from $G[\mathcal{U}(L_l)]$ to $G[C\mathcal{U}(L_l)]$. 

To calculate the T cost, consider that $\delta_{l,l^\prime}^r$ and $\beta_{l,l^\prime}$ are independent for different values of $l$, $l^\prime$ and $r$, and so the circuits may be run in parallel. The T cost can therefore be characterized by the most expensive circuit for any choice of $l$,$l^\prime$ and $r$, which, following Section \ref{sec:PREP and SELECT}, comes from the $\mathcal{U}(L_\lambda^\text{nlin})$ term for $\lambda\in\Lambda_3$. Specifically, the most expensive circuit is obtained by letting $\lambda = (3,2,\alpha-2,\alpha-3,q,m)$ for any $q\in\{1,\dots,Q\}$ and $m\in \{1,\dots,N_{\Gamma_q}\}$ since this value maximizes the cost of the commutation matrices. The T cost for this circuit is 
\begin{equation} \label{eqn:Cost Max VQLS}
\begin{split}
	G[C\mathcal{U}(L_{\lambda})] 
	&= G[C\overline{L_{\lambda}^\text{nlin}}] + G[C^{\mathcal{T}(L_{\lambda}^\text{nlin})+1}X] \\
	&=28(\alpha-2)(\log Qn)^2 + 28\log n + 2\log Q + 8\log n_t + 24\log Qn + 4 + G[CW_\Gamma] + G[CV_\Gamma] ,
\end{split}
\end{equation}
where the first equality is obtained using \eqref{eqn:Cost Ul}, and the second is obtained by evaluating the $\lambda = (3,2,\alpha-2,\alpha-3,q,m)$ value into both \labelcref{eqn:Cost Lnonlin Completion,eqn:Lnonlin LnonlinT} and combining the results. 

Both the total number of circuits and costliest circuit \eqref{eqn:Cost Max VQLS} are plotted in Figure \ref{fig:Cost VQLS}. Here, we find that while the maximum T cost is relatively shallow, the number of circuits is quite high since it depends quadratically on $N_s$. The reason that the D2Q9* and D3Q15* cases are approximately equal is because the $N_s$ parameter depends indirectly on $Q$ (through $N_{E_x}$, $N_{E_y}$, $N_{E_z}$, $N_{R}$ and $N_\Gamma$). There are two potential ways to reduce the number of circuits: (1) reduce the $N_{E_x}$, $N_{E_y}$, $N_{E_z}$, $N_{R}$ and $N_\Gamma$ parameters using different decompositions for their associated terms, and (2) change the cost function to eliminate the quadratic dependence on $N_s$. 

\begin{figure}[!htbp]
	\centering
	\includegraphics[width=0.75\linewidth]{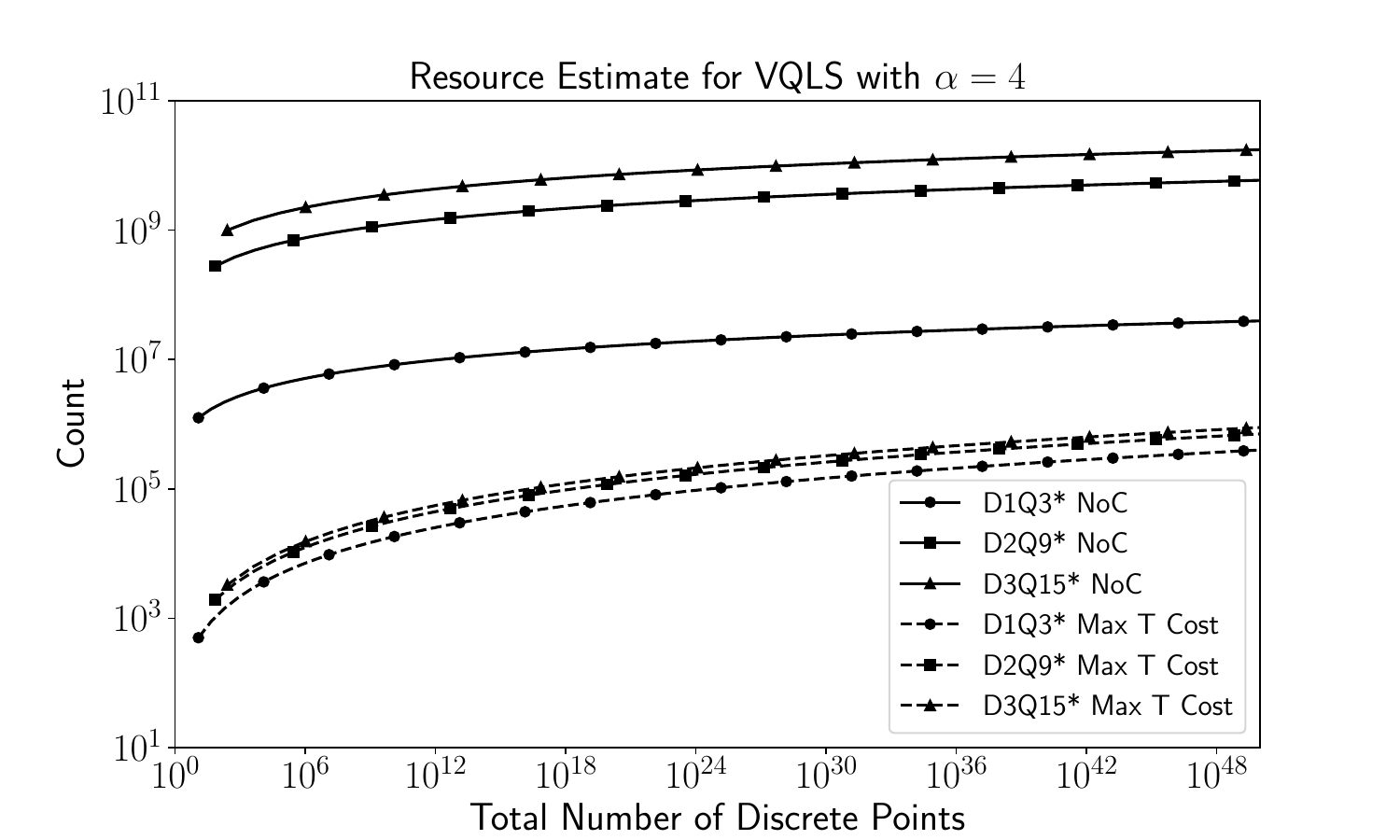}
	\caption{The number of circuits (NoC) and maximum T cost per circuit to encode the Carleman linearized LBE for the D1Q3*, D2Q9* and D3Q15* lattices. The total number of discrete points on the x-axis is defined as $nn_tQ$ where for D1Q3* $n=n_x$ and $Q=4$, for D2Q9* $n=n_xn_y$ and $Q=16$, and for D3Q15* $n=n_xn_yn_z$ and $Q=16$. Here, we use $n_t=n_x=n_y=n_z$, $\alpha=4$ and the values for $N_{E_x}$, $N_{E_y}$, $N_{E_z}$, $N_{R}$, and $N_{\Gamma}$ required to calculate the NoC for each case come from Table \ref{tbl:Pauli SVD Decompositions}.}
	\label{fig:Cost VQLS}
\end{figure}


\section{Conclusions and Discussion}
\twocolumngrid
	
In this work, we investigate the question of how to efficiently load an exponentially sized sparse and structured matrix onto a quantum computer, a critical bottleneck in obtaining a quantum advantage for solving ODEs and PDEs. First, we build upon the existing LCNU literature \cite{GS24,GS2025_LCofThings,Demirdjian2025,Gnanasekaran2024ConstrainedOpt,Surana2025ConstrainedOpt} by introducing a new class of matrices to be used in constructing decompositions whose basic form is described in equation \eqref{eqn:Llgeneral}. The upshot of our generalization is that non-unitary elements of this class admit straightforward unitary completions as demonstrated by Theorem \ref{thm:Theorem for Lbar}. Consequently, each term of our proposed LCNU can be easily embedded into a unitary matrix to obtain an LCU with the same number of terms. Algorithms \labelcref{alg:Ul1,alg:Ul2} outline a procedure to construct the embedded LCNU terms. Next, we introduce both a procedure to zero pad an arbitrary Carleman linearized system and an efficient LCNU decomposition for said system. After embedding the terms of this LCNU into unitary matrices, we investigate the efficiency of the resulting LCU when applied to the 3D Carleman linearized LBE. We find that the number of terms in this LCU scales like $N_s \sim (\alpha^2 Q^2)$, where $\alpha$ is the truncation order used in the Carleman linearization procedure, $Q$ is the number of discrete velocities from the LBE, and we have assumed $N_\Gamma \sim \mathcal{O}(Q^2)$ following Table \ref{tbl:Pauli SVD Decompositions}. Importantly, $N_s$ is completely independent of both the number of temporal and spatial discretization points. Next, we compose explicit circuits for each type of term in the LCU, demonstrating that the circuits are easy to construct. Finally, we perform a resource estimation of the T gate cost for data loading via our LCU using (1) the PREP and SELECT oracles as part of the fault-tolerant block encoding strategy, and (2) the VQLS approach for a NISQ implementation. In the former, we find that the T cost scales like $\mathcal{O}(\alpha^3 Q^2 (\log n)^2)$ where $n=n_x n_y n_z$ and $n_x$, $n_y$ and $n_z$ are the number of spatial grid points in the $x$, $y$ and $z$ dimensions, respectively. In the VQLS approach, we find that $N_s^2(\log (2n_tn^\alpha)+1)$ circuits, for $n_t$ time steps, are required with a worst case T cost of $\mathcal{O}(\alpha (\log Qn)^2)$. Therefore, in both cases, our data loading strategy has a polylogarithmic dependence on the number of discretization points, providing a viable strategy toward quantum advantage for solving real world fluid dynamics problems. 

A reasonable follow up question to ask is, does our Carleman linearized LBE decomposition strategy use fewer resources than a Pauli decomposition? To answer this question, we take the product between the number of terms in the decomposition and the gate count per term as our metric. For this comparison, we consider the D1Q3* lattice for the smallest possible case of $n_x=2$, $n_t=2$, $\alpha=4$ and the lattice parameters from Table \ref{tbl:Pauli SVD Decompositions}. These parameters are used to construct $L^{(\text{e})}$ from \eqref{eqn:Le Matrix}, which in this case is sized $2^{14} \times 2^{14}$. Beginning with the Pauli decomposition, we use IBM's \texttt{Qiskit} function \texttt{SparsePauliOp} (\cite{Javadi2024}) with an error tolerance of $10^{-8}$ resulting in $17,044,182$ terms in the decomposition. Since each term is implemented with tensor products of $\log(2 n_t (Qn)^\alpha)$ Pauli matrices, this results in $17,044,182 \cdot \log(2 n_t (Qn)^\alpha) \sim \mathcal{O}(10^8)$ Pauli gates. Note, this is an upper bound since we do not distinguish identity gates from Pauli gates. In contrast, our decomposition method requires only $654$ terms in the decomposition with a worst case circuit depth of $\mathcal{O}(\alpha(\log Qn)^2)$ from \eqref{eqn:Cost Max VQLS} resulting in just $\mathcal{O}(10^4)$ total T gates. Note, this is an upper bound since many of the circuits require substantially fewer T gates than the worst case used here. Though this is not a direct comparison, since we are comparing T gate counts to Pauli gate counts, our method does result in a four order of magnitude improvement when compared with the Pauli decomposition. Furthermore, the number of terms in our method scales independently with the number of spatial and temporal discretization points. Whereas, the number of terms in the Pauli decomposition has an upper bound that scales like $4^{N}$, where $N=2n_t (Qn)^\alpha$ is the matrix dimension, i.e. an exponential dependence on the number of spatial and temporal discretization points. In summary, our decomposition method vastly outperforms the Pauli decomposition for the Carleman linearized LBE. 

We emphasized that while an efficient decomposition of the Carleman linearized LBE is a necessary condition to solve said system with quantum advantage, it is not a sufficient condition. A number of additional challenges in the workflow must also be addressed before speed-up can be claimed. For example, QLSAs typically have a linear or worse dependence on the matrix condition number, and the zero padding method may lead to a substantial increase in condition number as found in Appendix \ref{sec:Condition Number}. Therefore, a matrix preconditioner is likely required to solve the exponentially sized linear systems of interest \cite{Lapworth2025Preconditioner,Jin2025Preconditioner,Hosaka2023Preconditioning,Golden2022Preconditioning}. Another challenge is the readout problem. In essence, even if it is possible to solve an exponentially sized linear system with advantage, its size prohibits a full solution readout and, therefore, only a small subset of the solution may be obtained. Additionally, if one is to consider using a variational method to solve the linear system (like VQLS), then other challenges sprout up. The most well-known of which is the barren plateau \cite{Larocca2025BarrenPlateau}, though exciting techniques like quantum multigrid methods \cite{Pool2024Multigrid,Keller2024Multigrid} could provide a solution. 

We close this article with a discussion on both some limitations of our method and how it can be improved upon. First and foremost, for simplicity we have used the backward Euler method, which is only first order accurate. Thankfully, the flexibility of our approach permits easy construction of sparse block matrices, and so higher order discretization schemes (both spatial and temporal) should add only minimal overhead. Next, we note that we have not made any effort to transpile our circuits, and, therefore, it is likely that our T cost estimates will be substantially reduced when implemented with a transpiler. In the opposite direction, we have also assumed an all-to-all connectivity, and so there will be an added overhead when considering other architectures. We also note that further improvements are likely possible if the $\Sigma_{\Gamma_q}$ matrices (diagonal matrix from the SVD in \eqref{eqn:Gamma Decomp}) are decomposed using elements from $\mathbb{P}_\rho$ rather than from $\mathbb{P}_\sigma$, something that we will explore in future work. Lastly, we would like to highlight that our LCNU approach from Section \ref{sec:LinComBlock} is best suited for the decomposition of sparse, structured matrices, and we believe it is likely useful for constructions well beyond Carleman linearized systems. 

\section{Acknowledgments}
We gratefully acknowledge the support of the Naval Research Laboratory's Base Program titled Simulating Fluid Dynamics using the Semi-Lagrangian Method on a Quantum Computer (PE0601153N).

\section{Author contribution statements}
R.D. and T.H. contributed with the conceptualization, investigation, methodology and writing of this work; A.G. and A.S. contributed with the conceptualization, investigation and methodology of this work; D.G. contributed with the conceptualization of this work.  

\section{Data Availability}
The data that support the findings of this article are not publicly available because of legal restrictions preventing unrestricted public distribution. The data are available upon request.


\onecolumngrid
\bibliography{bibliography}

@inproceedings{GS24,
	title = {Efficient Variational Quantum Linear Solver for Structured Sparse Matrices},
	url = {http://dx.doi.org/10.1109/QCE60285.2024.00033},
	DOI = {10.1109/qce60285.2024.00033},
	booktitle = {2024 IEEE International Conference on Quantum Computing and Engineering (QCE)},
	publisher = {IEEE},
	author = {Gnanasekaran,  Abeynaya and Surana,  Amit},
	year = {2024},
	month = sep,
	pages = {199–210}
}

@article{Demirdjian2025,
	title = {Efficient decomposition of the Carleman linearized Burgers' equation},
	author = {Demirdjian, Reuben and Hogancamp, Thomas and Gunlycke, Daniel},
	journal = {Phys. Rev. A},
	year = {2026},
	month = {Jan},
	publisher = {American Physical Society},
	doi = {10.1103/g27q-r2gk},
	url = {https://link.aps.org/doi/10.1103/g27q-r2gk}
}

@article{Liu2021,
  title={Efficient quantum algorithm for dissipative nonlinear differential equations},
	author={Liu, Jin-Peng and Kolden, Herman {\O}ie and Krovi, Hari K and Loureiro, Nuno F and Trivisa, Konstantina and Childs, Andrew M},
	journal={Proceedings of the National Academy of Sciences},
	volume={118},
	number={35},
	pages={e2026805118},
	year={2021},
	publisher={National Academy of Sciences}
}

@misc{Gidney2017,
	author = "Craig Gidney",
	title = "Simple Algorithm for Multiplicative Inverses mod $2^n$",
	url  = "https://algassert.com/post/1709",
	addendum = "(accessed: 11.26.2024)",
	year = "2017"
}

@article{GS2025_LCofThings,
	title={Efficient quantum access model for sparse structured matrices using linear combination of “things”},
	author={Gnanasekaran, Abeynaya and Surana, Amit},
	journal={Physical Review A},
	volume={113},
	number={2},
	pages={022437},
	year={2026},
	publisher={APS}
}

@book{Kruger2017,
	title={The lattice Boltzmann method},
	author={Kr{\"u}ger, Timm and Kusumaatmaja, Halim and Kuzmin, Alexandr and Shardt, Orest and Silva, Goncalo and Viggen, Erlend Magnus},
	volume={10},
	number={978-3},
	year={2017},
	publisher={Springer}
}

@article{Cruz2024,
	title={Shallow unitary decompositions of quantum Fredkin and Toffoli gates for connectivity-aware equivalent circuit averaging},
	author={MQ Cruz, Pedro and Murta, Bruno},
	journal={APL Quantum},
	volume={1},
	number={1},
	year={2024},
	publisher={AIP Publishing}
}

@online{Gidney2015Increment,
	ALTauthor = {Craig Gidney},
	title = {Constructing Large Increment Gates},
	month = {June},
	year = {2015},
	url = {https://algassert.com/circuits/2015/06/12/Constructing-Large-Increment-Gates.html}
}

@article{Xu2018,
	author = {Changqing Xu and Lingling He and Zerong Lin},
	title = {Commutation matrices and commutation tensors},
	journal = {Linear and Multilinear Algebra},
	volume = {68},
	number = {9},
	pages = {1721--1742},
	year = {2018},
	publisher = {Taylor \& Francis},
	doi = {10.1080/03081087.2018.1556242},
}

@article{Li2025,
	title={Potential quantum advantage for simulation of fluid dynamics},
	author={Li, Xiangyu and Yin, Xiaolong and Wiebe, Nathan and Chun, Jaehun and Schenter, Gregory K and Cheung, Margaret S and M{\"u}lmenst{\"a}dt, Johannes},
	journal={Physical Review Research},
	volume={7},
	number={1},
	pages={013036},
	year={2025},
	publisher={APS}
}

@article{Kay2018,
	title={Tutorial on the quantikz package},
	author={Kay, Alastair},
	journal={arXiv preprint arXiv:1809.03842},
	year={2018}
}

@article{Javadi2024,
	title={Quantum computing with Qiskit},
	author={Javadi-Abhari, Ali and Treinish, Matthew and Krsulich, Kevin and Wood, Christopher J and Lishman, Jake and Gacon, Julien and Martiel, Simon and Nation, Paul D and Bishop, Lev S and Cross, Andrew W and others},
	journal={arXiv preprint arXiv:2405.08810},
	year={2024}
}

@article{Jennings2025LBE,
	title={An end-to-end quantum algorithm for nonlinear fluid dynamics with bounded quantum advantage},
	author={Jennings, David and Korzekwa, Kamil and Lostaglio, Matteo and Ashworth, Richard and Marsili, Emanuele and Rolston, Stephen},
	journal={arXiv preprint arXiv:2512.03758},
	year={2025}
}

@article{Hariprakash2025,
	title={Strategies for simulating the time evolution of Hamiltonian lattice field theories},
	author={Hariprakash, Siddharth and Modi, Neel S and Kreshchuk, Michael and Kane, Christopher F and Bauer, Christian W},
	journal={Physical Review A},
	volume={111},
	number={2},
	pages={022419},
	year={2025},
	publisher={APS}
}

@article{Childs2012LCU,
	title={Hamiltonian simulation using linear combinations of unitary operations},
	author={Childs, Andrew M and Wiebe, Nathan},
	journal={arXiv preprint arXiv:1202.5822},
	year={2012}
}

@article{Khattar2025,
	title={Rise of conditionally clean ancillae for efficient quantum circuit constructions},
	author={Khattar, Tanuj and Gidney, Craig},
	journal={Quantum},
	volume={9},
	pages={1752},
	year={2025},
	publisher={Verein zur F{\"o}rderung des Open Access Publizierens in den Quantenwissenschaften}
}

@article{Selinger2013,
	title = {Quantum circuits of $T$-depth one},
	author = {Selinger, Peter},
	journal = {Phys. Rev. A},
	volume = {87},
	issue = {4},
	pages = {042302},
	numpages = {4},
	year = {2013},
	month = {Apr},
	publisher = {American Physical Society},
	doi = {10.1103/PhysRevA.87.042302},
	url = {https://link.aps.org/doi/10.1103/PhysRevA.87.042302}
}

@article{Bravo2023,
	title={Variational quantum linear solver},
	author={Bravo-Prieto, Carlos and LaRose, Ryan and Cerezo, Marco and Subasi, Yigit and Cincio, Lukasz and Coles, Patrick J},
	journal={Quantum},
	volume={7},
	pages={1188},
	year={2023},
	publisher={Verein zur F{\"o}rderung des Open Access Publizierens in den Quantenwissenschaften}
}

@article{Gunlycke2020,
	title={Efficient algorithm for generating Pauli coordinates for an arbitrary linear operator},
	author={Gunlycke, Daniel and Palenik, Mark C and Emmert, Alex R and Fischer, Sean A},
	journal={arXiv preprint arXiv:2011.08942},
	year={2020}
}

@article{Gnanasekaran2024ConstrainedOpt,
	author = {Gnanasekaran, Abeynaya and Surana, Amit and Zhu, Hongyu},
	doi = {10.2478/qic-2025-0014},
	url = {https://doi.org/10.2478/qic-2025-0014},
	title = {Variational Quantum Framework for Nonlinear PDE Constrained Optimization Using Carleman Linearization},
	journal = {Quantum Information \& Computation},
	number = {3},
	volume = {25},
	year = {2025},
	pages = {260--289}
}

@article{Surana2025ConstrainedOpt,
	title={Variational quantum framework for partial differential equation constrained optimization},
	author={Surana, Amit and Gnanasekaran, Abeynaya},
	journal={ACM Transactions on Quantum Computing},
	volume={7},
	number={1},
	pages={1--36},
	year={2025},
	publisher={ACM New York, NY}
}

@article{Pool2024Multigrid,
	title={Nonlinear dynamics as a ground-state solution on quantum computers},
	author={Pool, Albert J and Somoza, Alejandro D and Mc Keever, Conor and Lubasch, Michael and Horstmann, Birger},
	journal={Physical Review Research},
	volume={6},
	number={3},
	pages={033257},
	year={2024},
	publisher={APS}
}

@inproceedings{Keller2024Multigrid,
	author={Keller, Christo Meriwether and Eidenbenz, Stephan and Bärtschi, Andreas and O'Malley, Daniel and Golden, John and Misra, Satyajayant},
	booktitle={ISC High Performance 2024 Research Paper Proceedings (39th International Conference)}, 
	title={Hierarchical Multigrid Ansatz for Variational Quantum Algorithms}, 
	year={2024},
	volume={},
	number={},
	pages={1-11},
	doi={10.23919/ISC.2024.10528934}
}

@article{Larocca2025BarrenPlateau,
	title={Barren plateaus in variational quantum computing},
	author={Larocca, Martin and Thanasilp, Supanut and Wang, Samson and Sharma, Kunal and Biamonte, Jacob and Coles, Patrick J and Cincio, Lukasz and McClean, Jarrod R and Holmes, Zo{\"e} and Cerezo, Marco},
	journal={Nature Reviews Physics},
	pages={1--16},
	year={2025},
	publisher={Nature Publishing Group UK London}
}

@article{Childs2017ImprovedHHL,
	title = {Quantum Algorithm for Systems of Linear Equations with Exponentially Improved Dependence on Precision},
	volume = {46},
	ISSN = {1095-7111},
	DOI = {10.1137/16m1087072},
	number = {6},
	journal = {SIAM Journal on Computing},
	publisher = {Society for Industrial \& Applied Mathematics (SIAM)},
	author = {Childs,  Andrew M. and Kothari,  Robin and Somma,  Rolando D.},
	year = {2017},
	month = jan,
	pages = {1920–1950}
}

@article{Harrow2009HHL,
	title = {Quantum Algorithm for Linear Systems of Equations},
	volume = {103},
	ISSN = {1079-7114},
	DOI = {10.1103/physrevlett.103.150502},
	number = {15},
	journal = {Physical Review Letters},
	publisher = {American Physical Society (APS)},
	author = {Harrow,  Aram W. and Hassidim,  Avinatan and Lloyd,  Seth},
	year = {2009},
	month = oct 
}

@article{kramer2025polynomialization,
	title={Discovering polynomial and quadratic structure in nonlinear ordinary differential equations},
	author={Kramer, Boris and Pogudin, Gleb},
	journal={arXiv preprint arXiv:2502.10005},
	year={2025}
}

@article{Morales2024QLSAsurvey,
	title={Quantum linear system solvers: A survey of algorithms and applications},
	author={Morales, Mauro ES and Pira, Lirand{\"e} and Schleich, Philipp and Koor, Kelvin and Costa, Pedro and An, Dong and Aspuru-Guzik, Al{\'a}n and Lin, Lin and Rebentrost, Patrick and Berry, Dominic W},
	journal={arXiv preprint arXiv:2411.02522},
	year={2024}
}

@article{Dalzell2024QLSAshortcut,
	title={A shortcut to an optimal quantum linear system solver},
	author={Dalzell, Alexander M},
	journal={arXiv preprint arXiv:2406.12086},
	year={2024}
}

@book{Kowalski1991CarlemanBook,
	title={Nonlinear dynamical systems and Carleman linearization},
	author={Kowalski, Krzysztof and Steeb, Willi-hans},
	year={1991},
	publisher={World Scientific}
}

@article{Lin2022Carleman,
	title={Challenges for quantum computation of nonlinear dynamical systems using linear representations},
	author={Lin, Yen Ting and Lowrie, Robert B and Aslangil, Denis and Suba{\c{s}}{\i}, Yi{\u{g}}it and Sornborger, Andrew T},
	journal={arXiv preprint arXiv:2202.02188},
	year={2022}
}

@article{Gonzalez2025Carleman,
	title={Quantum Carleman linearization efficiency in nonlinear fluid dynamics},
	author={Gonzalez-Conde, Javier and Lewis, Dylan and Bharadwaj, Sachin S and Sanz, Mikel},
	journal={Physical Review Research},
	volume={7},
	number={2},
	pages={023254},
	year={2025},
	publisher={APS}
}

@article{Jennings2025Carleman,
	title={Quantum algorithms for general nonlinear dynamics based on the Carleman embedding},
	author={Jennings, David and Korzekwa, Kamil and Lostaglio, Matteo and Sornborger, Andrew T and Subasi, Yigit and Wang, Guoming},
	journal={arXiv preprint arXiv:2509.07155},
	year={2025}
}

@article{Li2025CarlemanLBE,
	title={Potential quantum advantage for simulation of fluid dynamics},
	author={Li, Xiangyu and Yin, Xiaolong and Wiebe, Nathan and Chun, Jaehun and Schenter, Gregory K and Cheung, Margaret S and M{\"u}lmenst{\"a}dt, Johannes},
	journal={Physical Review Research},
	volume={7},
	number={1},
	pages={013036},
	year={2025},
	publisher={APS}
}

@article{Aaronson2015HHLconditions,
	title={Read the fine print},
	author={Aaronson, Scott},
	journal={Nature Physics},
	volume={11},
	number={4},
	pages={291--293},
	year={2015},
	publisher={Nature Publishing Group UK London}
}

@article{Li2023SigmaBasis,
	title={Variational quantum algorithms for Poisson equations based on the decomposition of sparse Hamiltonians},
	author={Li, Hui-Min and Wang, Zhi-Xi and Fei, Shao-Ming},
	journal={Physical Review A},
	volume={108},
	number={3},
	pages={032418},
	year={2023},
	publisher={APS}
}

@article{Bae2024,
	title = {Hardware efficient decomposition of the Laplace operator and its application to the Helmholtz and the Poisson equation on quantum computer},
	volume = {23},
	ISSN = {1573-1332},
	DOI = {10.1007/s11128-024-04458-y},
	number = {7},
	journal = {Quantum Information Processing},
	publisher = {Springer Science and Business Media LLC},
	author = {Bae,  Jaehyun and Yoo,  Gwangsu and Nakamura,  Satoshi and Ohnishi,  Shota and Kim,  Dae Sin},
	year = {2024},
	month = jul 
}

@article{Liu2021SigmaBasis,
	title = {Variational quantum algorithm for the Poisson equation},
	author = {Liu, Hai-Ling and Wu, Yu-Sen and Wan, Lin-Chun and Pan, Shi-Jie and Qin, Su-Juan and Gao, Fei and Wen, Qiao-Yan},
	journal = {Phys. Rev. A},
	volume = {104},
	issue = {2},
	pages = {022418},
	numpages = {12},
	year = {2021},
	month = {Aug},
	publisher = {American Physical Society},
	doi = {10.1103/PhysRevA.104.022418},
	url = {https://link.aps.org/doi/10.1103/PhysRevA.104.022418}
}

@article{Kondo2021SigmaBasis,
	title={Computationally efficient quantum expectation with extended bell measurements},
	author={Kondo, Ruho and Sato, Yuki and Koide, Satoshi and Kajita, Seiji and Takamatsu, Hideki},
	journal={arXiv preprint arXiv:2110.09735},
	year={2021}
}

@article{Lapworth2025Preconditioner,
	title={Preconditioned block encodings for quantum linear systems},
	author={Lapworth, Leigh and S{\"u}nderhauf, Christoph},
	journal={Quantum Science and Technology},
	volume={10},
	number={4},
	pages={045064},
	year={2025},
	publisher={IOP Publishing}
}

@article{Jin2025Preconditioner,
	title={Quantum preconditioning method for linear systems problems via Schrodingerization},
	author={Jin, Shi and Liu, Nana and Ma, Chuwen and Yu, Yue},
	journal={arXiv preprint arXiv:2505.06866},
	year={2025}
}

@article{Hosaka2023Preconditioning,
	title={Preconditioning for a variational quantum linear solver},
	author={Hosaka, Aruto and Yanagisawa, Koichi and Koshikawa, Shota and Kudo, Isamu and Alifu, Xiafukaiti and Yoshida, Tsuyoshi},
	journal={arXiv preprint arXiv:2312.15657},
	year={2023}
}

@article{Golden2022Preconditioning,
	title={Quantum computing and preconditioners for hydrological linear systems},
	author={Golden, John and O’Malley, Daniel and Viswanathan, Hari},
	journal={Scientific Reports},
	volume={12},
	number={1},
	pages={22285},
	year={2022},
	publisher={Nature Publishing Group UK London}
}

@article{Demirdjian2022Variational,
	title = {Variational quantum solutions to the advection–diffusion equation for applications in fluid dynamics},
	volume = {21},
	ISSN = {1573-1332},
	url = {http://dx.doi.org/10.1007/s11128-022-03667-7},
	DOI = {10.1007/s11128-022-03667-7},
	number = {9},
	journal = {Quantum Information Processing},
	publisher = {Springer Science and Business Media LLC},
	author = {Demirdjian,  Reuben and Gunlycke,  Daniel and Reynolds,  Carolyn A. and Doyle,  James D. and Tafur,  Sergio},
	year = {2022},
	month = sep 
}

@article{Huang2025FourierReadout,
	title={Real and Fourier space readout methods: Comparison of complexity and applications to CFD problems},
	author={Huang, Xinchi and Nishi, Hirofumi and Kawada, Yoshifumi and Zushi, Tomofumi and Matsushita, Yu-ichiro},
	journal={arXiv preprint arXiv:2511.20017},
	year={2025}
}

@InProceedings{Nguyen2026QAlgo4PDEReview,
	author="Nguyen, Thanh",
	editor="Atulya K., Nagar
	and Dharm Singh, Jat
	and Durgesh Kumar, Mishra
	and Joshi, Amit",
	title="Quantum Algorithms for Partial Differential Equations: A Performance Review and Future Trajectories",
	booktitle="Intelligent Sustainable Systems",
	year="2025",
	publisher="Springer Nature Switzerland",
	address="Cham",
	pages="18--37",
	isbn="978-3-032-11524-9"
}

@article{Bharadwaj2025,
	title={Compact quantum algorithms for time-dependent differential equations},
	author={Bharadwaj, Sachin S and Sreenivasan, Katepalli R},
	journal={Physical Review Research},
	volume={7},
	number={2},
	pages={023262},
	year={2025},
	publisher={APS}
}

@inproceedings{Jennings2026CarlLBEApplied,
	title={Simulating non-trivial incompressible flows with a quantum lattice Boltzmann algorithm},
	author={Jennings, David and Korzekwa, Kamil and Lostaglio, Matteo and Mannix, Paul and Ashworth, Richard and Marsili, Emanuele and Rolston, Stephen},
	booktitle={AIAA SCITECH 2026 Forum},
	pages={1936},
	year={2026}
}

@article{Lewis2024Limitations,
	title={Limitations for quantum algorithms to solve turbulent and chaotic systems},
	author={Lewis, Dylan and Eidenbenz, Stephan and Nadiga, Balasubramanya and Suba{\c{s}}{\i}, Yi{\u{g}}it},
	journal={Quantum},
	volume={8},
	pages={1509},
	year={2024},
	publisher={Verein zur F{\"o}rderung des Open Access Publizierens in den Quantenwissenschaften}
}

@article{Lin2022Challenges,
	title={Challenges for quantum computation of nonlinear dynamical systems using linear representations},
	author={Lin, Yen Ting and Lowrie, Robert B and Aslangil, Denis and Suba{\c{s}}{\i}, Yi{\u{g}}it and Sornborger, Andrew T},
	journal={arXiv preprint arXiv:2202.02188},
	year={2022}
}

@article{Sarma2024VariationalPDE,
	title={Quantum variational solving of nonlinear and multidimensional partial differential equations},
	author={Sarma, Abhijat and Watts, Thomas W and Moosa, Mudassir and Liu, Yilian and McMahon, Peter L},
	journal={Physical Review A},
	volume={109},
	number={6},
	pages={062616},
	year={2024},
	publisher={APS}
}

@article{Song2025Incompressible,
	title={Incompressible Navier--Stokes solve on noisy quantum hardware via a hybrid quantum--classical scheme},
	author={Song, Zhixin and Deaton, Robert and Gard, Bryan and Bryngelson, Spencer H},
	journal={Computers \& Fluids},
	volume={288},
	pages={106507},
	year={2025},
	publisher={Elsevier}
}

@article{Wu2025RevisitingCarl,
	title={Quantum algorithms for nonlinear dynamics: Revisiting carleman linearization with no dissipative conditions},
	author={Wu, Hsuan-Cheng and Wang, Jingyao and Li, Xiantao},
	journal={SIAM Journal on Scientific Computing},
	volume={47},
	number={2},
	pages={A943--A970},
	year={2025},
	publisher={SIAM}
}

@article{Surana2024PolynomialDynSys,
	title = {An efficient quantum algorithm for simulating polynomial dynamical systems},
	volume = {23},
	ISSN = {1573-1332},
	url = {http://dx.doi.org/10.1007/s11128-024-04311-2},
	DOI = {10.1007/s11128-024-04311-2},
	number = {3},
	journal = {Quantum Information Processing},
	publisher = {Springer Science and Business Media LLC},
	author = {Surana,  Amit and Gnanasekaran,  Abeynaya and Sahai,  Tuhin},
	year = {2024},
	month = mar 
}

@article{Gourianov2025, 
	author = {Nikita Gourianov  and Peyman Givi  and Dieter Jaksch  and Stephen B. Pope },
	title = {Tensor networks enable the calculation of turbulence probability distributions},
	journal = {Science Advances},
	volume = {11},
	number = {5},
	pages = {eads5990},
	year = {2025},
	doi = {10.1126/sciadv.ads5990},
	URL = {https://www.science.org/doi/abs/10.1126/sciadv.ads5990},
	eprint = {https://www.science.org/doi/pdf/10.1126/sciadv.ads5990},
}

@article{Garrett2025Feasibility,
	title={Feasibility of Accelerating Incompressible Computational Fluid Dynamics Simulations With Fault-Tolerant Quantum Computers},
	author={Garrett, Michael C and Ponkratov, Dmitry and Qiu, Wei},
	journal={Marine Technology Society Journal},
	volume={59},
	number={1},
	pages={62--65},
	year={2025}
}

@article{Jennings2024FastForwarding,
	title = {The cost of solving linear differential equations on a quantum computer: fast-forwarding to explicit resource counts},
	volume = {8},
	ISSN = {2521-327X},
	DOI = {10.22331/q-2024-12-10-1553},
	journal = {Quantum},
	publisher = {Verein zur Forderung des Open Access Publizierens in den Quantenwissenschaften},
	author = {Jennings,  David and Lostaglio,  Matteo and Lowrie,  Robert B. and Pallister,  Sam and Sornborger,  Andrew T.},
	year = {2024},
	month = dec,
	pages = {1553}
}

@article{Tennie2023,
	title = {Quantum Computers for Weather and Climate Prediction: The Good,  the Bad,  and the Noisy},
	volume = {104},
	ISSN = {1520-0477},
	DOI = {10.1175/bams-d-22-0031.1},
	number = {2},
	journal = {Bulletin of the American Meteorological Society},
	publisher = {American Meteorological Society},
	author = {Tennie,  F. and Palmer,  T. N.},
	year = {2023},
	month = feb,
	pages = {E488–E500}
}

@article{Gaitan2020,
	title = {Finding flows of a Navier–Stokes fluid through quantum computing},
	volume = {6},
	ISSN = {2056-6387},
	DOI = {10.1038/s41534-020-00291-0},
	number = {1},
	journal = {npj Quantum Information},
	publisher = {Springer Science and Business Media LLC},
	author = {Gaitan,  Frank},
	year = {2020},
	month = jul 
}

@article{Williams2025Iterative,
	title={Quantum iterative methods for solving differential equations with application to computational fluid dynamics},
	author={Williams, Chelsea A and Gentile, Antonio A and Elfving, Vincent E and Berger, Daniel and Kyriienko, Oleksandr},
	journal={Advanced Quantum Technologies},
	pages={e00618},
	year={2025},
	publisher={Wiley Online Library}
}

@article{Oz2021,
	title = {Solving Burgers’ equation with quantum computing},
	volume = {21},
	ISSN = {1573-1332},
	DOI = {10.1007/s11128-021-03391-8},
	number = {1},
	journal = {Quantum Information Processing},
	publisher = {Springer Science and Business Media LLC},
	author = {Oz,  Furkan and Vuppala,  Rohit K. S. S. and Kara,  Kursat and Gaitan,  Frank},
	year = {2021},
	month = dec 
}

@article{Jin2024QAlgoPDE,
	title = {Quantum algorithms for nonlinear partial differential equations},
	volume = {194},
	ISSN = {0007-4497},
	DOI = {10.1016/j.bulsci.2024.103457},
	journal = {Bulletin des Sciences Mathématiques},
	publisher = {Elsevier BV},
	author = {Jin,  Shi and Liu,  Nana},
	year = {2024},
	month = sep,
	pages = {103457}
}

@article{Lapworth2022Hybrid,
	title={A hybrid quantum-classical CFD methodology with benchmark HHL solutions},
	author={Lapworth, Leigh},
	journal={arXiv preprint arXiv:2206.00419},
	year={2022}
}

@article{Ljubomir2022,
	title = {Quantum algorithm for the Navier–Stokes equations by using the streamfunction-vorticity formulation and the lattice Boltzmann method},
	volume = {20},
	ISSN = {1793-6918},
	DOI = {10.1142/s0219749921500398},
	number = {02},
	journal = {International Journal of Quantum Information},
	publisher = {World Scientific Pub Co Pte Ltd},
	author = {Ljubomir,  Budinski},
	year = {2022},
	month = feb 
}

@article{Lubasch2020,
	title = {Variational quantum algorithms for nonlinear problems},
	volume = {101},
	ISSN = {2469-9934},
	DOI = {10.1103/physreva.101.010301},
	number = {1},
	journal = {Physical Review A},
	publisher = {American Physical Society (APS)},
	author = {Lubasch,  Michael and Joo,  Jaewoo and Moinier,  Pierre and Kiffner,  Martin and Jaksch,  Dieter},
	year = {2020},
	month = jan 
}

@article{Hogancamp2026Linear,
	title={A Linear Combination of Unitaries Decomposition for the Laplace Operator},
	author={Hogancamp, Thomas and Demirdjian, Reuben and Gunlycke, Daniel},
	journal={arXiv preprint arXiv:2601.06370},
	year={2026}
}

@article{soderlind2006logarithmic,
	title={The logarithmic norm. History and modern theory},
	author={S{\"o}derlind, Gustaf},
	journal={BIT Numerical Mathematics},
	volume={46},
	number={3},
	pages={631--652},
	year={2006},
	publisher={Springer}
}

@phdthesis{horstmann2018hybrid,
	title={Hybrid numerical methods based on the lattice Boltzmann approach with application to non-uniform grids},
	author={Horstmann, Tobias},
	year={2018},
	school={Universit{\'e} de Lyon}
}

@article{BGK1954,
	title = {A Model for Collision Processes in Gases. I. Small Amplitude Processes in Charged and Neutral One-Component Systems},
	author = {Bhatnagar, P. L. and Gross, E. P. and Krook, M.},
	journal = {Phys. Rev.},
	volume = {94},
	issue = {3},
	pages = {511--525},
	numpages = {0},
	year = {1954},
	month = {May},
	publisher = {American Physical Society},
	doi = {10.1103/PhysRev.94.511},
	url = {https://link.aps.org/doi/10.1103/PhysRev.94.511}
}

@article{hantzko2024tensorized,
	title={Tensorized Pauli decomposition algorithm},
	author={Hantzko, Lukas and Binkowski, Lennart and Gupta, Sabhyata},
	journal={Physica Scripta},
	volume={99},
	number={8},
	pages={085128},
	year={2024},
	publisher={IOP Publishing}
}

@article{itani2022analysis,
	title={Analysis of Carleman linearization of lattice Boltzmann},
	author={Itani, Wael and Succi, Sauro},
	journal={Fluids},
	volume={7},
	number={1},
	pages={24},
	year={2022},
	publisher={MDPI}
}

@article{sanavio2024lattice,
	title={Lattice Boltzmann--Carleman quantum algorithm and circuit for fluid flows at moderate Reynolds number},
	author={Sanavio, Claudio and Succi, Sauro},
	journal={AVS Quantum Science},
	volume={6},
	number={2},
	year={2024},
	publisher={AIP Publishing}
}

@article{sanavio2024three,
	title={Three Carleman routes to the quantum simulation of classical fluids},
	author={Sanavio, Claudio and Scatamacchia, Riccardo and De Falco, Carlo and Succi, Sauro},
	journal={Physics of Fluids},
	volume={36},
	number={5},
	year={2024},
	publisher={AIP Publishing}
}

@article{succi2026foundational,
	title={The foundational value of quantum computing for classical fluids},
	author={Succi, Sauro and Sanavio, Claudio and Cappelli, Luca and Love, Peter},
	journal={Europhysics Letters},
	volume={153},
	number={2},
	pages={28001},
	year={2026},
	publisher={EDP Sciences, IOP Publishing and Societ{\`a} Italiana di Fisica}
}

@article{gan2025provably,
	title={Provably Efficient Quantum Algorithms for Solving Nonlinear Differential Equations Using Multiple Bosonic Modes Coupled with Qubits},
	author={Gan, Yu and Alipanah, Hirad and Cheng, Jinglei and Wu, Zeguan and Li, Guangyi and Mendoza-Arenas, Juan Jos{\'e} and Givi, Peyman and Malik, Mujeeb R and McDermott, Brian J and Liu, Junyu},
	journal={arXiv preprint arXiv:2511.09939},
	year={2025}
}

@article{sanavio2025carleman,
	title={Carleman-lattice-Boltzmann quantum circuit with matrix access oracles},
	author={Sanavio, Claudio and Simon, William A and Ralli, Alexis and Love, Peter and Succi, Sauro},
	journal={Physics of Fluids},
	volume={37},
	number={3},
	year={2025},
	publisher={AIP Publishing}
}

@article{penuel2024fesibility,
	title={Detailed assessment of calculating drag force with quantum computers: Explicit time-evolution precludes exponential advantage for nonlinear differential equations},
	author={Penuel, John and Katabarwa, Amara and Johnson, Peter D and Kuklinski, Parker and Rempfer, Benjamin and Farquhar, Collin and Cao, Yudong and Garrett, Michael C},
	journal={arXiv preprint arXiv:2406.06323},
	year={2024}
}


\appendix


\newpage
\onecolumngrid

\section{Notation}

\begin{table}[!htbp]
	\begin{center}
		\renewcommand{\arraystretch}{1.15}
		\begin{tabular}{p{0.125\linewidth} | p{0.85\linewidth} }
			\hline
			\multicolumn{2}{c}{Carleman Linearization} \\
			\hline \hline
			$L$ & Matrix from the linear system resulting from Carleman linearization \\
			$L^{(\text{e})}$ & Matrix from the linear system resulting from Carleman linearization and zero padding \\
			$\alpha$ & Truncation order \\
			$\Delta$ & $=\sum_{j=1}^\alpha N^j$ \\
			$N_F$ & Degree of polynomial nonlinearity in the generic ODE in \eqref{eqn:GeneralIVP} \\
			$(\cdot)^{(\text{e})}$ & An embedded version of the matrix or vector \\
			\hline
			\multicolumn{2}{c}{Lattice Boltzmann} \\
			\hline \hline
			$Q$ & Number of discrete velocities \\
			$n_x, n_y, n_z$ & Number of discrete grid points in $x,y,z$ \\
			$n$ & $=n_xn_yn_z$ \\
			$\vec{e}_m$ & The $m\text{th}$ discrete velocity vector \\
			$w_m$ & Velocity weights \\
			$\beta_{m,q}$ & Coefficient for the linear collision term (defined in \eqref{eqn:gamma_beta}) \\
			$\gamma_{q,m,r}$ & Coefficient for the quadratic and cubic collision terms (defined in \eqref{eqn:gamma_beta}) \\
			D($n$)Q($m$) & A velocity set with $n$ spatial dimensions and $m$ speeds \\
			D($n$)Q($m$)* & A velocity set with $n$ spatial dimensions and $2^{ \lfloor \log m \rfloor + 1}$ speeds \\
			\hline
			\multicolumn{2}{c}{Miscellaneous} \\
			\hline \hline			
			$\log x$ & $=\log_2 x$ \\
			LCU & Linear Combination of Unitaries \\
			LCNU & Linear Combination of Non-Unitaries \\
			$N_s$ & Number of terms in the LCU or LCNU \\
			DVBE & Discrete-Velocity Boltzmann Equation \\
			LBE & Lattice Boltzmann Equation \\
			QLSA & Quantum Linear System Algorithm
		\end{tabular}
	\end{center}
\end{table}

%


\section{Proofs}


\subsection{Proof of Lemma \ref{Lemma for Ul}} \label{Proof for U_l}
\LemmaUl*
\begin{proof}
	We begin with the forward direction. Assume that $\overline{Q}$ is a unitary completion of $Q$. We want to show that $U$ is unitary. Using \eqref{eqn:U block of Q} and $Q^\perp\coloneq\overline{Q}-Q$ we have
	\begin{equation} \label{eqn:UUdagger}
		UU^\dagger = 
		\begin{pmatrix} 
			\overline{Q}\,\overline{Q}^\dagger 
			- \overline{Q}Q^\dagger 
			- Q\overline{Q}^\dagger 
			+ 2QQ^\dagger 
			& QQ^{\perp\dagger} + Q^\perp Q^\dagger \\
			QQ^{\perp\dagger} + Q^\perp Q^\dagger 
			& \overline{Q}\,\overline{Q}^\dagger 
			- \overline{Q}Q^\dagger 
			- Q\overline{Q}^\dagger 
			+ 2QQ^\dagger 
		\end{pmatrix} .
	\end{equation}
	So, if $\overline{Q}\,\overline{Q}^\dagger=I$, $\overline{Q}Q^\dagger=Q\overline{Q}^\dagger=QQ^\dagger$ and $QQ^{\perp\dagger} + Q^\perp Q^\dagger=\mathbf{0}$, then $U$ is unitary. Given that $\overline{Q}$ is unitary, then it follows that (1) $\overline{Q}\,\overline{Q}^\dagger=I$ by definition, and (2) $\overline{Q}$ has orthonormal eigenvectors $v_i$ with eigenvalues $\lambda_i$. Then, since the eigenvectors of a unitary operator form a basis, we can write
	\begin{equation*}
		\overline{Q}=\sum\nolimits_{j} \lambda_j v_jv_j^\dagger .
	\end{equation*}
	By definition, if $v\in W$ is an eigenvector of $Q$ with eigenvalue $\lambda$, then 
	\begin{equation*}
		\overline{Q}v=Qv=\lambda v ,
	\end{equation*}
	and so $v$ is also an eigenvector of $\overline{Q}$ with eigenvalue $\lambda$. Define 
	\begin{equation*}
		\mathcal{I}=\{i: (\lambda_i,v_i) \text{ is eigenpair for } Q\} ,
	\end{equation*}
	and let $\mathcal{I}_c=\{1,\cdots, N\}\setminus \mathcal{I}$. Then, 
	\begin{equation*}
		\overline{Q}
		=\sum_{j\in \mathcal{I}} \lambda_j v_jv_j^\dagger
		+\sum_{j\in \mathcal{I}_c} \lambda_j v_jv_j^\dagger .
	\end{equation*}
	Note that $\Span\{v_j \mid j \in \mathcal{I}\} = W$ and $\Span\{v_j \mid j \in \mathcal{I}^c\} = W^\perp$. Recall that $Q$ is trivial on $W^\perp$. It follows that $Q = \sum\nolimits_{j\in \mathcal{I}} \lambda_j v_jv_j^\dagger$. Thus, 
	\begin{equation*}
		\overline{Q}
		=Q+\sum_{j\in \mathcal{I}_c} \lambda_j v_jv_j^\dagger .
	\end{equation*}
	Since $Q^{\perp} \coloneq \overline{Q}-Q$, it follows that 
	\begin{equation*}
		Q^{\perp} = \sum_{j\in \mathcal{I}_c} \lambda_j v_jv_j^\dagger ,
	\end{equation*}
	which gives
	\begin{equation*}
		Q^{\perp} Q^\dagger 
		=
		\Bigg(\sum_{j\in \mathcal{I}_c} \lambda_j v_jv_j^\dagger\Bigg) \Bigg(\sum_{i \in \mathcal{I}} \lambda_i^* v_iv_i^\dagger\Bigg)
		=\mathbf{0} .
	\end{equation*}
	Therefore, we have
	\begin{equation*}
		Q^\perp Q^\dagger = \mathbf{0} 
		\text{ and }
		Q^{\perp}Q^\dagger=(\overline{Q}-Q)Q^\dagger=\mathbf{0} 
		\implies \overline{Q} Q^\dagger = Q Q^\dagger .
	\end{equation*}
	By taking the conjugate transpose of these results we find that $(Q^\perp Q^\dagger)^\dagger = Q Q^{\perp\dagger} = \textbf{0}$ and $(\overline{Q} Q^\dagger)^\dagger = (Q Q^\dagger)^\dagger \implies Q \overline{Q}^\dagger = Q Q^\dagger$. Bringing the relations together, we find that 	
	\begin{equation*}
		QQ^\dagger=\overline{Q} Q^\dagger = Q \overline{Q}^\dagger
		\quad \text{and}  \quad
		QQ^{\perp\dagger}=\mathbf{0} .
	\end{equation*}
	Using this and the fact that $\overline{Q}$ is unitary, it follows that $\overline{Q}\,\overline{Q}^\dagger - \overline{Q}Q^\dagger - Q\overline{Q}^\dagger + 2QQ^\dagger=I$ and $QQ^{\perp\dagger} + Q^\perp Q^\dagger=\mathbf{0}$. Evaluating both of these properties into \eqref{eqn:UUdagger} gives $UU^\dagger=I$ as desired. 
	
	Next, we prove the backward direction. Assume that $U$ from \eqref{eqn:U block of Q} is unitary. We want to show that $\overline{Q}$ is a unitary completion of $Q$. Suppose $Q \colon W \to V$ is unitary on $W$. Then for $w\in W$ we have
	\begin{equation}
		U \begin{pmatrix} w \\ 0 \end{pmatrix} 
		= \begin{pmatrix} Q w \\ Q^\perp w \end{pmatrix}.
	\end{equation}
	Since $U$ is unitary, it follows that $\|w\|^2 = \|Qw\|^2 + \|Q^\perp w\|^2 = \|w\|^2 + \|Q^\perp w\|^2$, where the last relation is true because $Q$ is unitary on $W$. So, $\|Q^\perp w\|^2=0$, which implies that $Q^\perp w=0$ for all $w\in W$. By definition of $Q^\perp$, it follows that $\overline{Q}w = Qw$ as desired. Lastly, we must show that $\overline{Q}$ is unitary. Since $U$ is unitary, using \eqref{eqn:UUdagger} yields the following relations: 
	\begin{equation*}
	\begin{split}
		\overline{Q} \, \overline{Q}^\dagger 
		- \overline{Q} Q^\dagger 
		- Q \overline{Q}^\dagger 
		+ 2Q Q^\dagger 
		&= I , \\
		Q(\overline{Q}^\dagger - Q^\dagger) + (\overline{Q} - Q) Q^\dagger &= 0 .
	\end{split}
	\end{equation*}
	By inserting the second into the first, we find that $\overline{Q} \, \overline{Q}^\dagger = I$. Thus, $\overline{Q}$ is unitary and is therefore a unitary completion of $Q$.
\end{proof}


\subsection{Proof of Theorem \ref{thm:Theorem for Lbar}} \label{Proof for Lbar}
\TheoremLbar*
\begin{proof}
	For convenience, in this proof we drop the subscript $l$ from $L_l$. Let us first consider the case in which $m=2$ and write $P_{1} = \bigotimes_{j=1}^{n_{1}}P_{1,j}$, $P_{2} = \bigotimes_{j=1}^{n_{2}}P_{2,j}$. Note that $P_{1},P_{2}$ are $1$-sparse matrices whose nonzero entries have modulus 1, so that Lemma \ref{trivial_cond} implies that $L \coloneq P_1P_2 = 0$ if and only if $\tilde{W}_{2}=\{ w \in W_{2} \mid P_{2}w \in W_{1}\} = \{0\}$, where $W_{1},W_{2}$ are taken to be the rowspaces of $P_{1}, P_{2}$, respectively. Since we are interested in nontrivial $L$, then we assume that $\tilde{W}_{2} \ne \{0\}$. To find the unitary completions for $P_1$ and $P_2$ individually, we invoke Lemma \ref{Lemma of Q1 kron Q2} to obtain
	\begin{equation} \label{tceq}
		\overline{P}_{i} = \bigotimes_{j=1}^{n_{i}}\overline{P}_{i,j} , 
	\end{equation}
	for $i=1,2$ and $\overline{P}_{i,j}$ are defined by \eqref{eqn:P bar}. That $\overline{P}_{1}\overline{P}_{2}$ is a unitary completion of $P_{1}P_{2}$ follows from Lemma \ref{incl_prod}. Note that, using the constructions in the proof of Lemma \ref{trivial_cond}, it follows that $\tilde{W}_{2} =\Span\{e_{i} \mid \beta_{i}\neq 0, \alpha_{\gamma(i)}\neq 0\}$, i.e. $\tilde{W}_{2}$ is the rowspace of $Q_{1}Q_{2}$. 
	
	The general case now follows from an inductive argument. Let $m>2$, set $P_{i} = \bigotimes_{j=1}^{n_{i}}P_{i,j}$, and write 
	\begin{equation*}
		L = P_{1}P_{2}\cdots P_{m-1}P_{m} = \left(P_{1}\cdots P_{m-1}\right)P_{m} .
	\end{equation*}
	Note that $P_{1}\cdots P_{m-1}$ and $P_{m}$ are both $1$-sparse matrices whose entries are either $0$ or have modulus $1$. Let $V_{1}$, $V_{2}$ denote the rowspaces of $P_{1}\cdots P_{m-1}$ and $P_{m}$, respectively. It follows from Lemma \ref{trivial_cond} that $L$ is either trivial or $\{v \in V_{2} \mid P_{m}v \in V_{1}\} \neq \{0\}$. Assume that $L$ is nontrivial. We take as an inductive hypothesis that $\prod_{i=1}^{m-1} \left(\bigotimes_{j=1}^{n_{i}}\overline{P}_{i,j}\right)$ is a unitary completion of $P_{1}\cdots P_{m-1}$, where $\overline{P}_{i,j}$ is taken according to \eqref{eqn:P bar}. Next, Lemma \ref{Lemma of Q1 kron Q2} shows that $\bigotimes_{j=1}^{n_{m}}\overline{P}_{m,j}$ is a completion of $P_{m}$. Finally, by application of Lemma \ref{incl_prod}, it follows that 
	\begin{equation*}
		\left( \prod_{i=1}^{m-1} \bigotimes_{j=1}^{n_{i}} \overline{P}_{i,j} \right) 
		\left( \bigotimes_{j=1}^{n_{m}}\overline{P}_{m,j} \right) 
		= \prod_{i=1}^{m} \left( \bigotimes_{j=1}^{n_{i}} \overline{P}_{i,j} \right)
	\end{equation*} 
	is a unitary completion of $P_{1}\cdots P_{m}$. Since $m$ is arbitrary, this completes the inductive argument.
\end{proof}


\subsection{Proof of Lemma \ref{Lemma of Q1 kron Q2}} \label{Proof of Q1 kron Q2}
\begin{restatable}[]{lemma}{LemmaQotimesR} \label{Lemma of Q1 kron Q2}
	Suppose that $W_i$ is a subspace of $V_i$ for $i=1,2$. Let $Q_i \colon W_i \to V_i$ be unitary on $W_i$, where $\overline{Q}_i \colon V_i \to V_i$ is a unitary completion of $Q_i$. Let $Q_1 \otimes Q_2 \colon \tilde{W} \to \tilde{V}$ where $\tilde{W}=W_1 \otimes W_2 \ne \{0\}$ and $\tilde{V}=V_1 \otimes V_2$. Then $\overline{Q}_1 \otimes \overline{Q}_2 \colon \tilde{V} \to \tilde{V}$ defines a valid unitary completion of $Q_1 \otimes Q_2$.
\end{restatable}
\begin{proof}	
	To prove that $\overline{Q}_1 \otimes \overline{Q}_2$ is a unitary completion to $Q_1 \otimes Q_2$, we must show that (1) $(\overline{Q}_1 \otimes \overline{Q}_2 ) (w_1 \otimes w_2) = (Q_1  \otimes Q_2) (w_1 \otimes w_2)$, where $w_1 \in W_1$ and $w_2 \in W_2$, and (2) $\overline{Q}_1 \otimes \overline{Q}_2$ is unitary on $\tilde{V}$. 
	
	To prove (1) let $w_1\in W_1$ and $w_2\in W_2$, then by definition $w_1 \otimes w_2 \in \tilde{W}$. Then,
	\begin{equation}
		\begin{split}
			(\overline{Q}_1\otimes \overline{Q}_2) (w_1\otimes w_2) 
			&= (\overline{Q}_1w_1)\otimes (\overline{Q}_2w_2) \\
			&= (Q_1w_1)\otimes (Q_2w_2) \\
			&= (Q_1\otimes Q_2)(w_1\otimes w_2) ,
		\end{split}
	\end{equation}
	where the second equality is true given that $\overline{Q}_i$ is a unitary completion of $Q_i$. 
	
	Next, to prove (2) let $v_1 \in V_1$ and $v_2 \in V_2$, then
	\begin{equation}
		\begin{split}
			(v_1\otimes v_2)^\dagger(\overline{Q}_1\otimes \overline{Q}_2)^\dagger(\overline{Q}_1\otimes \overline{Q}_2)(v_1\otimes v_2)
			&= \left( (\overline{Q}_1 v_1) \otimes (\overline{Q}_2 v_2) \right)^\dagger
			\left( (\overline{Q}_1 v_1) \otimes (\overline{Q}_2 v_2) \right) \\
			&=\left( v_1^\dagger \overline{Q}_1^\dagger \overline{Q}_1 v_1 \right)
			\otimes \left( v_2^\dagger \overline{Q}_2^\dagger \overline{Q}_2 v_2 \right) \\
			&=( v_1^\dagger v_1 ) \otimes ( v_2^\dagger v_2 ) \\
			&= (v_1 \otimes v_2)^\dagger (v_1 \otimes v_2) ,
		\end{split}
	\end{equation}
	where the second equality is true using the fact that $\overline{Q}_i$ is unitary on $V_i$. 
	
\end{proof}

\subsection{Proof of Lemma \ref{thm_abst_cond}} \label{Proof of thm_abst_cond}
\begin{restatable}[]{lemma}{LemmaQR} \label{thm_abst_cond}
	Let $W_{1}$ and $W_{2}$ be subspaces of $V$.  Suppose that $Q_{i}:V\to V$ is unitary on $W_{i}$ and trivial on $W_{i}^{\perp}$, for $i =1,2$. If $W_{2}$ admits a decomposition of the form
	\begin{equation*}
		W_{2} = \tilde{W}_{2}\oplus \tilde{W}_{2}^{'}\, , 
	\end{equation*}
	where $\tilde{W}_{2} = \{ w \in W_{2} \, | \, Q_{2}(w) \in W_{1}\}$ and $\tilde{W}_{2}^{'} = \{w \in W_{2} \, | \, Q_{2}(w) \in W_{1}^{\perp}\}$, then $Q_{1}Q_{2} \equiv 0$ if and only if $\tilde{W}_{2} = \{0\}$. 
\end{restatable}

\begin{proof}
	Since $V$ is finite dimensional, we can write 
	\begin{equation*}
		V = W_{2} \oplus W_{2}^{\perp} \, . 
	\end{equation*}
	Moreover, the decomposition of $W_{2}$ provided in the hypothesis implies that 
	\begin{equation*}
		V = \tilde{W}_{2}\oplus \tilde{W}^{'}_{2} \oplus W_{2}^{\perp} \, . 
	\end{equation*}
	Hence, each $v \in V$ admits a decomposition of the form 
	\begin{equation} \label{v_decomp}
		v = \tilde{w}_{2}+\tilde{w}_{2}^{'}+w_{2}^{\perp} \, , 
	\end{equation}
	where $\tilde{w}_{2} \in \tilde{W}_{2}$, $\tilde{w}^{'}_{2} \in \tilde{W}_{2}^{'}$, and $w_{2}^{\perp} \in \tilde{W}^{\perp}_{2}$. 
	
	Starting with the backward direction, assume that $\tilde{W}_2={0}$. It follows from \eqref{v_decomp} that
	\begin{align} \label{eq1}
		\begin{split}
			Q_{1}Q_{2}(v) &= Q_{1}Q_{2}(\tilde{w}_{2})+Q_{1}Q_{2}(\tilde{w}_{2}^{'})+Q_{1}Q_{2}(w_{2}^{\perp} ) \\ 
			&= Q_{1}Q_{2}(\tilde{w}_{2})+Q_{1}\left(Q_{2}(\tilde{w}_{2}^{'})\right) \\ 
			&= Q_{1}Q_{2}(\tilde{w}_{2})
		\end{split}
	\end{align}
	where the second inequality follows from $Q_{2}(w_{2}^{\perp})=0$, and the third is true since $Q_{2}(\tilde{w}^{'}_{2}) \in W_{1}^{\perp}$. It's clear from \eqref{eq1} that $\tilde{W}_{2}=\{0\}$ implies $Q_{1}Q_{2} =0$. 
	
	Next, for the forward direction, assume $Q_{1}Q_{2}=0$. Then, it follows from our hypothesis that 
	\begin{equation}
		0=\|Q_{1}Q_{2}(\tilde{w}_{2})\| = \|Q_{2}(\tilde{w}_{2})\| = \|\tilde{w}_{2}\| \, , 
	\end{equation}
	for all $\tilde{w}_{2} \in \tilde{W}_{2}$. The second equality follows from the fact that $Q_2 \tilde{w}_2 \in W_1$ and $Q_1$ is unitary on $W_1$. Similarly, the third equality follows from the fact that $Q_2$ is unitary on $W_2$. With this, we see that $\tilde{W}_{2} = \{0\}$ as desired.
\end{proof}


\subsection{Proof of Lemma \ref{trivial_cond}} \label{Proof of trivial_cond}
\begin{restatable}[]{lemma}{LemmaRowspace} \label{trivial_cond}
	Suppose that $Q_{1}, Q_{2} \in \mathbb{C}^{N \times N}$ are $1$-sparse matrices, having at most one non-zero entry per row, whose entries are either $0$ or have modulus $1$. Let $W_{1},W_{2}$ be the rowspaces of $Q_{1},Q_{2}$, respectively. Then $Q_{1}Q_{2} = 0$ if and only if $\tilde{W}_{2} \coloneq \{w \in W_{2} \mid Q_{2}w \in W_{1}\} =\{0\}$. 
\end{restatable}
\begin{proof}
	Let $\{e_{i}\}_{i=1}^{N}$ denote the standard orthonormal basis. It follows from our assumptions that the action of $Q_{1}, Q_{2}$, and $Q_{1}Q_{2}$ can be described in the form
	\begin{equation} \label{eqn:basis_mult}
		Q_{1}e_{i} = \alpha_{i}e_{\sigma(i)} \, , \qquad Q_{2}e_{i}=\beta_{i}e_{\gamma(i)}, \qquad \text{and} \qquad Q_{1}Q_{2}e_{i} = \alpha_{\gamma(i)}\beta_{i}e_{\sigma(\gamma(i))}\, , 
	\end{equation}
	where the coefficients $\alpha_{i}, \beta_{i}$ are constants that are either zero or have modulus $1$ and $\gamma, \sigma$ are permutations on $\{1,\dots, N\}$. Lemma \ref{thm_abst_cond} can be invoked to prove the claim once the relevant hypotheses (1) $Q_i$ is trivial on $W_i^\perp$ for $i=1,2$, and (2) $W_{2} = \tilde{W}_{2} \oplus \tilde{W}_{2}^\prime$ are verified. 
	
	To prove hypothesis (1), we use the constructions in \eqref{eqn:basis_mult}. Then, it follows that $W_1 = \text{span} \{e_i \mid \alpha_i \ne 0 \}$ and therefore $W_1^\perp = \text{span} \{e_i \mid \alpha_i = 0\}$, which shows that $Q_1$ is trivial on $W_1^\perp$. Similarly, $W_2 = \text{span} \{e_i \mid \beta_i \ne 0 \}$ and therefore $W_2^\perp = \text{span} \{e_i \mid \beta_i = 0\}$, which shows that $Q_2$ is trivial on $W_2^\perp$. Thus, hypothesis (1) is proved.
	
	To prove hypothesis (2), we note that $\tilde{W}_2 = \text{span}\{e_{i} \mid \beta_{i} \neq 0, \alpha_{\gamma(i)} \neq 0\}$. The expression $W_2 = \tilde{W}_2 + \tilde{W}_2^\prime$ is satisfied with $\tilde{W}_{2}^\prime = \textup{span}\{e_{i} \mid \beta_{i} \neq 0, \alpha_{\gamma(i)}=0\}$. Moreover, since $\tilde{W}_2 \cap \tilde{W}_2^\prime = \{0\}$, then $\tilde{W}_2 + \tilde{W}_2^\prime$ is a direct sum, and hypothesis (2) is proved. Therefore, both hypotheses of Lemma \ref{thm_abst_cond} are satisfied, and so $Q_1Q_2 = 0$ if and only if $\tilde{W}_2 = \{0\}$.
\end{proof}


\subsection{Proof of Lemma \ref{incl_prod}} \label{Proof of incl_prod}
\begin{restatable}[]{lemma}{LemmaQRcompletion} \label{incl_prod}
	Suppose $Q_i \in \mathbb{C}^{N\times N}$ for $i=1,2$ is a $1$-sparse matrix whose entries are either $0$ or have modulus $1$. Let $W_i$ be the rowspace of $Q_i$, and let $Q_i$ be unitary on $W_i$. Suppose also that $\tilde{W}_{2} \coloneq \{w \in W_{2} \mid Q_{2}w \in W_{1}\} \neq \{0\}$ and consider $Q_{1}Q_{2} \colon \tilde{W}_2 \to V$. If $\overline{Q}_i$ is a valid unitary completion of $Q_i$, then $\overline{Q}_{1}\overline{Q}_{2}$ is a valid unitary completion of $Q_{1}Q_{2}$. 
\end{restatable}
\begin{proof} 
	To prove that $\overline{Q}_1 \overline{Q}_2$ is a unitary completion to $Q_1 Q_2$, we must show that (1) $\overline{Q}_1 \overline{Q}_2 w = Q_1  Q_2 w$ where $w \in \tilde{W}_2$, and (2) $\overline{Q}_1 \overline{Q}_2$ in unitary on $\mathbb{C}^N$. 
	
	To prove (1), let $w \in \tilde{W}_{2}$, then
	\begin{equation*}
		\overline{Q}_{1}\overline{Q}_{2}w = \overline{Q}_{1}Q_{2}w = Q_{1}Q_{2}w \, , 
	\end{equation*}
	where the first relation is true because $\overline{Q}_2$ is a unitary completion of $Q_2$, and the second because $Q_2w \in W_1$ and $\overline{Q}_1$ is a unitary completion of $Q_1$.
	
	To prove (2), we note that the product of two unitary matrices is unitary. Thus, the claim is proved since both (1) and (2) are true.
\end{proof}


\subsection{Proof of Corollary \ref{thm:UPV}} \label{Proof for UPV}
\CorUPV*
\begin{proof}
	Let $W^{'}_{1}$ denote the rowspace of $\prod_{i=1}^{m}\bigotimes_{j=1}^{n_{i}} P_{i, j}$ and set $W_{1} \coloneqq V^{-1}(W^{'}_{1})$. We need to show two things: 1) That $\overline{L}$ is unitary, and 2) that $\overline{L}w = Lw$ for all $w \in W_{1}$. The first condition follows from the fact that each factor on the right hand side of \eqref{p2} is unitary. The second is verified by 
	\begin{equation*}
		\overline{L}w = U \left(\prod_{i=1}^{m} \bigotimes_{j=1}^{n_{i}}\overline{P}_{i,j}\right) \left(Vw\right) = 
		U \left(\prod_{i=1}^{m} \bigotimes_{j=1}^{n_{i}}P_{i,j}\right) \left(Vw\right) = Lw \, , 
	\end{equation*}
	where the third equality follows from Definition \ref{def:completion} and the hypothesis that $Vw \in W_{1}^{'}$. 
\end{proof}


\subsection{Proof of Corollary \ref{thm:U1_P2}} \label{Proof for U1_P2}
\corollaryLLT*
\begin{proof}
	Following Theorem \ref{thm:Theorem for U1}, we only need to show that $L_l L_l^T \in \mathcal{R}$ is true. We also need to consider the case that if $L_l L_l^T = \mathbf{0}$, then $U_{l,1}$ is simply the identity following \eqref{eqn:Ul1 Ul2}.
	
	Next, define $P_i \coloneq \bigotimes_{j=1}^{n_i} P_{i,j}$, where $P_{i,j} \in \mathscr{P}_2$ with size $2^{m_j} \times 2^{m_j}$ for some positive integer $m_j$. Using \eqref{eqn:Llgeneral}, this gives $L=P_1 \cdots P_m$. Let us first consider the case in which $m=2$ and write $L_lL_l^T = P_1 P_2 P_2^T P_1^T$. We now show that $L_lL_l^T \in \mathcal{R}$. Expanding out the middle terms, we find
	\begin{equation}
		P_2 P_2^T 
		= \left( \bigotimes_{j=1}^{n_2} P_{2,j} \right) \left( \bigotimes_{j=1}^{n_2} P_{2,j}^T \right)
		= \left( \bigotimes_{j=1}^{n_2} P_{2,j} P_{2,j}^T \right) .
	\end{equation}
	We now show that $P_2 P_2^T \in \mathcal{R}$ by considering all possible cases. First, let $P_{2,j} \in \mathbb{P}_\rho$, then $P_{2,j} P_{2,j}^T \in \{\rho_0,\rho_3,I\} \in \mathcal{R}$. Next, let $P_{2,j} \in \mathbb{P}_\sigma$, then $P_{2,j} P_{2,j}^T = \pm I \in \mathcal{R}$, since $\mathcal{R}$ is closed under scalar multiplication. Lastly, let $P_{2,j} \in \mathbb{S}_n$ is unitary, then we have $P_{2,j} P_{2,j}^T = I \in \mathcal{R}$. Since $\mathcal{R}$ is closed under tensor products, and $P_{2,j} P_{2,j}^T \in \mathcal{R}$, then $P_2 P_2^T \in \mathcal{R}$ as claimed. 
	
	To finalize the $m=2$ case, we now show that $P_1 P_2 P_2^T P_1^T \in \mathcal{R}$. First, observe that since $P_2P_2^T \in \mathcal{R}$ and $P_2$ is the same size as $P_1$, we may write $P_2 P_2^T = \bigotimes_{j=1}^{n_1} R_j$ for some $R_j \in \mathcal{R}$ with size $2^{m_j} \times 2^{m_j}$. With this, we expand out
	\begin{equation}
		\begin{split}
			P_1 P_2 P_2^T P_1^T &= \left( P_{1,1} \otimes \cdots \otimes P_{1,n_1} \right) 
			\left( R_1 \otimes \cdots \otimes R_{n_1} \right)
			\left( P_{1,1}^T \otimes \cdots \otimes P_{1,n_1}^T \right) \\
			&= P_{1,1}R_1P_{1,1}^T \otimes \cdots \otimes P_{1,n_1}R_{n_1}P_{1,n_1}^T .
		\end{split}
	\end{equation} 
	We claim that $P_{1,j}R_jP_{1,j}^T \in \mathcal{R} \cup \{\mathbf{0}\}$ for all $j=1,\dots,n_1$. Once again, this claim is proved by considering all possible cases. Let $P_{1,j} \in \mathbb{P}_\rho$, then $P_{1,j}R_jP_{1,j}^T \in \{\mathbf{0},\rho_0,\rho_3,I\} \in \mathcal{R} \cup \{\mathbf{0}\}$. Next, let $P_{1,j} \in \mathbb{P}_\sigma$, then $P_{1,j}R_jP_{1,j}^T \in \{\pm\rho_0,\pm\rho_3,\pm I\} \in \mathcal{R}$, since $\mathcal{R}$ is closed under scalar multiplication. Lastly, let $P_{1,j} \in \mathbb{S}_n$, then from \ref{eqn:Set Sn}
	\begin{equation}
		P_{1,j}R_jP_{1,j}^T = S \left( \bigotimes_{k=1}^{m_j} X^{b_k} \right) R_j \left( \bigotimes_{k=1}^{m_j} X^{b_k} \right) S^T
		= S \hat{R}_j S^T ,
	\end{equation}
	where $\hat{R} = ( \bigotimes_{k=1}^{m_j} X^{b_k} ) R_j ( \bigotimes_{k=1}^{m_j} X^{b_k} ) \in \mathcal{R}$, which follows from the fact that $X \rho X \in \mathcal{R}$ for $\rho \in \{\rho_0,\rho_3,I\}$. Let $S(i)$ be the $S$ permutation of a binary string representation of an integer $i$. Expand $\hat{R} = \hat{R}_1 \otimes \cdots \otimes \hat{R}_{m_j}$, then $S\hat{R}S^T = \hat{R}_{S(1)} \otimes \cdots \otimes \hat{R}_{S(m_j)} \in \mathcal{R}$. Therefore, $P_{1,j} R_j P_{1,j}^T \in \mathcal{R}$ and, since $\mathcal{R}$ is closed under tensor products, it follows that $P_1P_2P_2^TP_1^T \in \mathcal{R}$ as claimed.
	
	The general case now follows from an inductive argument. Let $m>2$, then by definition we have $L_l L_l^T = P_1 \cdots P_m P_m^T \cdots P_1^T$. We take as an inductive hypothesis that $P_{m-1} P_m P_m^T P_{m-1}^T \in \mathcal{R} \cup \{\mathbf{0}\}$. From the logic of the $m=2$ case, it follows that $P_{m-2} P_{m-1} P_m P_m^T P_{m-1}^T P_{m-2}^T \in \mathcal{R} \cup \{\mathbf{0}\}$. By successive applications of this, we find that $L_l L_l^T \in \mathcal{R} \cup \{\mathbf{0}\}$. Therefore, $U_{l,1}$ is either the identity or can be implemented with a single multi-controlled \textsc{NOT} gate following \cite{GS2025_LCofThings}.
\end{proof}


\section{Derivations of Important Relations}

\subsection{Derivation of $A^{(\textnormal{e})}$ from \eqref{eqn:Ae_expr}} \label{sec:Ae deriv}
Before deriving the expression for $A^{(\text{e})}$, we introduce a useful construction. Define the matrix $E_{i,j}^{n}$ as an $n \times n$ matrix such that $n$ is a power of two and all entries are zero except that the $(i,j)\text{th}$ entry is equal to $1$. We claim that $E_{i,j}^{n} = \bigotimes_{l=1}^{\log n} r_l$ for some $r_l \in \mathbb{P}_\rho$. As an example, consider
\begin{equation} \label{eqn:Example Matrix}
	E_{1,2}^4 = 
	\begin{pmatrix} 0&0&0&0 \\ 0&0&1&0 \\ 0&0&0&0 \\ 0&0&0&0 \end{pmatrix}
	=
	\rho_1 \otimes \rho_2 .
\end{equation}
The goal of this preamble is to determine how to construct the Kronecker product given $(i,j)$ and $n$ (see Theorem 3 of \cite{GS2025_LCofThings} for a similar discussion). First, define the function $\mathcal{B}_\beta(d)$, which maps the base-ten number $d \le 2^\beta-1$ to a binary number with $\beta$ digits. For example, $\mathcal{B}_4(7)=0111$. Next, following \cite{Gunlycke2020,Demirdjian2025}, we define the function $\mathcal{F}:\{0,1\}^K\times\{0,1\}^K\to \{0,1,2,3\}^K$ such that $\mathcal{F}(i_K,j_K)=f_{K-1}\dots f_0$ where $f_k=2i_k+j_k$ for $k=0,\dots,K-1$, $i_K\coloneq i_{K-1}\dots i_0$ and $j_K\coloneq j_{K-1}\dots j_0$ for $i_k, j_k \in \{0,1\}$. 

Together, these functions can be used to map the decimal indices $(i,j)$ into the quaternary bitstring $f_{K-1}\dots f_0$. This enables the desired relation: $E_{i,j}^{n} = \rho_{\mathcal{F}(\mathcal{B}_{\log n}(i),\mathcal{B}_{\log n}(j))} \coloneq \rho_{f_{K-1}} \otimes \cdots \otimes\rho_{f_0}$. As a demonstration, consider that the non-zero entry in \eqref{eqn:Example Matrix} has $(i,j)=(1,2)$ and $n=4$. Then, $\mathcal{F}(\mathcal{B}_2(1),\mathcal{F}(\mathcal{B}_2(2))) = \mathcal{F}(01,10) = 12$. Therefore, $E_{1,2}^4 = \rho_{\mathcal{F}(\mathcal{B}_2(1),\mathcal{F}(\mathcal{B}_2(2)))} = \rho_{12} = \rho_1 \otimes \rho_2$, as desired. 

We now derive $A^{(\text{e})}$. By inserting \eqref{eqn:A} into the definition of $A^{(\text{e})}$ from \eqref{eqn:Ae_def}, we can see that 
\begin{equation} \label{eqn:Ae decomp}
	A^{(\text{e})} 
	\coloneq 
	\begin{pmatrix} 0_{N_0 \times N_0} & 0_{N_0 \times \Delta} \\ 0_{\Delta \times N_0} & A \end{pmatrix}
	= \sum_{j=1}^\alpha \hat{A}_j^j + \sum_{k=2}^{N_F} \sum_{j=1}^{\alpha-k+1} \hat{A}_{j+k-1}^j 
	+ \sum_{j=1}^{\alpha-1} \hat{A}_{j-1}^j ,
\end{equation}
with
\begin{equation} \label{eqn:Ahatjj}
	\begin{split}
		\hat{A}_j^j &\coloneq
		\left(
		\begin{array}{c|ccccc}
			0_{N_0 \times N_0} & & & & & \\ 
			\hline
			& 0_{N \times N} & & & & \\ & & \ddots & & & \\ & & & A_j^j & & \\ & & & & \ddots & \\ & & & & & 	0_{N^\alpha \times N^\alpha} 
		\end{array} \right)_{2N^\alpha \times 2N^\alpha} \\[4pt]
		&=
		\left(
		\begin{array}{ccccc}
			0_{N^j \times N^j} & & & & \\ & \ddots & & & \\ & & A_j^j & & \\ 
			& & & \ddots & \\ & & & & 0_{N^j \times N^j} 
		\end{array} \right)
		= E_{r_1,c_1}^{2N^{\alpha-j}} \otimes A_j^j ,
	\end{split}
\end{equation}
where the second equality is true because there is an integer multiple of the $0_{N^j \times N^j}$ on all sides of $A_j^j$, the third equality is true since each block is equal in size and $r_1 = c_1 = N^{\alpha-j} - \sum_{l=0}^{\alpha-j-1} N^l$.  Similarly, we have
\begin{equation} \label{eqn:Ahatjpkm1j}
	\begin{split}
		\hat{A}_{j+k-1}^j &\coloneq
		\left(
		\begin{array}{c|ccccccc}
			0_{N_0 \times N_0} & & & & & & & \\
			\hline
			& 0_{N \times N} & \cdots & 0_{N \times N^k} & & & & \\
			& & \ddots & & \ddots & & & \\
			& & & & & A_{j+k-1}^j & & \\
			& & & & & & \ddots & \\
			& & & & & & & 0_{N^{\alpha-k+1} \times N^\alpha} \\
			& & & & & & \ddots & \vdots \\
			& & & & & & & 0_{N^\alpha \times N^\alpha}
		\end{array} \right)_{2N^\alpha \times 2N^\alpha} \\[4pt]
		&= 
		\left(
		\begin{array}{ccccccc}
			0_{a \times a} & \cdots & 0_{a \times a} & & & & \\
			& \ddots & & \ddots & & & \\
			& & & & \begin{pmatrix} 0_{b \times a} \\ A_{j+k-1}^j \\ 0_{c \times a} \end{pmatrix} & & \\
			& & & & & \ddots & \\
			& & & & & & 0_{a \times a} \\
			& & & & & \ddots & \vdots \\
			& & & & & & 0_{a \times a}
		\end{array} \right)
		= E_{r_2,c_2}^{2N^{\alpha-j-k+1}} \otimes 
		\begin{pmatrix} 0_{b \times a} \\ A_{j+k-1}^j \\ 0_{c \times a}\end{pmatrix} ,
	\end{split}
\end{equation}
where $a=N^{j+k-1}$, $b=N^{j+k-1} - \sum_{l=j}^{j+k-2}N^l$, $c=\sum_{l=j+1}^{j+k-2}N^l$, $r_2 = N^{\alpha-j-k+1} - \sum_{l=0}^{\alpha-j-k} N^l - 1$ and $c_2 = N^{\alpha-j-k+1} - \sum_{l=0}^{\alpha-j-k} N^l$. Note, for $k=2$, then $c=0$. The equalities in \eqref{eqn:Ahatjpkm1j} are true by the same logic as in the relations from \eqref{eqn:Ahatjj}, however, $A_{j+k-1}^j$ requires zero padding to make it square. The reason the zero padding comes from both above and below is to ensure that there is an integer multiple of the $0_{a \times a}$ blocks on all sides of it. Without this condition, the third equality would not be possible. Finally, the last term in \eqref{eqn:Ae decomp} is obtained in a similar way as \eqref{eqn:Ahatjpkm1j}, but for the subdiagonal terms, yielding
\begin{equation} \label{eqn:Ahatjm1j}
	\hat{A}_{j-1}^j = \left( E_{r_2,c_2}^{2N^{\alpha-j+1}} \right)^T
	\otimes	\begin{pmatrix} 0_{N^j \times (N^j-N^{j-1})} & A_{j-1}^j \end{pmatrix} .
\end{equation}

We now use the method in the preamble to determine the form of $E_{r_1,c_1}^{2N^{\alpha-j}}$ and $E_{r_2,c_2}^{2N^{\alpha-j-k+1}}$. Each is accomplished in two parts: (1) convert the decimal indices $(i,j)$ into the binary indices $(\mathcal{B}_\beta(i),\mathcal{B}_\beta(j))$, and (2) evaluate $\rho_{\mathcal{F}(\mathcal{B}_\beta(i),\mathcal{B}_\beta(j))}$. For $N$ a power of two and integer $m$, a useful property for step (1) is 
\begin{equation} \label{eqn:binary prop}
	\mathcal{B}_\beta(N^m) = \underbrace{0 \dots 0\underbrace{10 \dots 0}_{\mathclap{m \log N+1}} }_\beta .
\end{equation}
Beginning with $E_{r_1,c_1}^{2N^{\alpha-j}}$, define $\beta_1 = \log 2N^{\alpha-j}$, then \begin{equation*}
	\begin{split}
		\mathcal{B}_{\beta_1}(r_1)
		&= \mathcal{B}_{\beta_1}(N^{\alpha-j} - 1 - N - \cdots - N^{\alpha-j-1}) \\
		&= \mathcal{B}_{\beta_1}(N^{\alpha-j}) - \mathcal{B}_{\beta_1}(1) - \mathcal{B}_{\beta_1}(N) - \cdots - \mathcal{B}_{\beta_1}(N^{\alpha-j-1}) \\
		&= \underbrace{0 1 \dots 1}_{ \mathclap{\beta_1} } 
		- \mathcal{B}_{\beta_1}(N) - \cdots - \mathcal{B}_{\beta_1}(N^{\alpha-j-1}) \\
		&= \underbrace{
			\underbrace{0 1 \dots 1}_{ \mathclap{\log N} } 
			\underbrace{0 1 \dots 1}_{ \mathclap{\log N} }
			\dots
			\underbrace{0 1 \dots 1}_{ \mathclap{\log N} }
			1 }_{\beta_1} ,
	\end{split}
\end{equation*}
where we have used $\mathcal{B}_{\beta_1}(N^{\alpha-j}) - \mathcal{B}_{\beta_1}(1) = 01\dots1$ in the third equality and successive applications of \eqref{eqn:binary prop} for the forth. Since $r_1=c_1$, we have 
\begin{equation*}
	\mathcal{F}(\mathcal{B}_{\beta_1}(r_1),\mathcal{B}_{\beta_1}(c_1))
	= \underbrace{
	\underbrace{0 3 \dots 3}_{ \mathclap{\log N} } 
	\underbrace{0 3 \dots 3}_{ \mathclap{\log N} }
	\dots
	\underbrace{0 3 \dots 3}_{ \mathclap{\log N} }
	3 }_{\beta_1} ,
\end{equation*}
and therefore
\begin{equation} \label{eqn:Er1c1}
	E_{r_1,c_1}^{2N^{\alpha-j}} = 
	\rho_{\mathcal{F}(\mathcal{B}_{\beta_1}(r_1),\mathcal{B}_{\beta_1}(c_1))} = 
	\left( \rho_0 \otimes \rho_3^{\otimes \log N-1} \right)^{\otimes \alpha - j} \otimes \rho_3.
\end{equation}

Next, we apply a similar procedure to $E_{r_2,c_2}^{2N^{\alpha-j-k+1}}$. Let $\beta_2=\log 2N^{\alpha-j-k+1}$, then 
\begin{equation*}
	\begin{split}
		\mathcal{B}_{\beta_2}(r_2) 
		&= \mathcal{B}_{\beta_2}(N^{\alpha-j-k+1} - 2 - N - \cdots - N^{\alpha-j-k}) \\
		&= \mathcal{B}_{\beta_2}(N^{\alpha-j-k+1}) - \mathcal{B}_{\beta_2}(2) - \mathcal{B}_{\beta_2}(N) 
		- \cdots - \mathcal{B}_{\beta_2}(N^{\alpha-j-k}) \\
		&= \underbrace{0 1 \dots 10}_{ \mathclap{{\beta_2}} } 
		- \mathcal{B}_{\beta_2}(N) - \cdots - \mathcal{B}_{\beta_2}(N^{\alpha-j-k}) \\
		&= \underbrace{
			\underbrace{0 1 \dots 1}_{ \mathclap{\log N} } 
			\underbrace{0 1 \dots 1}_{ \mathclap{\log N} }
			\dots
			\underbrace{0 1 \dots 1}_{ \mathclap{\log N} }
			0 }_{\beta_2} .
	\end{split}
\end{equation*}
Similarly, the column index in binary is 
\begin{equation*}
\begin{split}
	\mathcal{B}_{\beta_2}(c_2) 
	&= \mathcal{B}_{\beta_2}(N^{\alpha-j-k+1} - 1 - N - \cdots - N^{\alpha-j-k}) \\
	&= \mathcal{B}_{\beta_2}(N^{\alpha-j-k+1}) - \mathcal{B}_{\beta_2}(1) - \mathcal{B}_{\beta_2}(N) 
	- \cdots - \mathcal{B}_{\beta_2}(N^{\alpha-j-k}) \\
	&= \underbrace{0 1 \dots 1}_{ \mathclap{{\beta_2}} } 
	- \mathcal{B}_{\beta_2}(N) - \cdots - \mathcal{B}_{\beta_2}(N^{\alpha-j-k}) \\
	&= \underbrace{
	\underbrace{0 1 \dots 1}_{ \mathclap{\log N} } 
	\underbrace{0 1 \dots 1}_{ \mathclap{\log N} }
	\dots
	\underbrace{0 1 \dots 1}_{ \mathclap{\log N} }
	1 }_{\beta_2} .
\end{split}
\end{equation*}
This gives
\begin{equation*}
	\mathcal{F}(\mathcal{B}_{\beta_2}(r_2),\mathcal{B}_{\beta_2}(c_2))
	= \underbrace{
	\underbrace{0 3 \dots 3}_{ \mathclap{\log N} } 
	\underbrace{0 3 \dots 3}_{ \mathclap{\log N} }
	\dots
	\underbrace{0 3 \dots 3}_{ \mathclap{\log N} }
	1 }_{\beta_2} ,
\end{equation*}
and therefore
\begin{equation} \label{eqn:Er2c2}
	E_{r_2,c_2}^{2N^{\alpha-j-k+1}} = 
	\rho_{\mathcal{F}(\mathcal{B}_{\beta_2}(r_2),\mathcal{B}_{\beta_2}(c_2))} = 
	\left( \rho_0 \otimes \rho_3^{\otimes \log N-1} \right)^{\otimes \alpha-j-k+1} \otimes \rho_1 .
\end{equation}

Next, we use the following relation
\begin{equation} \label{eqn:0_A_0 prop}
	\begin{pmatrix} 0_{b \times a} \\ A_{j+k-1}^j \\ 0_{c \times a}\end{pmatrix}
	= (\mathcal{P}_k \otimes I_{N^j}) A_{j+k-1}^{(\text{e}),j} ,
\end{equation}
where we recall that $\mathcal{P}_k \in \mathbb{C}^{N^{k-1} \times N^{k-1}}$ is defined in Appendix \ref{sec:Pk circuit} and we have used the definition of $A_{j+k-1}^{(\text{e}),j}$ from \eqref{eqn:Aejpkm1j_def}. In a similar fashion, the zero padded block in \eqref{eqn:Ahatjm1j} may be written as
\begin{equation} \label{eqn:0A prop}
	\begin{pmatrix} 0_{N^j \times (N^j-N^{j-1})} & A_{j-1}^j \end{pmatrix}
	= A_{j-1}^{(\text{e}),j} (\mathcal{P}_2^T \times I_{N^{j-1}}) ,
\end{equation}
where we have used the definition of $A_{j-1}^{(\text{e}),j}$ from \eqref{eqn:Aejm1j_def}. 

Finally, by evaluating \labelcref{eqn:Er1c1,eqn:Er2c2,eqn:0_A_0 prop,eqn:0A prop,eqn:Ahatjj,eqn:Ahatjpkm1j,eqn:Ahatjm1j} into \eqref{eqn:Ae decomp}, we obtain
\begin{equation}
	\begin{split}
		A^{(\text{e})} 
		&= \sum_{j=2}^{\alpha} (\rho_0 \otimes \rho_3^{\otimes \log N-1})^{\otimes \alpha-j} 
			\otimes \rho_2 \otimes \left[ A^{(\text{e}),j}_{j-1} \left( \mathcal{P}_2^T \otimes I_{N^{j-1}} \right) \right] \\		
		&+ \sum_{j=1}^{\alpha} (\rho_0 \otimes \rho_3^{\otimes \log N-1})^{\otimes \alpha-j} 
			\otimes \rho_3 \otimes A^j_j \\  
		&+ \sum_{k=2}^{N_F} \sum_{j=1}^{\alpha-k+1} (\rho_0 \otimes \rho_3^{\otimes \log N-1})^{\otimes \alpha-k-j+1} 
			\otimes \rho_{1} \otimes \left[ \left( \mathcal{P}_k\otimes I_{N^j} \right) A^{(\text{e}),j}_{j+k-1} \right] .
	\end{split}
\end{equation}


\subsection{Derivation of $A_{j+k-1}^{(\textnormal{e}),j}$ from \eqref{eqn:Aejpkm1j}} \label{Aejpkm1j_Derivation}
Here, we derive $A_{j+k-1}^{(\text{e}),j}$ from \eqref{eqn:Aejpkm1j} for $j=\{1,\dots,\alpha-1\}$ and $k \in \{0,1,\dots,N_F\}$. By definition, we have
\begin{equation} \label{eqn:Ajpkm1j_com}
	\begin{split}
		A_{j+k-1}^j &\coloneq \sum_{l=0}^{j-1} I_{N^l} \otimes F_k \otimes I_{N^{j-l-1}} \\
		&= \sum_{l=0}^{j-1} \Bigl( K^{(N^l,N)} (F_k \otimes I_{N^l}) K^{(N^k,N^l)} \Bigr) 
		\otimes I_{N^{j-l-1}} ,
	\end{split}
\end{equation}
where $K^{(a,b)}\in\mathbb{C}^{ab\times ab}$ is the commutation matrix which has the useful property $K^{(r,m)}(A \otimes B)K^{(n,q)} = B \otimes A$ for $A \in \mathbb{C}^{m \times n}$ and $B \in \mathbb{C}^{r \times q}$ \cite{Xu2018}. Next, we evaluate \eqref{eqn:Ajpkm1j_com} into \eqref{eqn:Aejpkm1j_def} to obtain
\begin{equation} \label{eqn:Ajpkm1j_Derivation}
	\begin{split}
		A_{j+k-1}^{(\text{e}),j} 
		=
		\sum_{l=0}^{j-1}
		\begin{pmatrix}
			K^{(N^l,N)} (F_k \otimes I_{N^l}) K^{(N^k,N^l)} \\
			0_{(N^{k+l}-N^{l+1})\times N^{k+l}}
		\end{pmatrix} 
		\otimes I_{N^{j-l-1}} ,
	\end{split}
\end{equation}
where we have followed steps similar to Appendix B in \cite{Demirdjian2025}. Next, we simplify the matrix product terms by
\begin{equation} \label{eqn:MatProd}
	\begin{split}
		&\begin{pmatrix}
			K^{(N^l,N)} (F_k \otimes I_{N^l}) K^{(N^k,N^l)} \\
			0_{(N^{k+l}-N^{l+1})\times N^{k+l}}
		\end{pmatrix}
		=
		\Bigl( \rho_0^{\otimes \log N^{k-1}} \otimes K^{(N^l,N)} \Bigr) 
		\Bigl( F_k^{(\text{e})} \otimes I_{N^l} \Bigr)    
		K^{(N^k,N^l)} ,
	\end{split}
\end{equation}
where we once again followed steps similar to Appendix B in \cite{Demirdjian2025}. Finally, evaluate \eqref{eqn:MatProd} into \eqref{eqn:Ajpkm1j_Derivation} to give the full expression
\begin{equation} \notag
	A^{(\text{e}),j}_{j+k-1} =
	\sum_{l=0}^{j-1}
	\left[
	\left( \rho_0^{\otimes \log N^{k-1}} \otimes K^{(N^l,N)} \right) 
	\left(
	F^{(\text{e})}_k 
	\otimes I_{N^l}
	\right)
	K^{(N^k,N^l)} \right]
	\otimes I_{N^{j-l-1}} .
\end{equation}

	
\section{Condition Number Analysis} \label{sec:Condition Number}
The DVBE from \eqref{eqn:DVBE} may be rewritten in the more general form 
\begin{equation} \label{eqn:DVBE r0}
	\frac{\partial f_m}{\partial t} + \vec{c}_m\cdot\nabla f_m = - \frac{1}{\tau}(f_m-f_m^\text{(eq)}) ,
\end{equation}
where the difference here is that we have used the more general velocities $\vec{c}_m = r_0 \frac{\Delta x}{\Delta t} \vec{e}_m$ for $r_0 \in \mathbb{R}$. In LBE applications, one typically chooses $r_{0}=1$ so that the characteristics pass exactly through lattice points. This allows one to derive the standard LBE via integration along characteristics. Here, however, we adopt a finite difference scheme independent of this procedure and, so, have more freedom in selecting $r_{0}$ (this is sometimes referred to as the off-lattice Boltzmann method, see for example Section 3.2 of \cite{horstmann2018hybrid}). Note in this new context we must redefine the following: 
\begin{equation*}
	c_{s}^{2} = r_{0}^{2}\frac{\left(\Delta x\right)^{2}}{\left(\Delta t\right)^{2}} c_s^{\star 2} 
	, \qquad
	\tau = r_0^{-2} \frac{(\Delta t)^2}{(\Delta x)^2} \frac{\nu}{c_s^{\star 2}} .
\end{equation*}
In the following analysis, we consider the matrix $L^{(e)}$ obtained after applying the zero padded Carleman Linearization from Section \ref{sec:Carl Lin Zero Pad} to this version of the DVBE. Specifically, our goal is to compare $\kappa(L^{(\text{e})})$ and $\kappa(L)$. 

Observe that, the matrix $L^{(e)}$ is unitarily equivalent to 
\begin{equation*}
	\tilde{L}^{\text{(e)}} \coloneqq \begin{pmatrix}
		V_{n_{t}}\otimes I_{N_{0}} & \\ 
		& L
	\end{pmatrix} \, , 
\end{equation*}
where $ V_{n_{t}} = \text{tridiag}{(-1,1,0)}$ of size $n_{t}\times n_{t}$, and therefore $\tilde{L}^{(\text{e})}$ must have the same condition number as $L^{(\text{e})}$. It follows from the block structure of $\tilde{L}^\text{(e)}$, the properties of $\|\cdot\|_{2}$, and the definition of the condition number that 
\begin{equation} \label{eqn:keq1}
	\kappa(L^{(e)}) = \max\{ \|V_{n_{t}}\|_{2}, \|L\|_{2} \} \cdot \max\{\|V_{n_{t}}^{-1}\|_{2}, \|L^{-1}\|_{2}\} .
\end{equation}
Direct calculation shows that $\|V_{n_{t}}\|_{2} \sim 2$ and $\|V^{-1}_{n_{t}}\|_{2} \sim (n_{t}+1)/\pi$. Let $B \coloneqq I -  A$. It follows from the structure of $L$ and the subadditivity of $\|\cdot \|_{2}$ that  
\begin{align}\label{eqn:keq2}
	\begin{split}
		\sqrt{1+\|B\|_{2}^{2}} & \leq \|L\|_{2} \leq \max\{2, 1+\|B\|_{2}\} \,\\ 
		\sqrt{\sum_{k=0}^{n_{t}-1}\|B^{-1}\|_{2}^{2k}} & \leq \|L^{-1}\|_{2} \leq \sum_{k=0}^{n_{t}-1}\|B^{-1}\|_{2}^{k} .
	\end{split}
\end{align}
Using the fact that $\max\{2,1+ \|B\|_2\} \le 2(1+\|B\|_2)$, equations \eqref{eqn:keq1} and \eqref{eqn:keq2} show that 
\begin{equation} \label{eqn:upperLowerKest}
	\left(\sqrt{1+\|B\|^{2}_{2}} \right) \max\{\|V_{n_{t}}^{-1}\|_{2}, \|L^{-1}\|_{2}\} \leq \kappa(L^{(\text{e})}) \leq 2\left(1+\|B\|_{2}\right)\max\{\|V_{n_{t}}^{-1}\|_{2}, \|L^{-1}\|_{2}\} \, . 
\end{equation}
Hence, we are left to estimate $\|L^{-1}\|_{2}$. To this end, we require the logarithmic norm of a matrix $M$ defined by 
\begin{equation*}
	\mu_p(M) \coloneq \lim_{\epsilon \to 0^{+}} \frac{\| I + \epsilon M \|_p - 1}{\epsilon} .
\end{equation*}
Equipped with this, we use the following inequality: 
\begin{equation}\label{eqn:logInvEst}
	\|B^{-1}\|_{2} \leq \frac{1}{\max\{-\mu_{2}(-B), -\mu_{2}(B)\} } \, , 
\end{equation}
which is valid provided either $\mu_{2}(B)<0$ or $\mu(-B)<0$ (see Proposition 2.3 of \cite{soderlind2006logarithmic}). Let $D$ be a matrix whose diagonal blocks are exactly $A^{i}_{i}$ and whose off-diagonal blocks are trivial. Furthermore, let $N = A-D$ and let the maximum eigenvalue of an arbitrary matrix $M$ be defined by $\lambda_\text{max}(M)$. It follows from the definition of $\mu_{2}(\cdot)$ that 
\begin{align}\label{eqn:mu1}
	\begin{split}
		\mu_{2}(A) \coloneqq \lim_{\epsilon \to 0^{+}}\frac{\|I+\epsilon D +\epsilon N\|_{2}-1}{\epsilon} &\leq \|N\|_{2} +\lim_{\epsilon \to 0^{+}}\frac{\|I+\epsilon D \|_{2}-1}{\epsilon} = \mu_{2}(D) + \|N\|_{2}\\
		& \leq \lambda_{\text{max}}\left(\frac{D+D^{\dagger}}{2}\right) + \Delta t \alpha\left(\|F_{2}\|_{2}+\|F_{3}\|_{2}\right) ,
	\end{split}
\end{align}
where in the last inequality we have used the subadditivity property and the fact that $\mu_{2}(M) = \lambda_{\text{max}}\left(\frac{M+M^{\dagger}}{2}\right)$ for an arbitrary matrix $M$. Moreover, from \eqref{eqn:S_etapm} we observe that $S_{\eta}+S_{\eta}^{\dagger}=0$ for $\eta \in \{x,y,z\}$ and the eigenvalues of $\left(A_{j}^{j}+(A_{j}^{j})^{\dagger}\right)/2$ for $j \in \{1,\dots,\alpha\}$ are simply $\{\lambda_{i_{1}}+ \cdots + \lambda_{i_{k}} \;|\; \lambda_{i_{j}} \in \sigma((R+R^{\dagger})/2) \}$ where $\sigma(M)$ is the spectrum of $M$. These observations may now be combined with \eqref{eqn:mu1} to show that 
\begin{equation}\label{eqn:general_mu_ineq}
	\mu_{2}(A) \leq \begin{cases}
		\Delta t \left(\lambda_{\text{max}}\left(\frac{R+R^{\dagger}}{2}\right) + \alpha\left(\|F_{2}\|_{2}+\|F_{3}\|_{2}\right) \right) \qquad \lambda_\text{max}(R+R^{\dagger}) < 0 \\
		\Delta t \alpha \left(\lambda_{\text{max}}\left(\frac{R+R^{\dagger}}{2}\right) + \left(\|F_{2}\|_{2}+\|F_{3}\|_{2}\right)\right) \qquad \lambda_\text{max}(R+R^{\dagger}) \ge 0 \, .
	\end{cases}
\end{equation}
Motivated by \eqref{eqn:general_mu_ineq}, we take 
\begin{equation}
	\beta_{1} \coloneqq \begin{cases}
		\left(\lambda_{\text{max}}\left(\frac{R+R^{\dagger}}{2}\right) + \alpha\left(\|F_{2}\|_{2}+\|F_{3}\|_{2}\right) \right) \qquad \lambda_\text{max}(R+R^{\dagger}) < 0 \\
		\alpha \left(\lambda_{\text{max}}\left(\frac{R+R^{\dagger}}{2}\right) + \left(\|F_{2}\|_{2}+\|F_{3}\|_{2}\right)\right) \qquad \lambda_\text{max}(R+R^{\dagger}) \ge 0 \, ,
	\end{cases}
\end{equation}
so that 
\begin{equation}\label{eqn:logEst2}
	\mu(-B) = -1 +\mu(A) \leq -1+\Delta t \beta_{1} ,
\end{equation}
in all cases. Equation \eqref{eqn:logEst2} ensures that $\mu(-B)<0$ for $\Delta t$ small enough (relative to $\tau$, $\alpha$, and the norms of each $F_{i}$ for standard velocity sets), and consequently that \eqref{eqn:logInvEst} is valid with $-\mu_{2}(-B)$. Note that $\mu_{2}(-B)<0$ implies that $\mu_{2}(B)>0$.

Observe that the $r_{0}$ parameter appears in the finite difference approximation of spatial derivatives in \eqref{eqn:DVBE r0}. Hence, 
\begin{equation} \label{eqn:A_2_est}
	1+\|A\|_{2} \le 1 + \Delta t \beta_{2} ,
\end{equation}
where 
\begin{equation*}
	\beta_2 = \frac{\alpha r_{0}}{\Delta t}+\alpha \sum_{i=1}^{3} \|F_{i}\|_{2} . 
\end{equation*}
Equations \eqref{eqn:logInvEst}, \eqref{eqn:keq2}, \eqref{eqn:logEst2}, and \eqref{eqn:A_2_est} can be combined to show 
\begin{align} \label{eqn:keq3}
	\begin{split}
		\frac{1}{1+\beta_{2}\Delta t} \leq \frac{1}{\|B\|_{2}} & \leq \|B^{-1}\| \leq \frac{1}{1-\beta_{1}\Delta t} \\ 
		\sqrt{\sum_{k=0}^{n_{t}-1}\left(\frac{1}{1+\beta_{2}\Delta t}\right)^{2k}} &\leq \|L^{-1}\|_{2} \leq \frac{1-\left(\frac{1}{1-\Delta t\beta_{1}}\right)^{n_{t}-1}}{1-\frac{1}{1-\Delta t \beta_{1}}} \, . 
	\end{split}
\end{align}
In the second inequality of \eqref{eqn:keq3} we want to give a lower bound for 
\begin{equation}
	\sqrt{\sum_{k=0}^{n_{t}-1}\left(\frac{1}{1+\beta_{2}\Delta t}\right)^{2k}} = \sqrt{\frac{1-\left(1+\beta_{2}\Delta t\right)^{-2n_{t}}}{1-(1+\beta_{2}\Delta t)^{-2}}} . 
\end{equation}
First, notice that 
\begin{equation} \label{eqn:ln_ineq_0}
	\frac{1}{1-\left(1+\beta_{2}\Delta t\right)^{-2}} 
	= \frac{\left( 1 + \beta_{2} \Delta t \right)^{2}}
	{2\beta_{2}\Delta t +(\beta_{2} \Delta t)^{2}} \geq \frac{1+\beta_{2}\Delta t}{2\beta_{2}\Delta t} 
	= \frac{n_{t}(1+\beta_{2}\Delta t)}{2\beta_{2}T} ,
\end{equation}
where the inequality follows from the fact that $\beta_2 > 0$ and in the last equality we use $T=n_t \Delta t$. Next, we need to make use of the basic inequality 
\begin{equation*}
	\ln(1+u) = \int_{0}^{u}\frac{1}{1+t} \, dt \geq \frac{u}{1+u} \qquad \text{for} \qquad u>0 ,
\end{equation*}
which is equivalent to  
\begin{equation}\label{eqn:ln_ineq_2}
	\frac{1}{1+u} \leq \textup{exp}\left(-\frac{u}{1+u}\right) \qquad \text{for} \qquad u>0. 
\end{equation}
Let $u = \beta_{2}\Delta t$ and raise both sides of \eqref{eqn:ln_ineq_2} to the power $2n_{t}$ to obtain 
\begin{equation} \label{eqn:ln_ineq_3}
	\left(1+\beta_{2}\Delta t\right)^{-2n_{t}} \leq \textup{exp}\left(-\frac{2\beta_{2}T}{1+\beta_{2}\Delta t}\right) \, . 
\end{equation}
Now, \eqref{eqn:ln_ineq_3} can be combined with \eqref{eqn:ln_ineq_0} to show 
\begin{equation} \label{eqn:ln_ineq_4}
	\frac{1-\left(1+\beta_{2}\Delta t\right)^{-2n_{t}}}{1-(1+\beta_{2}\Delta t)^{-2}} \geq \frac{n_{t}(1+\beta_{2}\Delta t)}{2\beta_{2}T}\left(1-\textup{exp}\left(-\frac{2\beta_{2}T}{1+\beta_{2}\Delta t}\right)\right)\, .  
\end{equation}
Note that the right hand side of \eqref{eqn:ln_ineq_4} has the structure 
\begin{equation*}
	n_{t} \left(\frac{1-e^{-x}}{x}\right) ,\qquad x = \frac{2\beta_{2}T}{1+\beta_{2}\Delta t} ,
\end{equation*}
and $(1-e^{-x})/x$ is decreasing in $x$. Since $2\beta_{2}T/(1+\beta_{2}\Delta t) < 2\beta_{2}T$, it follows that 
\begin{equation} \label{eqn:ln_ineq_6}
	\frac{1-\left(1+\beta_{2}\Delta t\right)^{-2n_{t}}}{1-(1+\beta_{2}\Delta t)^{-2}} \geq \frac{n_{t}}{2\beta_{2}T}\left(1-e^{-2\beta_{2}T}\right) \, . 
\end{equation}
Finally, by inserting the square root of \eqref{eqn:ln_ineq_6} into the second inequality of \eqref{eqn:keq3}, we obtain
\begin{equation}\label{eqn:keq4}
	\sqrt{n_{t}}\sqrt{\frac{1-e^{-\beta_{2}T}}{2\beta_{2}T}}\leq \|L^{-1}\|_{2} \leq \frac{1-\left(\frac{1}{1-\Delta t\beta_{1}}\right)^{n_{t}-1}}{1-\frac{1}{1-\Delta t \beta_{1}}} ,
\end{equation}
so that equations \eqref{eqn:upperLowerKest} and \eqref{eqn:keq4} can be combined to yield 
\begin{equation}
	\kappa(L) \leq \kappa(L^{(e)}) \leq \frac{2}{\pi} \kappa(L) \frac{(n_{t}+1)}{\|L^{-1}\|_{2}} \leq \frac{4\sqrt{n_{t}}}{\pi}\left(\sqrt{\frac{2\beta_{2}T}{1-e^{-\beta_{2}T}}}\right)\kappa(L)\, . 
\end{equation}
Thus, in summary, the ratio of the condition number of $L$ to $L^{(e)}$ is 
\begin{equation*}
	1 \leq \frac{\kappa(L^{(e)})}{\kappa(L)} \leq \frac{4\sqrt{n_{t}}}{\pi}\left(\sqrt{\frac{2\beta_{2}T}{1-e^{-\beta_{2}T}}}\right) ,
\end{equation*}
under the assumption $\beta_{1}\Delta t <1$. For fixed $\beta_{2}T$, we can write $1 \le \kappa(L^{(e)}) \le O(\sqrt{n_{t}}\kappa(L))$.


\section{Circuit Constructions for Relevant Matrices}


\subsection{The $\mathcal{P}_k$ Matrix} \label{sec:Pk circuit}
The only requirement of $\mathcal{P}_k \in \mathbb{C}^{N^{k-1} \times N^{k-1}}$ is that the element in the $(N^{k-1} - \sum_{l=0}^{k-2} N^l)\text{th}$ row and the $0\text{th}$ column be equal to 1. Here, we show that such $\mathcal{P}_k$ can be constructed using at most $\log N^{k-1}$ NOT gates. Consider that a permutation matrix with a $1$ in row $b \in \mathbb{N}_0$ and column $0$ will perform the mapping $\ket{0} \rightarrow \ket{b}$. So, we seek a matrix that performs this operation where $b=N^{k-1} - \sum_{l=0}^{k-2} N^l$. Define, $\mathcal{X}_c(b)=\bigotimes_{j=c-1}^0 X^{b_j}$, where $b$ is a decimal number with binary form $b_{c-1} \dots b_0$ for $b_j \in \{0,1\}$ and $c$ digits, and $X$ is the Pauli-X matrix. Then, $\mathcal{X}_c(b)\ket{0}^{\otimes c} = \ket{b}$ as desired. Thus, we may take $\mathcal{P}_k = \mathcal{X}_{\log N^{k-1}}\left( N^{k-1} - \sum_{l=0}^{k-2} N^l \right)$, which requires at most $\log N^{k-1}$ NOT gates.


\subsection{The $M_{m+1}^r$ Matrix} \label{sec:Construction of M}
Let $m = 2^{q_{m}}$ and $r = 2^{q_{r}}$, with $r\leq m$, for some integers $q_{m}$ and $q_{r}$. Consider a permutation matrix $M^{r}_{m+1} \in \mathbb{C}^{rm\times rm}$ where the first $r$-rows each have one nonzero element (equal to $1$) located at the $i(m+1)\text{th}$ column for $i =0,1,\dots,(r-1)$. We can then write $M_{m+1}^r=\begin{pmatrix} P \\ P^\perp\end{pmatrix}$, where $P(i,j)=1$ for $(i,j) \in \{(0,0),(1,m+1),(2,2(m+1)),\dots,(r-1,(r-1)(m+1))\}$ and $P^\perp$ is a non-square permutation matrix whose row space is orthogonal to $P$. Note, while similar, $P^\perp$ is not the unitary complement from Definition \ref{def:completion}. Our goal is to construct $M_{m+1}^r$ such that
\begin{equation} \label{eqn:Mm1r decimal}
	M_{m+1}^r \ket{i(m+1)}_{q_m+q_r} = {\ket{i}}_{q_m+q_r} ,
\end{equation}
for $i\in\{0,\dots,r-1\}$. Here, we have used the notation $\ket{i}_q=\ket{x_{q-1} \dots x_0}$ for decimal number $i$, positive integer $q$ and  $x_j\in\{0,1\}$. Next, we use the property that $\ket{y}_{q_y}\ket{x}_{q_x}=\ket{x+2^{q_x}y}_{q_x+q_y}$ where $x \in \{0,\dots,2^{q_x}-1\}$ and $y \in \{0,\dots,2^{q_y}-1\}$. Applying this property twice gives
\begin{equation}
\begin{split}
	\ket{i(m+1)}_{q_m+q_r} 
	&= \ket{i+i2^{q_m}}_{q_m+q_r} \\
	&= \ket{i}_{q_r}\ket{i}_{q_m} \\
	&= \ket{i}_{q_r}\ket{0}_{q_m-q_r}\ket{i}_{q_r} .
\end{split}
\end{equation}
With this, \eqref{eqn:Mm1r decimal} may be written in the more amenable form 
\begin{equation}
	M_{m+1}^r (\ket{i}_{q_r} \ket{0}_{q_m-q_r} \ket{i}_{q_r})
	= \ket{i}_{q_m + q_r} .
\end{equation}
By modifying the inline multiplication circuit introduced in \cite{Gidney2017}, we claim the following defines a matrix $M_{m+1}^r$ with the desired property \eqref{eqn:Mm1r decimal}
\begin{equation} \label{eqn:Mmp1r}
	M_{m+1}^r = \prod_{i=0}^{q_r-1} CX(i,i+q_m) ,
\end{equation}
where $CX(i,j)$ is the CNOT gate conditioned on qubit $c$ and targeted on qubit $t$. Here, we use little-endian notation whereby the least significant bit corresponds to the top-most wire of the circuit. With this framework, a CNOT has the form $CX(c,t)\ket{t}\ket{c} = \ket{t \oplus c}\ket{c}$. Then, by sequentially apply the CNOT gates from \eqref{eqn:Mmp1r} we have
\begin{equation}
	\begin{split}
		\ket{i}_{q_r} \ket{0}^{\otimes q_m-q_r} \ket{i}_{q_r}
		&\xrightarrow{CX(q_r-1,q_m+q_r-1)} \ket{x_{q_r-1} \oplus x_{q_r-1}} \ket{x_{q_r-2} \dots x_0} \ket{0}^{\otimes q_m-q_r} \ket{i}_{q_r} \\
		&\xrightarrow{CX(q_r-2,q_m+q_r-2)} \ket{0} \ket{x_{q_r-2} \oplus x_{q_r-2}} \ket{x_{q_r-3} \dots x_0} \ket{0}^{\otimes q_m-q_r} \ket{i}_{q_r} \\
		&\vdots \\
		&\xrightarrow{CX(0,q_m)} \ket{0}^{\otimes q_r-1} \ket{x_0 \oplus x_0} \ket{0}^{\otimes q_m-q_r} \ket{i}_{q_r} \\
		&= \ket{0}^{\otimes q_m} \ket{i}_{q_r}
		= \ket{i}_{q_m+q_r} ,  
	\end{split}
\end{equation}
where we have used the property that $\ket{x_i \oplus x_i}=\ket{0}$.


\subsection{The $\overline{B}_{2,q}$ Matrix} \label{sec:Circuit for B2q}
From \eqref{eqn:Bq}, we can see that the submatrix $B_{2,q}$ has exactly one nonzero element per row and that these elements are 
\begin{enumerate}
	\item spaced every $Qn+1$ elements apart, \label{Bq Cond 1}
	\item shifted forward by $(q-1)n$ elements. \label{Bq Cond 2}
\end{enumerate}
Putting these requirements together, it follows that $B_{2,q}(i,j)=1$ where
\begin{equation}
	(i,j) \in \{(0,b),(1,a+b),(2,2a+b),\dots,(n-1,(n-1)a+b)\} ,
\end{equation}
where $a=Qn+1$, $b=(q-1)n$, $q \in \{1,\dots,Q\}$, $Q=2^{q_Q}$ and $n=2^{q_n}$. We must therefore find an operator to perform the mapping
\begin{equation} \label{eqn:B2q mapping}
	\ket{ia+b}_{2q_n+q_Q} \rightarrow \ket{i}_{2q_n+q_Q} ,
\end{equation}
where $i \in \{0,\dots,n-1\}$. Using the property $\ket{y}_{q_y}\ket{x}_{q_x} = \ket{x+y2^{q_x}}_{q_x+q_y}$, the LHS of \eqref{eqn:B2q mapping} may be written in the more amenable form
\begin{equation}
\begin{split}
	\ket{ia+b}_{2q_n+q_Q} 
	&= \ket{Qni+(q-1)n+i}_{2q_n+q_Q} \\
	&= \ket{i}_{q_n}\ket{(q-1)n+i}_{q_n+q_Q} \\
	&= \ket{i}_{q_n}\ket{q-1}_{q_Q}\ket{i}_{q_n} .
\end{split}
\end{equation}
With this, an equivalent statement to \eqref{eqn:B2q mapping} is
\begin{equation} \label{eqn:B2q mapping equiv}
	\ket{i}_{q_n}\ket{q-1}_{q_Q}\ket{i}_{q_n} \rightarrow \ket{i}_{2q_n+q_Q} .
\end{equation}
We can achieve the desired mapping with the following operations
\begin{equation}
	\begin{split}
		\ket{i}_{q_n}\ket{q-1}_{q_Q}\ket{i}_{q_n} 
		\xrightarrow{M_{Qn+1}^n} & 
		\ket{0}_{q_n}\ket{q-1}_{q_Q}\ket{i}_{q_n} \\
		\xrightarrow{I_n \otimes \mathcal{X}_{q_Q}(q-1) \otimes I_n} &
		\ket{0}_{q_n}\ket{0}_{q_Q}\ket{i}_{q_n}  
		= \ket{i}_{2q_n+q_Q} ,
	\end{split}
\end{equation}
where $M_{m+1}^r$ is defined in Appendix \ref{sec:Construction of M} and $\mathcal{X}_c(b)$ is defined in Appendix \ref{sec:Pk circuit}. Since these operators produce the desired output, it follows that 
\begin{equation} \label{eqn:B2q circ eq}
	\overline{B}_{2,q}
	= \left( I_n \otimes \mathcal{X}_{q_Q}(q-1) \otimes I_n \right) 
	\cdot M_{Qn+1}^n  ,
\end{equation}
defines a matrix with the desired property \eqref{eqn:B2q mapping equiv}. The total resource cost of this circuit is exactly $\log n$ CNOT gates for the first operation and at most $\log Q$ NOT gates for the second as shown in Figure \ref{fig:Circuit B2q}.

\begin{figure}
	\centering
	\includegraphics[]{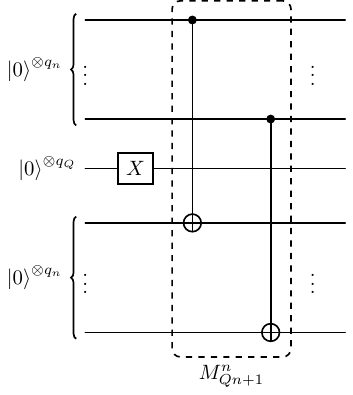}
	\caption{The $\overline{B}_{2,q}$ circuit from \eqref{eqn:B2q circ eq} for $q=Q$. The NOT gate on the $q_Q$ register should be interpreted as being applied to every wire within that register.}
	\label{fig:Circuit B2q}
\end{figure}


\subsection{The $\overline{B}_{3,q}$ Matrix} \label{sec:Circuit for B3q} 
From \eqref{eqn:B3q}, we can see that the submatrix $B_{3,q}$ has $Q$ nonzero elements equal to $1$ per row, and that 
\begin{enumerate} 
	\item in each row the $1$'s are spaced by $Qn^2$, \label{Cond B3q Qn2}
	\item in adjacent rows the pattern of $1$'s are spaced by $(Qn)^2+Qn+1$, \label{Cond B3q Qn2+Qn+1}
	\item shifted forward by $(q-1)n$ elements. \label{Cond B3q adder}
\end{enumerate}
Putting these requirements together, it follows that $B_{3,q}(i,j)=1$ where 
\begin{equation}
	\begin{split}
		(i,j)=\{&(0,c),(0,a+c),\dots,(0,(Q-1)a+c), \\
		&(1,b+c),(1,a+b+c),\dots,(1,(Q-1)a+b+c), \\
		&(2,2b+c),(2,a+2b+c),\dots,(2,(Q-1)a+2b+c), \\
		&\vdots \\
		&(n-1,(n-1)b+c),(n-1,a+(n-1)b+c),\dots,(n-1,(Q-1)a+(n-1)b+c) \} ,
	\end{split}
\end{equation}
where $a=Qn^2$, $b=(Qn)^2+Qn+1$ and $c=(q-1)n$ for $q \in \{1,\dots,Q\}$, $Q=2^{q_Q}$ and $n=2^{q_n}$. We must find an operator to perform
\begin{equation} \label{eqn:abc to i}
	d\sum_{j=0}^{Q-1} \ket{ja + ib + c}_{3q_n+2q_Q} \rightarrow \ket{i}_{3q_n+2q_Q} ,
\end{equation}
where $d \in \mathbb{C}$ is a normalization coefficient and $i \in \{0,\dots,n-1\}$. Repeatedly applying the property $\ket{y}_{q_y}\ket{x}_{q_x} = \ket{x+y2^{q_x}}_{q_x+q_y}$ from Appendix \ref{sec:Construction of M}, the LHS may be written in the more amenable form 
\begin{equation}
	\begin{split}
		\ket{aj + bi + c} 
		&= \ket{(Qn)^2i+Qn^2j+Qni+(q-1)n+i}_{3q_n+2q_Q} \\
		&= \ket{i}_{q_n}\ket{Qn^2j+Qni+(q-1)n+i}_{2q_n+2q_Q} \\
		&= \ket{i}_{q_n}\ket{j}_{q_Q}\ket{Qni+(q-1)n+i}_{2q_n+q_Q} \\
		&= \ket{i}_{q_n}\ket{j}_{q_Q}\ket{i}_{q_n}\ket{(q-1)n+i}_{q_n+q_Q} \\
		&= \ket{i}_{q_n}\ket{j}_{q_Q}\ket{i}_{q_n}\ket{q-1}_{q_Q}\ket{i}_{q_n} .
	\end{split}
\end{equation}
With this, an equivalent statement to \eqref{eqn:abc to i} is
\begin{equation} \label{eqn:B3q map}
	d\sum_{j=0}^{Q-1}
	\ket{i}_{q_n} \ket{j}_{q_Q} \ket{i}_{q_n} \ket{q-1}_{q_Q} \ket{i}_{q_n}
	\rightarrow
	\ket{i}_{3q_n+2q_Q} .
\end{equation}
To find an operator to perform this mapping, we will need to use the property
\begin{equation}
	2^{-(\log N)/2} \prod_{k=0}^{\log N -1} H_k \left( \sum_{j=0}^{N-1} \ket{j} \right)
	= \ket{0}^{\otimes \log N} ,
\end{equation}
where $H_k$ is the Hadamard gate applied to the $k\text{th}$ qubit. With this, we have the following relations:
\begin{equation}
	\begin{split}
		\sum_{j=0}^{Q-1}
		\ket{i}_{q_n} \ket{j}_{q_Q} \ket{i}_{q_n} \ket{q-1}_{q_Q} \ket{i}_{q_n} 
		\xrightarrow{M_{(Qn)^2+1}^n} 
		&\sum_{j=0}^{Q-1}
		\ket{0}_{q_n} \ket{j}_{q_Q}\ket{i}_{q_n} \ket{q-1}_{q_Q} \ket{i}_{q_n} \\  
		\xrightarrow{I_{Qn} \otimes M_{Qn+1}^n} 
		&\sum_{j=0}^{Q-1}
		\ket{0}_{q_n} \ket{j}_{q_Q} \ket{0}_{q_n} 	\ket{q-1}_{q_Q} \ket{i}_{q_n} \\
		\xrightarrow{2^{-q_Q/2}\prod_{k=2q_n + q_Q}^{2(\log Qn)-1} H_k} 
		& \ket{0}_{2q_n + q_Q} \ket{q-1}_{q_Q} \ket{i}_{q_n} \\
		\xrightarrow{I_{Qn^2} \otimes \mathcal{X}_{q_Q}(q-1) \otimes I_n} 
		&\ket{0}_{2q_n + 2q_Q} \ket{i}_{q_n} 
		= \ket{i}_{3q_n + 2q_n} ,
	\end{split}
\end{equation}
where $M_{m+1}^r$ is defined in Appendix \ref{sec:Construction of M} and $\mathcal{X}_c(b)$ is defined in Appendix \ref{sec:Pk circuit}. Since these operators produce the desired output, it follows that 
\begin{equation} \label{eqn:B3q circ eq}
\begin{split}
	\overline{B}_{3,q}
	&= \left( I_{Qn^2} \otimes \mathcal{X}_{q_Q}(q-1) \otimes I_n \right)
	\cdot \left( I_{Qn} \otimes M_{Qn+1}^n \right)
	\cdot M_{(Qn)^2+1}^n
	\cdot \prod_{k=2q_n + q_Q}^{2\log(Qn) -1} H_k \\
	&= \left( I_{Qn} \otimes \overline{B}_{2,q} \right)
	\cdot M_{(Qn)^2+1}^n
	\cdot \prod_{k=2q_n + q_Q}^{2\log(Qn) -1} H_k .
\end{split}
\end{equation}
Note that $\overline{B}_{3,q}$ must be multiplied by a factor of $2^{-q_Q/2}$ to offset the Hadamard coefficients. The total resource cost of this circuit is exactly $2\log n$ CNOT gates for the combined $M_{m+1}^r$ operations, at most $\log Q$ NOT gates for the $\mathcal{X}_{q_Q}(q-1)$ operation and exactly $\log Q$ Hadamard gates as shown in Figure \ref{fig:Circuit B3q}.

\begin{figure}
	\centering
	\includegraphics[]{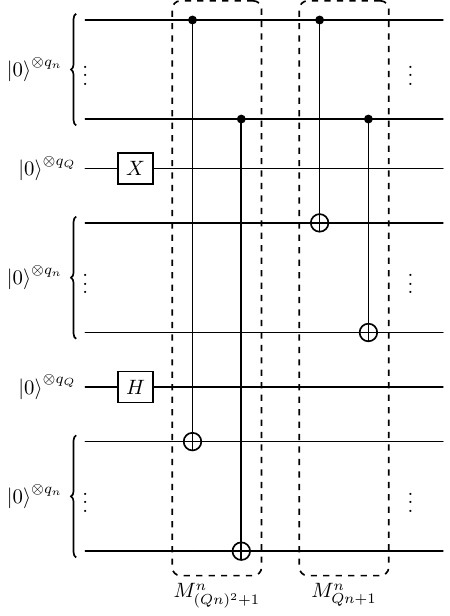}
	\caption{The $\overline{B}_{3,q}$ circuit from \eqref{eqn:B3q circ eq} for $q=Q$. A single qubit gate applied on a multi-qubit register should be interpreted as being applied to each individual wire within that register.}
	\label{fig:Circuit B3q}
\end{figure}


\subsection{The Commutation Matrix} \label{sec:Circ for Commutation}

Following \cite{Xu2018}, the commutation matrix is defined by 
\begin{equation}
	K^{(a,b)} = \sum_{i=1}^a \sum_{j=1}^b (\vec{e}_{a,i} \otimes \vec{e}_{b,j})(\vec{e}_{a,j} \otimes \vec{e}_{b,i})^T ,
\end{equation}
where $\vec{e}_{a,j}$ denotes the $j\text{th}$ cononical vector of dimension $a$ (i.e. a column vector with $1$ in the $j\text{th}$ row and $0$ elsewhere). From \cite{Demirdjian2025}, if $a=2^{q_a}$ and $b=2^{q_b}$ for integers $q_a$ and $q_b$, then the quantum circuit for the commutation matrix is 
\begin{equation}
	K^{(a,b)} = \prod_{i=0}^{q_a-1} \prod_{j=0}^{q_b-1} \text{SWAP}(i+q_b-j-1,i+q_b-j) ,
\end{equation}
So, $K^{(a,b)}$ requires $\log a \log b$ SWAP gates.


\section{Velocity Sets} \label{sec:Velocity Sets}
Conventionally, a lattice with $n$ spatial dimensions and $m$ speeds is denoted as D($n$)Q($m$). Here we summarize the D1Q3, D2Q9 and D3Q15 velocity sets along with their respective weights. We also provide a method to zero pad velocity sets that are not powers of $2$. 

\subsection{The D1Q3, D2Q9, and D3Q15 Velocity Sets} \label{sec:Specific Velocity Sets} 
Here, we provide the discrete velocities (in lattice units) and weights for the D1Q3, D2Q9 and D3Q15 cases. Following \cite{Kruger2017}, the discrete velocities and their weights for the D1Q3 case are 
\begin{equation}
	\left[\vec{e}_m,w_m\right] \vert_\text{D1Q3}  
	\begin{cases} 
		\left[ (0), \frac{2}{3} \right],    & \text{for } m\in\{1\} \\
		\left[ (\pm1), \frac{1}{6} \right], & \text{for } m\in\{2,3\} ,
	\end{cases}
\end{equation}
for the D2Q9 case is
\begin{equation}
	\left[\vec{e}_m,w_m\right] \vert_\text{D2Q9} =  
	\begin{cases} 
		\left[ (0,0), \frac{4}{9} \right],    & \text{for } m\in\{1\} \\
		\left[ (\pm1,0), \frac{1}{9} \right], & \text{for } m\in\{2,3\}\\
		\left[ (0,\pm1), \frac{1}{9} \right], & \text{for } m\in\{4,5\}\\
		\left[ (\pm1,\pm1) , \frac{1}{36}\right], & \text{for } m\in\{6,7,8,9\} ,		
	\end{cases}
\end{equation}
and for the D3Q15 case is
\begin{equation}
	\left[\vec{e}_m,w_m\right] \vert_\text{D3Q15} =  
	\begin{cases} 
		\left[ (0,0,0), \frac{2}{9} \right],    & \text{for } m\in\{1\} \\
		\left[ (\pm1,0,0), \frac{1}{9} \right], & \text{for } m\in\{2,3\}\\
		\left[ (0,\pm1,0), \frac{1}{9} \right], & \text{for } m\in\{4,5\}\\
		\left[ (0,0,\pm1), \frac{1}{9} \right], & \text{for } m\in\{6,7\}\\
		\left[ (\pm1,\pm1,\pm1) , \frac{1}{72}\right], & \text{for } m\in\{8,\dots,15\} .
	\end{cases}
\end{equation}
Furthermore, the non-dimensional speed of sound is $c_s = 1/\sqrt{3}$ for each of these three cases.

\subsection{Embedding Velocity Sets} \label{sec:Velocity Sets Embed}
While the velocity sets described in Appendix \ref{sec:Specific Velocity Sets} are popular choices, none have a velocity dimension $Q$ that is power of two, an assumption made in Section \ref{sec:Carl_LBE}. To solve this issue, first embed the velocity set $Q$ into one that is a power of two by defining $Q^{(\text{e})} \coloneq 2^{\lfloor \log Q\rfloor + 1}$, as well as $Q_z=Q^{(\text{e})}-Q$ for convenience. Next, let $\tilde{f}_m(t,\vec{X})$ describe the zero padded velocity distribution function with $Q^{(\text{e})}$ discrete velocities and let $f_m(t,\vec{X})$ describe the original distribution function with $Q$ discrete velocities. We then embed $f_m$ into $\tilde{f}_m$ by setting $\tilde{f}_m(t=0,\vec{X})=0$ for $m>Q$, and 
\begin{equation} \label{eqn:zero pad RGE}
\begin{gathered}
	R = 
	\begin{pmatrix}\begin{array}{ccc|c} 
			\beta_{1,1} & \cdots & \beta_{1,Q} & \multirow{5}{*}{$\hat{R}_1$} \\ 
			\vdots & \ddots & \vdots & \\ 
			\beta_{Q,1} & \cdots & \beta_{Q,Q} & \\
			\cline{1-3}
			\multicolumn{3}{c|}{0_{Q_z \times Q}} &
	\end{array}\end{pmatrix}
	, \quad 
	\Gamma_q = 
	\begin{pmatrix}\begin{array}{ccc|c} 
			\gamma_{q,1,1} & \cdots & \gamma_{q,1,Q} & \multirow{5}{*}{$\hat{\Gamma}_q$} \\ 
			\vdots & \ddots & \vdots & \\ 
			\gamma_{q,Q,1} & \cdots & \gamma_{q,Q,Q} & \\
			\cline{1-3}
			\multicolumn{3}{c|}{0_{Q_z \times Q}} &
	\end{array}\end{pmatrix} ,\\[4pt]
	E_\eta = 
	\begin{pmatrix}\begin{array}{ccc|c} 
			e_1^\eta & & & \multirow{5}{*}{$\hat{E}_\eta$} \\ 
			& \ddots & & \\ 
			& & e_Q^\eta & \\
			\cline{1-3}
			\multicolumn{3}{c|}{0_{Q_z \times Q}} &
	\end{array}\end{pmatrix} ,
\end{gathered}
\end{equation}
where $\eta \in \{x,y,z\}$, $\hat{R}_1,\hat{\Gamma}_q, \hat{E}_\eta \in \mathbb{C}^{Q^{(\text{e})} \times Q_z}$, and $R,\Gamma_q, E_\eta \in \mathbb{C}^{Q^{(\text{e})} \times Q^{(\text{e})}}$. By creating an analogous form of \eqref{eqn:LBE_ExtraCompact} for $\tilde{f}_m$ and inserting \eqref{eqn:zero pad RGE} into it, one can see that $\partial \tilde{f}_m/\partial t=0$ for $m>Q$, i.e. these zero padded elements keep their initial values of $0$ for all time. This means that the submatrices $\hat{E}_\eta$, $\hat{R}_1$ and $\hat{\Gamma}_q$ can take any value (including $\mathbf{0}$). Furthermore, there is flexibility in how to define the $\Gamma_q$ matrices for $q>Q$ with one valid approach leveraging the symmetry $\gamma_{q,m,r} = \gamma_{r,m,q}$. Using this procedure, we can transform any D($n$)Q($m$) lattice into a D($n$)Q($m$)* lattice where the asterisk indicates that the number of speeds have been increased to the nearest power of two ($2^{\lfloor\log m\rfloor + 1}$) using the embedding procedure described above.

\end{document}